\documentclass{article}

\usepackage{graphicx}
\usepackage{amsmath,amssymb,amsthm}
\usepackage{framed}
\usepackage{color}
\usepackage{array}
\usepackage{enumerate}

\usepackage{tikz}
\usetikzlibrary{arrows,automata,backgrounds,calc,chains,decorations.fractals,decorations.markings,decorations.pathreplacing,fadings,fit,folding,patterns,positioning,mindmap,shadows,shapes.geometric,shapes.symbols,through,trees,plotmarks}

\theoremstyle{plain}
\newtheorem{theorem}{Theorem}[section]
\newtheorem{lemma}[theorem]{Lemma}
\newtheorem{proposition}[theorem]{Proposition}
\newtheorem{corollary}[theorem]{Corollary}

\theoremstyle{definition}
\newtheorem{definition}[theorem]{Definition}

\newtheorem{example}[theorem]{Example}
\newtheorem{remark}[theorem]{Remark}

\newcommand{\tuplft}{{\langle}}
\newcommand{\tuprgt}{{\rangle}}
\newcommand{\tupsep}{{,\,}}
\newcommand{\tup}[1]{\tuplft#1\tuprgt}
\newcommand{\tuple}{\tup}
\newcommand{\pair}[2]{\tup{#1\tupsep#2}}

\newcommand{\nb}{\nobreakdash}

\newcommand{\sub}[2]{#1_{#2}}

\newcommand{\tsp}{\hspace{0em}}

\newcommand{\where}{\mid} 
\newcommand{\displset}[2]{\left\{\tsp #1 \tsp \where \tsp #2 \tsp \right\}}
\newcommand{\enumset}[1]{\left\{\tsp #1 \tsp\right\}}

\newcommand{\msf}{\mathsf}
\newcommand{\mbb}{\mathbb}
\newcommand{\mbf}{\mathbf}
\newcommand{\mbs}{\boldsymbol}
\newcommand{\mcl}{\mathcal}
\newcommand{\mrm}{\mathrm}

\newcommand{\nat}{\mbb{N}}

\newcommand{\emp}{\varepsilon}
\renewcommand{\emptyset}{\varnothing}

\newcommand{\dol}{D0L}
\newcommand{\pdol}{P\dol}
\newcommand{\cdol}{C\dol}

\newcommand{\divides}{\mid}
\newcommand{\ndivides}{\nmid} 

\newcommand{\aprog}{F}
\newcommand{\bprog}{G}

\newcommand{\aprg}{\aprog}
\newcommand{\iaprg}{\sub{\aprg}}
\newcommand{\bprg}{\bprog}

\newcommand{\undefd}[1]{{#1}{\uparrow}}
\newcommand{\defd}[1]{{#1}{\downarrow}}

\newcommand{\cbox}[2]{{\setlength{\fboxsep}{1.5pt}\colorbox{#1}{$#2$}}}

\colorlet{corange}{orange!90!red!60}
\colorlet{cpureorange}{orange}
\colorlet{cred}{red}
\colorlet{cdarkred}{red!70!black}
\colorlet{clightred}{red!70}
\colorlet{cpurple}{blue!50}
\definecolor{cblue}{rgb}{0,0.4,0.7}
\definecolor{clightblue}{rgb}{0.3,0.7,1.0}
\colorlet{cgreen}{green!80!black!60}
\colorlet{cdarkgreen}{green!60!black!70}
\colorlet{chighlight}{yellow!80!orange}

\newcommand{\symbfont}{\msf}
\newcommand{\symbs}{\cbox{white}{\symbfont{s}}}
\newcommand{\symba}{\cbox{corange}{\symbfont{a}}}
\newcommand{\nsymba}[1]{\cbox{corange}{\symbfont{a}^{#1}}}
\newcommand{\symbb}{\cbox{cpureorange}{\symbfont{b}}}
\newcommand{\symbi}{\cbox{cdarkgreen}{\symbfont{z}}}
\newcommand{\symbl}{\cbox{cgreen}{\symbfont{l}}}
\newcommand{\symbr}{\cbox{cgreen}{\symbfont{r}}}
\newcommand{\symbe}{\cbox{clightblue}{\symbfont{e}}}

\newcommand{\symbA}{\cbox{corange}{\symbfont{A}}}
\newcommand{\symbB}{\cbox{cpureorange}{\symbfont{B}}}
\newcommand{\symbI}{\cbox{cdarkgreen}{\symbfont{Z}}}
\newcommand{\symbL}{\cbox{cgreen}{\symbfont{L}}}
\newcommand{\symbR}{\cbox{cgreen}{\symbfont{R}}}
\newcommand{\symbE}{\cbox{clightblue}{\symbfont{E}}}

\newcommand{\symbc}{\cbox{clightblue}{\symbfont{c}}}
\newcommand{\symbd}{\cbox{clightblue!50}{\symbfont{d}}}
\newcommand{\symbo}{\cbox{cpurple!70}{\symbfont{o}}}
\newcommand{\symbp}{\cbox{white}{\symbfont{\star}}}

\newcommand{\nsymbpp}[1]{\cbox{white}{\symbfont{\star}^{#1}}}
\newcommand{\szero}{\cbox{chighlight}{\symbfont{0}}}
\newcommand{\sone}{\cbox{chighlight}{\symbfont{1}}}
\newcommand{\symbO}{\cbox{cpurple!70}{\symbfont{O}}}
\newcommand{\symbX}{\cbox{cpurple!70}{\symbfont{P}}}
\newcommand{\symbQ}{\cbox{black!70}{\textcolor{white}{\symbfont{Q}}}}

\newcommand{\Xspa}{\cbox{white}{\symbfont{X}}}

\newcommand{\spa}{\cbox{white}{\text{\textvisiblespace}}}
\newcommand{\nspa}[1]{\cbox{white}{\text{\textvisiblespace}^{#1}}}

\newcommand{\smin}{\cbox{white}{\circ}}
\newcommand{\nsmin}[1]{\cbox{white}{\circ^{#1}}}

\newcommand{\splu}{\cbox{white}{\bullet}}
\newcommand{\nsplu}[1]{\cbox{white}{\bullet^{#1}}}

\newcommand{\xline}[1]{
  \begin{tikzpicture}[nodes={rectangle},inner sep=0,outer sep=0,baseline=.3ex]
    \node (n) {};
    #1
  \end{tikzpicture}
}

\newcommand{\sat}[2]{
  \node (n) [at=(n.south east),anchor=south west] {$#1$};
  \node (a) [at=(n.south),anchor=north,yshift=-1mm] {{\tiny #2}};
}

\newcommand{\pto}{\rightharpoonup}

\newcommand{\sfunin}{{:}}
\newcommand{\funin}{\mathrel{\sfunin}}
\newcommand{\funap}[2]{#1(#2)}
\newcommand{\bfunap}[3]{\funap{#1}{#2,#3}}

\newcommand{\applicable}{\psi}

\newcommand{\anum}{n}
\newcommand{\aden}{d}

\newcommand{\ianum}{\sub{\anum}}
\newcommand{\iaden}{\sub{\aden}}
\newcommand{\iafrac}[1]{\frac{\ianum{#1}}{\iaden{#1}}}
\newcommand{\iafrc}{\sub{f}}%

\newcommand{\astate}{\alpha}
\newcommand{\iastate}{\sub{\astate}}
\newcommand{\areg}{r}
\newcommand{\iareg}{\sub{\areg}}

\newcommand{\mirror}[1]{\reflectbox{$#1$}}

\newcommand{\texcomp}[3]{#1{#2{#3}}}

\newcommand{\sspadic}{v}
\newcommand{\spadic}{\sub{\sspadic}}
\newcommand{\padic}{\texcomp{\funap}{\spadic}}

\newcommand{\aalph}{\Sigma}
\newcommand{\balph}{\Gamma}
\newcommand{\nalph}{\sub{\Sigma}}

\newcommand{\wrd}[1]{#1^{\ast}}
\newcommand{\newrd}[1]{#1^{+}}
\newcommand{\str}[1]{#1^{\omega}}

\newcommand{\wrdlt}{\prec}
\newcommand{\wrdle}{\preceq}

\renewcommand{\pmod}[1]{~~(\mrm{mod}~#1)}

\newcommand{\kola}{K}

\newcommand{\lt}{<}
\newcommand{\gt}{>}

\newcommand{\wrdemp}{\emp}

\newcommand{\zrep}[1]{(#1)_{\msf{z}}}
\newcommand{\zval}[1]{[#1]_{\msf{z}}}

\newcommand{\BIN}{\mrm{BIN}}

\newcommand{\col}{:~}
\newcommand{\sep}{\,,\;}%
\newcommand{\st}{\iastate}
\newcommand{\sh}{\sub{\mirror{\astate}}}%
\newcommand{\reg}{\iareg}
\newcommand{\out}{\sub{\beta}}%

\newcommand{\smetric}{d}
\newcommand{\metric}{\bfunap{\smetric}}

\newcommand{\cpi}[2]{\mathrm{\Pi}^{#1}_{#2}}
\newcommand{\csig}[2]{\mathrm{\Sigma}^{#1}_{#2}}

\newcommand{\length}[1]{|#1|}

\newcommand{\gap}{?}

\newcommand{\sssubwcompl}{p}
\newcommand{\ssubwcompl}{\sub{\sssubwcompl}}
\newcommand{\subwcompl}[1]{\funap{\ssubwcompl{#1}}}

\newcommand{\iobit}[1]{\cbox{chighlight}{\kappa_{#1}}}
\newcommand{\obit}{\funap{\kappa}}

\newcommand{\simpstep}{\rightsquigarrow}
\newcommand{\simp}{\sim}

\newcommand{\seq}{\mbs}

\newcommand{\problem}[2]{{\setlength{\leftmargini}{47pt}
  \vspace{-1ex}
  \begin{itemize}
    \item [{\sc Input:}] #1
    \item [{\sc Question:}] #2
  \end{itemize}\vspace{-1ex}}}

\begin{document}
\allowdisplaybreaks

\title{On Periodically Iterated Morphisms%
  \thanks{%
    This research has been funded by 
    the Netherlands Organization for Scientific Research (NWO)
    under grant numbers 639.021.020 and 612.000.934.%
  }%
}

\author{%
  J\"{o}rg Endrullis \\
  VU University Amsterdam \\
  Department of Computer Science \\
  {\normalfont{\texttt{j.endrullis@vu.nl}}}
  \and
  Dimitri Hendriks \\
  VU University Amsterdam \\
  Department of Computer Science \\
  {\normalfont{\texttt{r.d.a.hendriks@vu.nl}}}
}

\date{}

\maketitle

\begin{abstract}
  
We investigate the computational power of 
periodically iterated morphisms,
also known as \dol{} systems with periodic control, 
\pdol{} systems for short.
These systems give rise to a class of one-sided infinite sequences, 
called \pdol{} words.

We construct a \pdol{} word with exponential subword complexity,
thereby answering a question raised by Lepist\"o~\cite{lepi:1993}
on the existence of such words. 
We solve another open problem concerning the decidability of the first-order theories 
of \pdol{} words~\cite{much:prit:seme:2009}; 
we show it is already undecidable whether
a certain letter occurs in a \pdol{} word.
This stands in sharp contrast to the situation for 
\dol{} words (purely morphic words), 
which are known to have at most quadratic subword complexity,
and for which the monadic theory is decidable.

The main result of our paper, leading to these answers, 
is that every computable word~$\seq{w} \in \str{\aalph}$ 
can be embedded in a \pdol{} word~$\seq{u}\in\str{\balph}$
with $\balph\supset\aalph$ in the following two ways:
(i)~such that every finite prefix of $\seq{w}$ is a subword of~$\seq{u}$,
and 
(ii)~such that $\seq{w}$ is obtained from $\seq{u}$ by erasing all letters 
from $\balph\setminus\aalph$.
The \pdol{} system generating such a word $\seq{u}$
is constructed by encoding a Fractran program that computes the word $\seq{w}$;
Fractran is a programming language as powerful as Turing Machines.

As a consequence of~(ii), 
if we allow the application of finite state transducers to \pdol{} words,
we obtain the set of all computable words.
Thus the set of \pdol{} words is not closed under finite state transduction,
whereas the set of \dol{} words is.
It moreover follows that equality of \pdol{} words (given by their \pdol{} system)
is undecidable.
Finally, we show that if erasing morphisms are admitted,
then the question of productivity becomes undecidable,
that is, the question whether a given \pdol{} system defines an infinite word.

\end{abstract}%

\setcounter{section}{-1}

\section{Introduction}\label{sec:intro}
Morphisms for transforming and generating infinite words 
provide a fundamental tool for formal languages,
and have been studied extensively; 
we refer to~\cite{allo:shal:2003} and the bibliography therein.

In this paper we investigate the class of infinite words 
generated by periodically alternating morphisms~\cite{culi:karh:1992,culi:karh:lepi:1992,lepi:1993,cass:karh:1997}.
Instead of repeatedly applying a single morphism,
one alternates 
several morphisms from a given (finite) set in a periodic fashion.
Let us look at an example right away,
and consider the most famous word generated by such a procedure,
namely the Kolakoski word~\cite{kola:1965}
\begin{align*}
  \kola = 1 \, 22 \, 11 \, 2 \, 1 \, 22 \, 1 \, 22 \, 11 \, 2 \, 11 \, 22 \, 1 \, 2 \, 11 \, 2 \, 1 \, 22 \, 11 \, 2 \, \cdots
%
\end{align*}
which is defined such that 
$\kola(0) = 1$ and $\kola(n)$ equals the length of the $n$-th run of $\kola$; 
here by a `run' we mean a block of consecutive identical symbols.
The Kolakoski word can be generated by alternating two morphisms on the starting word $12$, 
$h_0$ for the even positions and $h_1$ for the odd positions, defined as follows:
\[
\begin{aligned}
  h_0 : 
  \begin{array}{l}
    1 \to 1 \\
    2 \to 11
  \end{array}
  &&
  &&
  h_1 : 
  \begin{array}{l}
    1 \to 2 \\
    2 \to 22
  \end{array}
\end{aligned}
\]
The first few iterations then are
\vspace{-1ex}
\begin{align*}
  & \mathrel{\phantom{=}} 1 2 \\
  h_0(1) \, h_1(2) & = 1 2\mbf{2} \\
  h_0(1) \, h_1(2) \, h_0(2) & = 1 22 \mbf{11} \\
  h_0(1) \, h_1(2) \, h_0(2) \, h_1(1) \, h_0(1) & = 1 22 11 \mbf{21} 
\end{align*}
It is known that the Kolakoski word is not purely morphic~\cite{culi:karh:lepi:1992}, 
i.e, cannot be generated by iterating a single morphism. 
However it is an open problem whether it is a morphic word, 
i.e., the image of a purely morphic word under a coding (= letter-to-letter morphism).
We shall use the `\dol' terminology: \dol{} for purely morphic, \cdol{} for morphic,
and \pdol{} for words generated by periodically alternating morphisms, like the Kolakoski word above.

A natural characteristic of sequences is their subword complexity~\cite{fere:1999,allo:1994,allo:shal:2003}.
The subword complexity of a sequence $\seq{u}$ is a function $\nat \to \nat$
mapping $n$ to the number of $n$\nb-length words that occur in~$\seq{u}$.
It is well-known that morphic words have at most quadratic subword complexity~\cite{ehre:lee:roze:1975}.
Lepist\"{o}~\cite{lepi:1993} proves that for all $r\in\mbb{R}$ there is a \pdol{} word 
whose subword complexity is in $\Omega(n^r)$; hence there are \pdol{} words that are not \cdol{}.
It remained an open problem whether \pdol{} words can exhibit exponential subword complexity.
This intriguing question formed the initial motivation for our investigations.
We actually establish a stronger result from 
which the existence of such words can be derived, 
as we will describe next.

The main results of our paper can be stated as follows:
\begin{itshape}%
  For every computable word~$\seq{w}\in\str{\aalph}$
  there exists a \pdol{} word~$\seq{u}$ such that
  \begin{enumerate}[I.]
    \item 
      all prefixes of $\seq{w}$ occur in $\seq{u}$
      as subwords between special marker symbols (Theorem~\ref{thm:comp:prefix}), 
    \item 
      $\seq{w}$ is the subsequence of~$\seq{u}$
      obtained from selecting all letters from $\aalph$ (Theorem~\ref{thm:comp:sparse}). 
  \end{enumerate}
\end{itshape}%
The construction of the \pdol{} systems generating such words $\seq{u}$ 
makes use of Fractran~\cite{conw:1972,conw:1987}, 
a Turing complete programming language invented by Conway,
in the following way.
First, in Section~\ref{sec:fractran}, we show how to employ Fractran to generate any computable infinite word.
Then we encode Fractran programs as \pdol{} systems, 
and prove that the \pdol{} system correctly simulates the Fractran program
and records its output, see Sections~\ref{sec:prod} and~\ref{sec:comp}.

Consequences of 
I and II are as follows:
\begin{enumerate}[(1)]
  \item 
    There exist \pdol{} words with exponential subword complexity (Theorem \ref{thm:exp}).
  \item 
    It is undecidable to determine, given a \pdol{} system~$\mcl{H}$ and a letter~$b$,  
    whether the letter $b$ occurs (infinitely often) in the word generated by $\mcl{H}$ (Theorem~\ref{thm:letter}).
  \item 
    The first-order theory of \pdol{} words is undecidable (Corollary~\ref{cor:logic:undecidable}).
  \item 
    Equality of \pdol{} words is undecidable (Corollary~\ref{cor:equality:undecidable}).
  \item 
    The set of \pdol{} words is not closed under finite state transductions
    (Corollary~\ref{cor:fst}).
\end{enumerate}
All the above results concern \pdol{} systems 
whose morphisms are non-erasing.
But we also study erasing \pdol{} systems, and find
\begin{enumerate}[(1)]
  \setcounter{enumi}{5}
  \item 
    It is undecidable to determine, on the input of an erasing \pdol{} system,
    whether it generates an infinite word (Theorem~\ref{thm:prod}).
\end{enumerate}

The outline of the paper is as follows.
In Sections~\ref{sec:pdol} and~\ref{sec:fractran}
we introduce the \emph{dramatis personae} of our story:
\pdol{} systems 
and Fractran programs. 
We explain the workings of the Fractran algorithm,
and how to program in this language.

Then, as a steppingstone to our main result, we start with a proof of~(6) in Section~\ref{sec:prod}.
This proof illustrates our key construction: encoding Fractran programs as \pdol{} systems.
We then modify and extend this encoding in Section~\ref{sec:comp} 
to prove Theorems~\ref{thm:comp:prefix} and~\ref{thm:comp:sparse}: 
\pdol{} words can embed  every computable word, in the sense of I~and~II above.
We give a detailed example of the translation, and prove (1)--(5) listed above.

\pdol{} systems resulting from encoding Fractran programs can be quite large.
For example, the system obtained from a simple binary counter 
(computing an infinite word with exponential subword complexity)
consists of  
\[ 536393214598471230 \] 
morphisms.
We present a direct solution in Section~\ref{sec:exp}, 
namely a \pdol{} system with $16$ morphisms
simulating such a counter.

\section[D0L Systems with Periodic Control]{\dol{} Systems with Periodic Control}\label{sec:pdol}

We use standard terminology and notations, see, e.g., \cite{allo:shal:2003}.
Let $\aalph$ be a finite alphabet. 
We denote by $\wrd{\aalph}$ the set of all finite words over~$\aalph$,
by $\wrdemp$ the empty word,
and by $\newrd{\aalph} = \wrd{\aalph} \setminus \enumset{\wrdemp}$ 
the set of finite non-empty words.

The set of infinite words over $\aalph$ is 
$\str{\aalph} = \displset{\seq{x}}{\seq{x} \funin \nat \to \aalph}$.
On the set of all words $\aalph^{\infty} = \wrd{\aalph} \cup \str{\aalph}$
we define the metric $\smetric$ for all $\seq{u},\seq{v} \in \aalph^{\infty}$ by
$\metric{\seq{u}}{\seq{v}} = 2^{-n}$, where $n$ is the length of the longest common prefix 
of $\seq{u}$ and $\seq{v}$.

We let $\nalph{p} = \enumset{0,\ldots,p-1}$.
We write~$\length{x}$ for the length of $x\in\aalph^\infty$,
with $\length{x} = \infty$ if $x$ is infinite.
We call a word $v \in\wrd{\aalph}$ a \emph{factor} of $x\in\aalph^\infty$
if $x = uvy$ for some $u\in \wrd{\aalph}$ and $y\in\aalph^\infty$,
and say that $v$ occurs at position~$\length{u}$.
For words $u,v\in\wrd{\aalph}$, we write $u \wrdlt v$ if $u$ is a strict prefix of $v$, 
i.e., if $v = u u'$ for some $u' \in \newrd{\aalph}$,
and use\hspace{.5pt} ${\wrdle}$\, for its reflexive closure.

A \emph{morphism} is a map $h \funin \wrd{\aalph} \to \wrd{\balph}$ 
such that $h(uv) = h(u)h(v)$ for all $u,v\in\wrd{\aalph}$,
and can thus be defined by giving its values on the symbols of $\aalph$.
A morphism~$h$ is called \emph{erasing} if $h(a) = \wrdemp$ for some $a\in\aalph$,
and \emph{$k$\nb-uniform}, with $k \in \nat$,
if $\length{h(a)} = k$ for all $a\in\aalph$;
$h$ is a \emph{coding} if it is $1$\nb-uniform.

Infinite sequences generated by periodically alternating morphisms, 
also called `\dol{} words with periodic control' or just `\pdol{} words' for short, 
were introduced in~\cite{culi:karh:1992}.
These form a generalization of \dol{} words, also known as \emph{purely morphic} words,
which are obtained by iterating a single morphism. 
\begin{definition}\label{def:pdol}
  Let $H = \tuple{h_0,\ldots,h_{p-1}}$ be a tuple 
  of morphisms $h_i \funin \wrd{\aalph} \to \wrd{\aalph}$.
  We define the map $H \funin \wrd{\aalph} \to \wrd{\aalph}$ as follows:
  \begin{align*}
    H(a_0 & a_1 \cdots a_n) = u_0 u_1 \cdots u_n \\ 
    & \text{where $u_i = h_k(a_i)$, with $k \equiv i \pmod{p}$ and $k \in \nalph{p}$.}
  \end{align*}
  If $s \in \wrd{\aalph}$ is such that $s \wrdle H(s)$,
  then the triple $\mcl{H} = \tuple{\aalph,H,s}$
  is called a \emph{\pdol{} system}.
  Then in the metric space $\pair{\aalph^\infty\!}{\smetric}$ the limit 
  \begin{align*}
    H^\omega(s) = \lim_{i\to\infty} H^i(s)
  \end{align*}
  exists, and we call $H^\omega(s)$ the \emph{\pdol{} word} generated by $\mcl{H}$.
  We say that $\mcl{H}$ is \emph{productive} if $H^\omega(s)$ is infinite,
  and $\mcl{H}$ is \emph{erasing} if some of its morphisms~$h_i$ are erasing.

  If $x$ is a \pdol{} word generated by $p$ morphisms,
  and $x = uvy$ for some $u,v\in\wrd{\aalph}$ and $y\in\aalph^\infty$,
  we say that the factor~$v$ of $x$ occurs at \emph{morphism index~$i$}
  when $i\in\nalph{p}$ and $i \equiv \length{u} \pmod{p}$.
\end{definition}

\dol{} words are generated by \dol{} systems~$\tup{\aalph,h,s}$, 
i.e., \pdol{} systems $\tup{\aalph,\tup{h},s}$ consisting of one single morphism~$h$.
Following~\cite{cass:karh:1997}, we call the image of a \dol{} word under a coding 
(a letter-to-letter morphism), a \emph{\cdol{}} word, better known as morphic words.

In the literature, one typically requires the morphisms $h_i$ to be non-erasing
to ensure that the limit is infinite. 
We have taken a more general definition of \pdol{}-words,
since also erasing morphisms may yield an infinite word in the limit.
See Remark~\ref{rem:pdol:prod} below.

In the sequel it will be helpful to have a recursive definition of the map~$H$.
\begin{lemma}\label{lem:H:rec}
  Let $H = \tup{h_0,h_1,\ldots,h_{p-1}}$ be a tuple of morphisms. 
  For $i \in \nalph{p}$ define $H_i = \tup{h_i,\ldots,h_{p-1},h_0,\ldots,h_{i-1}}$
  and the corresponding map $H_i \funin \wrd{\aalph} \to \wrd{\aalph}$
  by 
  \begin{align*}
    H_i(\wrdemp) & = \wrdemp \\
    H_i(au) & = h_i(a) H_{i+1}(u) & \text{($a\in\aalph,\,u\in\wrd{\aalph}$)}
  \end{align*}
  where addition in the subscript of $H$ is taken modulo~$p$.

  Then $H_0 = H$ with $H$ the map defined in Definition~\ref{def:pdol},
  and $H_i(uv) = H_i(u) H_{i+\length{u}}(v)$ for all $u,v\in\wrd{\aalph}$ and $i\in\nalph{p}$.
  \qed
\end{lemma}

Using this notation we now formulate the \pdol{} analogue 
of the usual condition for productivity of \dol{} systems.
In Section~\ref{sec:prod} we show that productivity of \pdol{} systems in general is undecidable.
Productivity has been studied in the wider perspective of 
term rewriting systems in~\cite{sijt:1989,endr:grab:hend:2008,endr:grab:hend:isih:klop:2010,endr:hend:2011}.

\begin{remark}\label{rem:pdol:prod}
  Let $\tuple{\aalph,h,s}$  be a \dol{} system.
  We say that $h$ is \emph{prolongable} on $s$
  if $h(s) = sx$ for some $x\in\wrd{\aalph}$
  and $h^i(x) \neq \wrdemp$ for all $i \ge 0$.
  Then $h^i(s) \wrdlt h^{i+1}(s)$ for all $i \ge 0$,
  and hence the limit $h^\omega(s) = s\, x\, h(x)\, h^2(x) \cdots$ is infinite.
  The generalization of this condition 
  to \pdol{} systems $\mcl{H} = \tuple{\aalph,H,v_0}$ 
  is: 
  (*)~$H(v_0) = v_0 v_1$ for some $v_1 \in \wrd{\aalph}$
  such that $v_n \neq \wrdemp$ for all $n \in \nat$,
  where $v_n \in \wrd{\aalph}$ and $z_n\in\nalph{p}$ 
  are defined by $z_0 = 0$ and
  \begin{align*}
    v_n & = H_{z_{n-1}}(v_{n-1}) && (n \ge 2)
    \\ 
    z_n & \equiv z_{n-1} + \length{v_{n-1}} \pmod{p} && (n \ge 1)
  \end{align*} 
  Then $H^n(v_0) = H^{n-1}(v_0) v_n$ for all $n \ge 1$, and so (*)
  forms a necessary and sufficient condition for productivity of~$\mcl{H}$, that is,
  for the limit $H^\omega(v_0) = v_0 v_1 v_2 \cdots$ to be infinite.
\end{remark}

\begin{definition}\label{def:subword:complexity}
  The \emph{subword complexity} of an infinite word $\seq{x}\in\str{\aalph}$ 
  is the function~$\ssubwcompl{\seq{x}} \funin \nat \to \nat$ such that
  $\subwcompl{x}{n}$ is the number of factors (subwords) of $\seq{x}$ of length $n$.
\end{definition}

\begin{proposition}[\cite{ehre:lee:roze:1975}]\label{prop:dol:quad}
  The subword complexity of \dol{} words, and hence of \cdol{} words, 
  is at most quadratic. 
\end{proposition}

We first consider an example of an erasing \pdol{} system.
\begin{example}\label{ex:erasin:pdol}
  Let $\mcl{H} = \tup{\nalph{3},H,0}$ with $H = \tup{h_0,h_1,h_2}$
  defined for all $b\in\nalph{3}$ as follows, where addition runs modulo~$3$:
  \begin{align*}
    h_0(b) & = b (b+1) (b+2) &&& 
    h_1(b) & = \wrdemp &&&
    h_2(b) & = b + 2 
  \end{align*}
  Then $\mcl{H}$ is productive (by Proposition~\ref{prop:prod:loc:cat})
  and generates the word
  \begin{align*}
    H^\omega(0) = 0121120101221201120212010120201001210 \cdots 
  \end{align*}
\end{example}

\begin{definition}\label{def:pdol:uniform}
  Let $\mcl{H} = \tup{\aalph,\tup{h_0,\ldots,h_{p-1}},s}$
  be a \pdol{} system.
  We say $\mcl{H}$ is \emph{locally uniform} if 
  every morphism $h_i$ is uniform, i.e,
  if for all $i \in \nalph{p}$ there is $k_i\in\nat$ 
  such that $k_i = \length{h_i(b)}$ for all $b\in\aalph$.
  We say $\mcl{H}$ is \emph{(globally) uniform} 
  if, for some $k \in \nat$, each $h_i$ is $k$-uniform ($i\in\nalph{p}$).
\end{definition}
Obviously, a globally $k$-uniform \pdol{} system is productive if and only if $k \ge 2$.
For locally uniform systems the condition is formulated as follows,
and is easy to check.
\begin{proposition}\label{prop:prod:loc:cat}
  Let\ $\mcl{H} = \tup{\aalph,\tup{h_0,\ldots,h_{p-1}},w}$ be a locally uniform \pdol{} system,
  where $h_i$ is $k_i$\nb-uniform. 
  Let $s(n)$ be defined by
  $s(0) = 0$ and $s(n+1) = s(n) + k_i$ with 
  $i \equiv n \pmod{p}$.
  Then $\mcl{H}$ is productive if and only if
  $s(n) \gt n$ for all $n \ge \length{w}$.
\end{proposition}
\begin{proof}
  The word $H^\omega(w)$ can be defined as the limit 
  of the 
  sequence 
  \[ w_{\length{w}}, w_{\length{w}+1}, w_{\length{w}+2},\ldots \] 
  of finite words
  defined for $n \ge \length{w}$ by
  \begin{align*}
    w_{\length{w}} &= H(w) \\
    w_{n+1} &= 
    \begin{cases}
      w_{n} &\text{if $n \ge \length{w_n}$}\\
      w_{n}\,h_i(w_{n}(n)) &\text{if $n < \length{w_n}$ and $n \equiv i \pmod{p}$}
    \end{cases}
  \end{align*}
  We have $\length{w_{\length{w}}} = s(\length{w})$
  and by induction we get
  $\length{w_n} = s(n)$ for every $n \ge \length{w}$.
  The limit $\lim_{n\to\infty} w_n$ is infinite if and only if
  we never get to the clause $n \ge \length{w_n}$, 
  which holds in turn if and only if $s(n) > n$ for all $n \ge \length{w}$.
\end{proof}

\begin{example}\label{ex:arshon}%
  Let $\aalph = \enumset{0,1,2}$, $H = \tuple{h_0,h_1}$ 
  with morphisms $h_0,h_1 \funin \wrd{\aalph} \to \wrd{\aalph}$ 
  for all $a\in\aalph$ defined by
  \begin{align*}
    h_0(a) & = a (a+1) (a+2)
    &&&
    h_1(a) & = (a+2) (a+1) a
  \end{align*}
  (with addition modulo~$3$).
  Then the \pdol{} system $\tuple{\aalph,H,0}$ generates the word
  \begin{align*}
        H^\omega(0) 
        = 0 1 2 0 
          2 1 2 0 
          1 2 1 0 
          2 0 1 0 
          2 1 2 0 
          1 2 1 0 
          1 2 0 1 
          0 2 1 2 
          0 \cdots
  \end{align*}
  This is the square-free Arshon word~\cite{arsh:1937} (of rank~$3)$,
  which Berstel proved to be an example of a \cdol{} word 
  that is not a \dol{} word~\cite{bers:1980};
  see S\'e\'ebold~\cite{seeb:2003} for a generalization.
  That $H^\omega(\msf{L})$ can indeed be defined
  as a \cdol{} word follows from Proposition~\ref{prop:pdol:uniform}.
\end{example}

It is not hard to see that, when a word~$\seq{u}$ is 
generated by a (globally) $k$\nb-uniform \pdol{} system,
it is $k$\nb-automatic~\cite{allo:shal:2003},
i.e., $\seq{u}$ 
is the image of a coding 
of the fixed point of a $k$\nb-uniform morphism. 
\begin{proposition}\label{prop:pdol:uniform}
  Let $k \ge 2$, and $\mcl{H} = \tup{\aalph,H,s}$ a $k$\nb-uniform \pdol{} system.
  Then $H^\omega(s)$ is $k$\nb-automatic. 
\end{proposition}
\begin{proof}
  Let $H = \tup{h_0,\ldots,h_{p-1}}$, 
  where every~$h_i$ is $k$\nb-uniform.
  We define the ($k$\nb-uniform) 
  morphism $g \funin \nalph{p} \times \aalph \to \nalph{p} \times \aalph$\, by 
  \begin{align*}
    g(\pair{i}{a}) & = \pair{ki}{b_0} \pair{ki+1}{b_1} \cdots \pair{ki+k-1}{b_{k-1}}
  \end{align*}
  where addition in the first entries runs modulo~$p$, 
  and for $j \in \nalph{k}$, $b_j \in \aalph$ is such that 
  $h_i(a) = b_0 b_1 \cdots b_{k-1}$. 
  Let $s = s_0 s_1 \cdots s_{q-1}$, 
  $t = \pair{0}{s_0}\pair{1}{s_1}\cdots\pair{q-1}{s_{q-1}}$,
  and $\seq{u} = H^\omega(s)$.
  Then 
  \begin{align*}
    g^n(t) = \pair{0}{\seq{u}(0)} \pair{1}{\seq{u}(1)} \cdots \pair{qk^n-1}{\seq{u}(qk^n-1)}
  \end{align*}
  follows by induction on $n$. Hence $\tau(g^\omega(t)) = \seq{u}$
  with $\tau$ the coding defined by $\tau(\pair{i}{a}) = a$.
\end{proof}

One might wonder whether also locally uniform, productive \pdol{} systems
always generate morphic words. Examples~\ref{ex:lepisto} and \ref{ex:toeplitz}
show that this is not the case.

\begin{example}[\cite{lepi:1993}]\label{ex:lepisto}
  Define the word $F_p \in \str{\{0,1\}}$ for every $p \ge 2$
  by $F_p = H^\omega(0)$ where 
  $\tup{\{0,1\},H,0}$
  with $H = \tup{h_0,\ldots,h_{p-1}}$ is a \pdol{} system,
  and $h_i$ are morphisms defined by
  \begin{align*}
    h_0 : 
    \begin{cases}
      0 \to 01 \\
      1 \to 00
    \end{cases}
    &&
    h_i :
    \begin{cases}
      0 \to 1 \\
      1 \to 0
    \end{cases}
    & \text{for $i \in \{1,\ldots,p-1\}$}
  \end{align*}
  For example, the word $F_3$ starts like this: 
  \begin{align*}
      01 0 1 00 1 1 00 0 1 01 1 0 01 0 0 01 1 0 01 1 1 00 0 1 01 0 0 00 1 1 01 0 1 00 1 1 \cdots
  \end{align*}
  Lepist\"{o}~\cite{lepi:1993} proves that $F_p$ has more than quadratic subword complexity,
  for every $p \ge 2$.
  Hence, with Proposition~\ref{prop:dol:quad}, 
  these \pdol{} words~$F_p$ cannot be \cdol{} words.
  We note that, conversely, the existence of \cdol{} words 
  that are not \pdol{} words was shown in~\cite{culi:karh:1992}.
\end{example}

\begin{example}[\cite{cass:karh:1997}]\label{ex:toeplitz}
  A Toeplitz word~\cite{jaco:kean:1969} over an alphabet $\aalph$ 
  is generated by a seed word $u \in \aalph\wrd{(\aalph\cup\enumset{\gap})}$ 
  with $\gap\not\in\aalph$,
  as follows.
  Start with the periodic $u^\omega$ 
  and then replace its subsequence of~$\gap$'s by the sequence itself. 
  For example $u = 12\gap\gap\gap$ generates the infinite word 
  $T(u) = 12 121 12 211 12 221 \cdots$. 
  Cassaigne and Karhum\"{a}ki~\cite{cass:karh:1997} show that all 
  Toeplitz words are \pdol{} words; 
  e.g., $T(u) = H^\omega(1)$ 
  where $H = \tuple{h_0,h_1,h_2}$ and
  $h_0(a) = 12a$ and $h_1(a) = h_2(a) = a$ for all $a\in\enumset{1,2}$.
  Moreover, from~\cite[Theorem~5]{cass:karh:1997} 
  it follows that $\subwcompl{T(u)}{n} \in \Theta(n^r)$ 
  with $r = \frac{\log 5}{{\log 5} - {\log 3}} \simeq 3.15066$, 
  thus forming an alternative proof of what was established in~\cite{lepi:1993}:
  there are \pdol{} words that are not \cdol.
\end{example}

\section{Fractran for Computing Streams}\label{sec:fractran}
Fractran~\cite{conw:1972,conw:1987} is a universal 
programming language 
invented by John Horton Conway.
The simplicity of its execution algorithm,
based on the unique prime factorization of integers,
makes Fractran ideal for coding it into other formalisms.

A Fractran program~$\aprg$ is a finite list of fractions 
\begin{align}
  \aprg = \iafrac{1},\ldots,\iafrac{k} 
  \label{eq:aprg}
\end{align}
with $\ianum{i},\iaden{i}$ positive integers.
Let $\iafrc{i} = \iafrac{i}$.
The action of $\aprg$ on an input integer~$N \ge 1$ is to
multiply $N$ by the first `applicable' fraction $\iafrc{i}$, 
that is, the fraction $\iafrc{i}$ with $i$ the least index 
such that the product $N' = N \cdot \iafrc{i}$ is an integer again, 
and then to continue with $N'$.
The program halts if there is no applicable fraction for the current integer $N$.

For example, consider the program 
\begin{align*}
  \aprg = \frac{5}{2\cdot 3} ,\, \frac{1}{2} ,\, \frac{1}{3} 
\end{align*}
and the run of $\aprg$ on input $N = 2^3 3^5$\,:
\begin{align*}
  2^3 3^5 
  & \to 2^2   3^4   5^1 
    \to 2^1   3^3   5^2
    \to 2^0   3^2   5^3 
    \to 2^0   3^1   5^3
    \to 2^0   3^0   5^3
    \,.
\end{align*}
Note that each multiplication by $\frac{5}{6}$ 
decrements the exponents of $2$ and $3$ while
incrementing the exponent of $5$.
Once $\frac{5}{6}$ is no longer applicable,
i.e., when one of the exponents of $2$ and $3$ 
in the prime factorization of the current integer $N$ equals $0$,
the other is set to~$0$ as well.
Hence, executing $\aprg$ on $N = 2^a\, 3^b$ halts after $\max(a,b)$ steps with $5^{\min(a,b)}$.

Thus the prime numbers that occur as factors in the numerators and denominators of a Fractran program 
can be regarded as registers, and if the current working integer is $N = 2^a \, 3^b \, 5^c \, \ldots$
we can say that register~2 holds $a$, register~3 holds $b$, and so on.

The real power of Fractran, however, comes from the use of prime exponents as \emph{states}. 
To explain this, we temporarily let programs consist of multiple lines of the form
\begin{align} 
  \astate \col \iafrac{1} \to \iastate{1} ,\, \iafrac{2} \to \iastate{2} ,\, \ldots ,\, \iafrac{m} \to \iastate{m} 
  \label{eq:fractran:program:line}
\end{align}
forming the instructions for the program in state $\astate$:
multiply $N$ with the first applicable fraction $\iafrac{i}$
and proceed in state~$\iastate{i}$,
or terminate if no fraction is applicable. 
We call the states $\alpha_1,\ldots,\iastate{m}$ in~\eqref{eq:fractran:program:line}
the \emph{successors} of $\astate$, and we say a state is \emph{looping} if it is its own successor.

For example, the program~$P_{\msf{add}}$ given~by the lines
\begin{align*}
  \alpha & : \frac{2\cdot 5}{3} \to \alpha ,\, \frac{1}{1} \to \beta
  &
  \text{and}
  && 
  \beta & : \frac{3}{5} \to \beta
\end{align*}
realizes addition; running $P_{\msf{add}}$ in state $\astate$ on $N = 2^a 3^b$ 
terminates in state $\beta$ with $2^{a+b} 3^b$.

A program with $n$ lines is called a \emph{Fractran-$n$ program}.
A flat list of fractions $\iafrc{1}, \ldots , \iafrc{k}$ 
now is a shorthand for the Fractran\nb-{1} program~%
$\astate : \iafrc{1} \to \astate ,\, \iafrc{2} \to \astate ,\, \ldots ,\, \iafrc{k} \to \astate$. 
Conway~\cite{conw:1987} explains how every Fractran\nb-$n$ program ($n \ge 2$) 
can be compiled into a Fractran\nb-$1$ program,
using the following steps:
\begin{enumerate}
  \item\label{item:compile:1}
    For every looping state~$\astate$, 
    introduce a `mirror' state $\mirror{\astate}$, 
    substitute $\mirror{\astate}$ for all 
    occurrences of $\astate$ in the right-hand sides of its program line, 
    and add the line
    \vspace*{-1ex}
    \begin{align*}
      \quad \quad \mirror{\astate} & : \frac{1}{1} \to \astate
    \end{align*}

  \item\label{item:compile:2}
    Replace state identifiers $\astate$ by `fresh' prime numbers.

  \item\label{item:compile:3}
    For every line of the form~\eqref{eq:fractran:program:line}
    append the following fractions:
    \begin{align*}
      \frac{\ianum{1} \cdot \iastate{1}}{\iaden{1} \cdot \astate} \sep 
      \frac{\ianum{2} \cdot \iastate{2}}{\iaden{2} \cdot \astate} \sep 
      \ldots \sep
      \frac{\ianum{k} \cdot \iastate{m}}{\iaden{m} \cdot \astate} 
    \end{align*}
    (preserving the order) to the list of fractions constructed so far.
\end{enumerate}
Let us illustrate these steps on the adder~$P_{\msf{add}}$ given above.
Step~\ref{item:compile:1} 
of splitting loops, results in
\begin{align*}
  \alpha & : \frac{2\cdot 5}{3} \to \mirror{\alpha} ,\, \frac{1}{1} \to \beta &&&
  \beta  & : \frac{3}{5} \to \mirror{\beta} 
  \\
  \mirror{\alpha} & : \frac{1}{1} \to \alpha &&&
  \mirror{\beta} & : \frac{1}{1} \to \beta\,.
\end{align*}
In step~\ref{item:compile:2}, we introduce `fresh' primes to serve as state indicators,
for example, $\tuple{\alpha,\mirror{\alpha},\beta,\mirror{\beta}} = \tuple{7,11,13,17}$.
Finally, step~\ref{item:compile:3},
we replace lines by fractions, 
to obtain the Fractran\nb-1 program
\begin{align*}
  \iaprg{\msf{add}} =
  \frac{2\cdot 5\cdot \mirror{\alpha}}{3\cdot\alpha} \sep 
  \frac{\alpha}{\mirror{\alpha}} \sep
  \frac{\beta}{\alpha} \sep
  \frac{3\cdot\mirror{\beta}}{5\cdot\beta} \sep
  \frac{\beta}{\mirror{\beta}}
  \,.
\end{align*}
Then indeed the run of $\iaprg{\msf{add}}$ on $2^a 3^b \alpha$ ends in $2^{a+b} 3^b \beta$.

For `sensible' programs any state indicator 
has value $0$ (`off') or $1$ (`on'),
and the program is always in exactly one state at a time.
Hence, if a program~$\aprg$ uses primes~$r_1,\dots,r_p$ 
for storage, and primes $\iastate{1},\ldots,\iastate{q}$ for control,
at any instant the entire \emph{configuration} of~$\aprg$ 
(= register contents + state)
is uniquely represented by the current working integer~$N$
\begin{align*}
  N = r_1^{e_1} \, r_2^{e_2} \, \cdots \, r_p^{e_p} \, \iastate{j}
\end{align*}
for some integers $e_1,\ldots,e_p \ge 0$ and $1 \le j \le q$.

The reason to employ two state indicators $\astate$ and $\mirror{\astate}$
to break self\nb-loops in step~\ref{item:compile:1}, 
is that each state indicator is consumed whenever it is tested, 
and so we need a secondary indicator $\mirror{\astate}$ to say 
``continue in the current state''. 
This secondary indicator $\mirror{\astate}$
is swapped back to the primary indicator $\astate$ in the next instruction, 
and the loop continues.


We now introduce some further notation.
For partial functions $g \funin A \pto B$ we write $\defd{g(x)}$ to indicate that $g$ 
is defined on $x \in A$, and $\undefd{g(x)}$ otherwise.
\begin{definition}\label{def:fractran:function}
  Let $\aprg = \iafrc{1},\ldots,\iafrc{k}$ be a Fractran program with 
  $\iafrc{i} = \iafrac{i} \in \mbb{Q}_{\gt 0}$.
  We define the partial function
  $\applicable_\aprg : \nat \pto \nat$
  which, given an integer $N \ge 1$, 
  selects the index of the first fraction 
  applicable to $N$,
  and is undefined if no such fraction exists, i.e.,
  \begin{align*}
    \applicable_\aprg(N) = \min\, \displset{i}{1 \le i \le k, \; N \cdot \iafrc{i} \in \nat}\,,
  \end{align*}
  where we stipulate $\undefd{(\min \emptyset)}$.
  We write $\applicable(n)$ for short
  when $\aprg$ is clear from the context.
  
  We overload notation and use $\aprg \funin \nat \pto \nat$ to denote
  the \emph{one-step computation} of the program~$\aprg$,
  defined for all $N \ge 1$ by
  \begin{equation*}
    \funap{\aprg}{N} = N\cdot\iafrc{\applicable(N)}
  \end{equation*}
  where it is to be understood that $\undefd{\funap{\aprg}{N}}$
  whenever $\undefd{\applicable(N)}$.
  The \emph{run of $\aprg$ on $N$} is the finite or infinite sequence
  $N,\aprg(N),\aprg^2(N),\ldots$.
  We say that $\aprg$ \emph{halts} or \emph{terminates} on $N$ 
  if the run of $\aprg$ on $N$ is finite.
\end{definition}

The halting problem for Fractran programs is undecidable.
\newcommand{\citecade}{\cite[Theorem~2.2]{endr:grab:hend:2009}}
\begin{proposition}[\citecade]
  The uniform halting problem for Fractran programs, 
  that is, deciding whether a program halts for every starting integer $N \ge 0$, 
  is $\cpi{0}{2}$\nb-complete.
\end{proposition}

\newcommand{\citelics}{\cite[Theorem~68]{grab:endr:hend:klop:moss:2012}}
\begin{proposition}[\citelics]\label{prop:fractran2}
  The input-$2$ halting problem for Fractran programs, 
  that is, deciding whether a program halts for the starting integer $N = 2$, 
  is $\csig{0}{1}$\nb-complete.
\end{proposition}

\begin{remark}\label{rem:primes:interchangeable}
  In some sense it does not matter which prime numbers are used in a Fractran program.
  Let us make this precise.
  Let $p$ be a prime number, and $n$ a positive integer. 
  Then let $\padic{p}{n}$ denote the \emph{$p$\nb-adic valuation of $n$}  
  i.e., $\padic{p}{n} = a$ with $a\in\nat$ maximal such that~$p^a$ divides~$n$.
  For $\vec{p} = \tuple{p_1,p_2,\ldots,p_t}$ 
  we write $\padic{\vec{p}}{n}$ to denote $\tuple{\padic{p_1}{n},\padic{p_2}{n},\ldots,\padic{p_t}{n}}$.
  Let $\aprg$ be a Fractran program with $t$ distinct primes $\vec{p} = p_1, p_2, \ldots, p_t$, 
  let $\vec{q} = q_1, q_2, \ldots, q_t$ be any vector of $t$ distinct primes,
  and let $\bprg$ be the program obtained from $\aprg$ by uniformly substituting 
  the $q_i$'s for the $p_i$'s. 
  Then clearly, for all integers $M,N \ge 0$ 
  such that $\padic{\vec{p}}{M} = \padic{\vec{q}}{N}$,
  we have $\padic{\vec{p}}{\aprg^i(M)} = \padic{\vec{q}}{\bprg^i(N)}$ for all $i \ge 0$.
\end{remark}

We employ Fractran programs to define finite or infinite words over the alphabet $\{0,1\}$
by giving the primes $3$ and $5$ a special meaning,
namely for indicating output $0$ and $1$, respectively.
The construction easily generalizes to arbitrary finite alphabets.
\begin{definition}\label{def:fractran:output}
  Let $\aprg$ be a Fractran program.
  The finite or infinite \emph{word $W_{\!\aprg}$ computed by $\aprg$} is 
  $W_{\!\aprg} = W(2)$ where $W(N) = \emp$ if the sequence $\aprg(N),\aprg^2(N),\ldots$ 
  does not contain values divisible by $3$ or $5$,
  (note that this includes $W(N) = \emp$ if $\undefd{\aprg(N)}$),
  and otherwise 
  \begin{align*}
    W(N) = 
    \begin{cases}
      0\;W(\aprg(N)) &\text{if $3 \divides \aprg(N)$,} \\
      1\;W(\aprg(N)) &\text{if $5 \divides \aprg(N)$ and $3 \ndivides \aprg(N)$,} \\
      W(\aprg(N)) &\text{otherwise.}
    \end{cases}
  \end{align*}
\end{definition}

So the word~$W_{\!\aprg}$ is infinite if and only if 
$\aprg$ does not terminate on input $2$ and
the run of $\aprg$ on $N$
contains infinitely many numbers that are divisible by $3$ or $5$.
The infinite word can be read off from the infinite run by
dropping all entries neither divisible by~$3$ nor~$5$,
and then mapping the remaining entries to $0$ or $1$,
if they are divisible by $3$ or $5$ (and not $3$), respectively.

\begin{example}
  The Fractran program $\frac{3}{2},\frac{5}{3},\frac{3}{5}$
  gives rise to the computation $3,5,3,5,3,5,\ldots$,
  and hence computes the infinite word $010101\ldots$ of alternating bits.
\end{example}

\begin{proposition}\label{prop:frac:comp}
  Every (finite or infinite) computable, binary word
  can be computed by a Fractran program.
\end{proposition}

\begin{proof}
  In \cite{endr:grab:hend:2009} it is shown that Fractran programs
  can simulate any Turing machine computation.
  By Remark~\ref{rem:primes:interchangeable}
  we may assume that this translation does not employ the primes $\{2,3,5\}$.
  Then a straightforward adaptation of the proof in \cite{endr:grab:hend:2009}
  yields the claim:
  we multiply the fractions corresponding to the Turing machine
  generating an output $0$ or $1$ by the primes $3$ or $5$, respectively, 
  and make sure the thus introduced factor $3$ or $5$
  is removed in the next step by putting fractions $\frac{1}{3}$ and $\frac{1}{5}$ in front of the program.
\end{proof}

We define a Fractran\nb-$n$ program and compile it to a Fractran\nb-$1$ program~$\iaprg{\BIN}$ 
which computes an infinite word that has every finite binary word as one of its factors.
For this we use the bijective `$\msf{z}$\nb-representation' defined as follows.
\begin{definition}\label{def:BIN}
  Let $\aalph = \enumset{0,1}$. 
  For all $n \in \nat$ and $w\in\wrd{\aalph}$, we define $\zrep{n} \in \wrd{\aalph}$
  and $\zval{w} \in \nat$ by
  \begin{align*}
    \zrep{0} & = \wrdemp &&&
    \zval{\wrdemp} & = 0 \\
    \zrep{2n+1} & = 0\zrep{n} &&&
    \zval{0w} & = 2\zval{w}+1 \\
    \zrep{2n+2} & = 1 \zrep{n} &&&
    \zval{1w} & = 2\zval{w}+2
  \end{align*}
  and we let $\BIN$ denote the infinite word 
  \begin{align*}
    \BIN & = \zrep{0} \zrep{1} \zrep{2} \cdots = 0 \, 1 \, 00 \, 10 \, 01 \, 11 \, 000 \, 100 \, 010 \, \cdots
  \end{align*}
\end{definition}

We will now define a Fractran program that computes~$\BIN$;
it will be the compilation of the following Fractran\nb-7 program:
\begin{align*}
  \st{1}  & \col \frac{\reg{2}}{\reg{3}} \to \st{5} \sep \frac{\reg{1}}{1} \to \st{2} &
  \out{0} & \col \frac{1}{1} \to \st{1} \\
  \st{2}  & \col \frac{\reg{2}\reg{3}}{\reg{1}} \to \st{2} \sep \frac{1}{1} \to \st{3} &
  \out{1} & \col \frac{1}{1} \to \st{1} \\
  \st{3}  & \col \frac{\reg{1}}{\reg{3}} \to \st{3} \sep \frac{1}{1} \to \st{4} \\
  \st{4}  & \col \frac{\reg{3}}{\reg{2}^2} \to \st{4} \sep 
                 \frac{1}{\reg{2}} \to \out{0} \sep
                 \frac{1}{\reg{3}} \to \out{1} \\
  \st{5}  & \col \frac{\reg{2}}{\reg{3}} \to \st{5} \sep \frac{1}{1} \to \st{4}
\end{align*}%
\noindent%
We first explain its workings, and then compile it into a Fractran\nb-$1$ program.
Let $e_1,e_2,e_3$ be the register contents of the current integer~$N$, i.e., 
such that $N = \reg{1}^{e_1} \reg{2}^{e_2} \reg{3}^{e_3}$.
In the run (= sequence of states) of the above program starting in~$\st{1}$ with $e_1 = e_2 = e_3 = 0$,
the subsequence of `output' states $\out{0}$ and $\out{1}$
corresponds to the infinite word $\BIN$.
The idea is that $\reg{1}$ holds the current value~$n$ 
for producing the factor $\zrep{n}$ of $\BIN$.
State~$\st{1}$ with $e_3 = 0$ increments~$e_1$, and the program proceeds in state $\st{2}$.
States~$\st{2}$ and~$\st{3}$ copy $e_{1}$ to $e_{2}$ and we continue in state~$\st{4}$.
State~$\st{4}$ subtracts~$2$ from $e_{2}$ while incrementing $e_{3}$ as long as possible
(corresponding to division of $\reg{2}$ by~$2$ and storing the quotient in~$e_{3}$), 
and then goes to output state $\out{0}$ if the remainder $e_{2} \ne 0$, 
and to output state $\out{1}$ after decrementing $e_{3}$, otherwise 
(corresponding to the definition of $\zrep{\cdot}$ above).
After any of the two output states, the program returns to state $\st{1}$.
State~$\st{1}$ with a non-zero quotient~$\reg{3}$ copies $e_{3}$ to $e_{2}$
using state $\st{5}$, and then continues with state~$\st{4}$.

We compile the above program into a flat list of fractions 
using the steps (i)--(iii) given above.
For the looping states $\st{2}$, $\st{3}$, $\st{4}$, and $\st{5}$,
we introduce mirror states $\sh{2}$, $\sh{3}$, $\sh{4}$, and $\sh{5}$.
Second, we assign the following prime numbers to the identifiers:
\[
  \scalebox{.87}{$%
    \begin{array}{*{14}{c}}
      \st{1} & \st{2} & \sh{2} & \st{3} & \sh{3} & \st{4} & \sh{4} & \st{5} & \sh{5} & \out{0} & \out{1} & \reg{1} & \reg{2} & \reg{3} \\
      2 & 7 & 11 & 13 & 17 & 19 & 23 & 29 & 31 & 3 & 5 & 37 & 41 & 43
    \end{array}%
  $}
\]%
Finally, with \ref{item:compile:3}, we obtain the following Fractran\nb-$1$ program:
\begin{align*}
  \iaprg{\BIN}
  \mathrel{=}\mbox{}
  &\frac{\reg{2}\st{5}}{\reg{3}\st{1}}, \frac{\reg{1}\st{2}}{\st{1}},
  \frac{\reg{2}\reg{3}\sh{2}}{\reg{1}\st{2}}, \frac{\st{3}}{\st{2}}, \frac{\st{2}}{\sh{2}},
  \frac{\reg{1}\sh{3}}{\reg{3}\st{3}}, \frac{\st{4}}{\st{3}}, \frac{\st{3}}{\sh{3}}, \\
  &\frac{\reg{3}\sh{4}}{\reg{2}^2\st{4}}, \frac{\out{0}}{\reg{2}\st{4}}, \frac{\out{1}}{\reg{3}\st{4}}, \frac{\st{4}}{\sh{4}},
  \frac{\reg{2}\sh{5}}{\reg{3}\st{5}}, \frac{\st{4}}{\st{5}}, \frac{\st{5}}{\sh{5}},
  \frac{\st{1}}{\out{0}},
  \frac{\st{1}}{\out{1}}
\end{align*}%
which is run on input $N = \st{1} = 2$
to force the program to start in the initial state.
We note that the huge number mentioned at the end of the introduction 
is the least common denominator of~$\iaprg{\BIN}$.

\begin{proposition}\label{prop:BIN}
  The word computed by~$\iaprg{\BIN}$ is $\BIN$.
  \qed
\end{proposition}

\section[Productivity for Erasing PD0L Systems]{Productivity for Erasing \pdol{} Systems}\label{sec:prod}
We show that the problem of deciding productivity of erasing \pdol{} systems
is undecidable.
The idea is to encode a given Fractran program~$\aprg$ as 
a \pdol{} system~$\mcl{H}_{\aprg} = \tuple{\aalph,H,s}$
such that $H^\omega(s)$ is infinite if and only if $\aprg$ does not terminate on input~$2$.

We consider Fractran programs of the form $\frac{\ianum{1}}{\aden},\ldots,\frac{\ianum{k}}{\aden}$;
every program 
can be brought into this form 
by taking $\aden$ the least common denominator of the fractions.


\begin{definition}\label{def:pdol:prod}
  Let $\aprg = \frac{\ianum{1}}{\aden},\ldots,\frac{\ianum{k}}{\aden}$ be a Fractran program.
  We define the \pdol{} system $\mcl{H}_{\aprg} = \tuple{\balph,H,\symbs}$ 
  where 
  \[ \balph = \enumset{\symbs, \spa, \symba, \symbA, \symbb, \symbB} \] 
  and $H = \tuple{h_0,\ldots,h_{\aden-1}}$  
  consisting of morphisms $h_i \funin \wrd{\balph} \to \wrd{\balph}$ 
  defined for all $i\in\nalph{\aden}$ as follows: 
  \begin{align}
    h_i(\symbs) &= \symbs \,\nspa{\aden-1} \symba \symba \symbb \nspa{\aden-1} \\
    h_i(\spa)   &= \emp\\
    h_i(\symba) &= 
      \begin{cases}
      \symbA \nspa{\aden-1} &\text{if $i = \aden-1$}\\
      \emp &\text{otherwise}
      \end{cases} \\
    h_i(\symbb) &= \symbB \nspa{\aden-1-i}\\
    h_i(\symbA) &= 
    \begin{cases}
    \nsymba{\ianum{\applicable(i)}} & \text{if $\applicable(i)$ is defined}\\
    \emp &\text{otherwise}
    \end{cases}\\
    h_i(\symbB) &= 
    \begin{cases}
    \nsymba{i \cdot \frac{\ianum{\applicable(i)}}{\aden}}\,\symbb \nspa{\aden-1} & \text{if $\applicable(i)$ is defined}\\
    \emp &\text{otherwise}
    \end{cases}
  \end{align}%
\end{definition}

Before we show that productivity of the \pdol{} system $\mcl{H}_{\aprg}$
coincides with $\aprg$ running forever on input~$2$,
we give some intuition and an example to illustrate the working of $\mcl{H}_{\aprg}$.

The following trivial fact is useful to state separately.
\begin{lemma}\label{lem:remainder}
  Let $N,d,q,r \in \nat$ such that $N = q d + r$,
  and $\aprg$ a Fractran program.
  Then $\applicable_{\aprg}(N) = \applicable_{\aprg}(r)$.
  If moreover $b \in \nat$ divides $d$, 
  then $b \divides N$ \ if and only if \ $b \divides r$. 
  \qed
\end{lemma}

Let $\aprg$ be a Fractran program with common denominator~$\aden$,
and (finite or infinite) run $N_0,N_1,N_2,\ldots$.
Let $q_i\in\nat$ and $r_i\in\nalph{\aden}$ such that $N_i = q_i d + r_i$, for all $i \ge 0$.
We let $x_n$ be the `contribution' of the iteration $H^{n+1}$, i.e.,  
$x_n$ is such that $H^{n+1}(\symbs) = H^{n}(\symbs) x_n$.
Then $H^\omega(\symbs) = \symbs x_0 x_1 x_2 \cdots$.
We will display $H^\omega(\symbs)$ in separate lines each corresponding to an~$x_n$.
The computation of the word $H^\omega(\symbs)$ proceeds in two alternating phases:
the transition from even to odd lines corresponds to division by $d$,
and the transition from odd to even lines corresponds to multiplication 
by the currently applicable fraction $\frac{\ianum{\applicable(N_i)}}{\aden}$. 
These phases are indicated by the use of lower- and uppercase letters,
that is,
$x_{2n} \in \{ \spa , \symba , \symbb \}$
and $x_{2n+1} \in \{ \spa , \symbA , \symbB \}$,
as can be seen from the definition of the morphisms.
Now the intuition behind the alphabet symbols 
(in view of the defining rules of the morphisms)
can be described as follows.
We use $\symbs$ as the starting symbol, and the symbol~$\spa$ is used 
to shift the morphism index of subsequent letters.

In every even line~$x_{2i}$ 
\begin{enumerate}
  \item 
    there is precisely one block of $\symba$'s;
    this block is positioned at morphism index~$0$ and is of length $N_i$,
    representing the current value $N_i$ in the run of $\aprg$;
  \item 
    $\symbb$ is a special marker for the end of a block of $\symba$'s, 
    so positioned at morphism index $r_i$, the remainder of dividing $N_i$ by $\aden$.
\end{enumerate}%
In every odd line~$x_{2i+1}$
\begin{enumerate}%
  \setcounter{enumi}{2}%
  \item 
    the number of $\symbA$'s corresponds to the quotient $q_i$, 
    and every occurrence of $\symbA$ is positioned at morphism index $r_i$\,;
  \item 
    $\symbB$ (also at morphism index $r_i$) 
    takes care of the multiplication of the remainder~$r_i$ 
    with $\frac{\ianum{\applicable(N_i)}}{\aden}$.
    Then $\applicable(N_i) = \applicable(r_i)$ ensures, via Lemma~\ref{lem:remainder}, 
    that the morphism can select the right fraction to multiply with.
\end{enumerate}

We illustrate the encoding by means of an example.
\begin{example}
  Consider the Fractran program $\frac{9}{2},\frac{5}{3}$, 
  or equivalently
  \begin{align*}
    \aprg = \frac{27}{6},\frac{10}{6}
  \end{align*}
  and its finite run $2,9,15,25$.
  Following Definition~\ref{def:pdol:prod} 
  we construct 
  the \pdol{} system $\mcl{H}_{\aprg} = \tuple{\balph,H,\symbs}$
  with $H = \tuple{h_0,\ldots,h_5}$ and
  \begin{align*}
    h_i(\symbs) &= \symbs \,\nspa{5} \symba \symba \symbb \nspa{5} \\
    h_0(\symba) &= \ldots = h_4(\symba) = \emp \\
    h_5(\symba) &= \symbA \nspa{5} \\
    h_i(\symbb) &= \symbB \nspa{5-k}\\
    h_0(\symbA) &= h_2(\symbA) = h_4(\symbA)= \nsymba{27}\\
    h_3(\symbA) &= \nsymba{10}\\
    h_1(\symbA) &= h_5(\symbA)= \emp\\
    h_0(\symbB) &= \symbb \nspa{5}\\
    h_2(\symbB) &= \nsymba{9} \symbb \nspa{5}\\
    h_4(\symbB) &= \nsymba{18} \symbb \nspa{5}\\
    h_3(\symbB) &= \nsymba{5} \symbb \nspa{5}\\
    h_1(\symbB) &= h_5(\symbB)= \emp
  \end{align*}%
  for $i \in \nalph{6}$.
  Then $H^\omega(\symbs)$ is finite and the stepwise 
  computation of this fixed point can be displayed as follows.
  To ease reading, we write below each letter its morphism index.
  Let $z_n$ denote the morphism index of $x_n$.
  Moreover, the word $H^\omega(\symbs) = \symbs x_0 x_1 \cdots$ 
  is broken into lines in such a way that every line $x_{n+1}$ is the image of the previous line~$x_n$ under $H_{z_n}$
  (except for the line~$x_0$, which is the tail of the image of~$\symbs$ under $H = H_0$).
  \begin{align*}
          &\xline{ \sat{\symbs}{0}  }\\
    x_0 = \mbox{} &\xline{ \sat{\nspa{5}}{1}\sat{\symba}{0}\sat{\symba}{1}\sat{\symbb}{2} \sat{\nspa{5}}{3} }\\
    x_1 = \mbox{} &\xline{ \sat{\symbB}{2} \sat{\nspa{3}}{3} }\\
    x_2 = \mbox{} &\xline{ \sat{\symba}{0}\sat{\symba}{1}\sat{\symba}{2}\sat{\symba}{3}\sat{\symba}{4}\sat{\symba}{5}\sat{\symba}{0}\sat{\symba}{1}\sat{\symba}{2}\sat{\symbb}{3} \sat{\nspa{5}}{4} }\\
    x_3 = \mbox{} &\xline{ \sat{\symbA}{3}\sat{\nspa{5}}{4}\sat{\symbB}{3} \sat{\nspa{2}}{4} }\\
    x_4 = \mbox{} &\xline{ \sat{\symba}{0}\sat{\symba}{1}\sat{\symba}{2}\sat{\symba}{3}\sat{\symba}{4}\sat{\symba}{5}\sat{\symba}{0}\sat{\symba}{1}\sat{\symba}{2}\sat{\symba}{3}\;\sat{\symba}{4}\sat{\symba}{5}\sat{\symba}{0}\sat{\symba}{1}\sat{\symba}{2}\sat{\symbb}{3} \sat{\nspa{5}}{4} }\\
    x_5 = \mbox{} &\xline{ \sat{\symbA}{3}\sat{\nspa{5}}{4}\sat{\symbA}{3}\sat{\nspa{5}}{4}\sat{\symbB}{3} \sat{\nspa{2}}{4} }\\
    x_6 = \mbox{} &\xline{ \sat{\symba}{0}\sat{\symba}{1}\sat{\symba}{2}\sat{\symba}{3}\sat{\symba}{4}\sat{\symba}{5}\sat{\symba}{0}\sat{\symba}{1}\sat{\symba}{2}\sat{\symba}{3}\;\sat{\symba}{4}\sat{\symba}{5}\sat{\symba}{0}\sat{\symba}{1}\sat{\symba}{2}\sat{\symba}{3}\sat{\symba}{4}\sat{\symba}{5}\sat{\symba}{0}\sat{\symba}{1}\;\sat{\symba}{2}\sat{\symba}{3}\sat{\symba}{4}\sat{\symba}{5}\sat{\symba}{0}\sat{\symbb}{1} \sat{\nspa{5}}{2} }\\
    x_7 = \mbox{} &\xline{ \sat{\symbA}{1}\sat{\nspa{5}}{2}\sat{\symbA}{1}\sat{\nspa{5}}{2}\sat{\symbA}{1}\sat{\nspa{5}}{2}\sat{\symbA}{1}\sat{\nspa{5}}{2}\sat{\symbB}{1} }\\
    x_8 = \mbox{} & \wrdemp
  \end{align*}%
\end{example}

Now we characterize the contribution of every iteration of $H$
in the construction of the word $H^\omega(2)$.
We employ the notations given in Lemma~\ref{lem:H:rec}.
\begin{lemma}\label{lem:prod:growth}
  Let $\aprg = \frac{\ianum{1}}{d},\ldots,\frac{\ianum{k}}{d}$,
  and $N \ge 1$. 
  Let $q\in\nat$ and $r\in\nalph{d}$ be such that 
  $N = qd + r$. Let $\Xspa = \nspa{d-1}$.
  Then we have
  \begin{align}
    H(\nsymba{N} \symbb \Xspa) = (\symbA \Xspa)^q \symbB \nspa{d-1-r} 
    \label{eq:lem:prod:growth:1}
  \end{align}
  of length $d(q+1) - r$.
  If, moreover, $\aprg(N)$ is defined, 
  then
  \begin{align}
    H_r((\symbA \Xspa)^q \symbB \nspa{d-1-r})
    & = \nsymba{\aprg(N)} \symbb \Xspa
    \label{eq:lem:prod:growth:2}
  \end{align}
  of length $\aprg(N) + d$.
\end{lemma}
\begin{proof}
  Equation~\eqref{eq:lem:prod:growth:1} follows immediately by induction on $q$.
  To see that~\eqref{eq:lem:prod:growth:2} holds
  for $\defd{\aprg(N)}$,
  we note that $\applicable(N)$ is defined, 
  and so is $\applicable(r) = \applicable(N)$, 
  by Lemma~\ref{lem:remainder}.
  Hence we obtain
  \begin{align*}
    H_r((\symbA \Xspa)^q \symbB \nspa{d-1-r})
    & = (\nsymba{\ianum{\applicable(r)}})^q  
        H_r(\symbB \nspa{d-1-r}) \\
    & = (\nsymba{\ianum{\applicable(r)}})^q  
        \nsymba{r\cdot\frac{\ianum{\applicable(r)}}{d}} \symbb \Xspa  
  \end{align*}
  and we conclude by 
  $F(N) = N \cdot \frac{\ianum{\applicable(N)}}{d} 
        = q \cdot \ianum{\applicable(N)} + r \cdot \frac{\ianum{\applicable(N)}}{d}$.
\end{proof}

\begin{lemma}\label{lem:prod}
  For all Fractran programs~$\aprg$, 
  the \pdol{} system~$\mcl{H}_{\aprg}$ 
  is productive if and only if $\aprg$ does not terminate on input~$2$.
\end{lemma}
\begin{proof}
  Let $\aprg$ and $\mcl{H}_{\aprg}$ be as in Definition~\ref{def:pdol:prod}.
  Let $N_0,N_1,N_2,\ldots$ 
  be the finite or infinite run of $\aprg$ on $2$,
  i.e., $N_i = F^i(2)$, and let $t \in \nat \cup \enumset{\infty}$ 
  denote its length. 
  For all $i$ with $0 \le i \lt t$,
  let $q_i \in \nat$ 
  and $r_i \in \nalph{d}$ 
  be such that $N_i = q_i d + r_i$.
  
  We define 
  $x_n \in \wrd{\aalph}$
  and 
  $z_n\in \nalph{d}$ for all $n \ge 0$, as follows.
  Let $\Xspa = \nspa{d-1}$, $x_0 = \symba\symba\symbb \Xspa$, $z_0 = 0$,
  and, for $n \ge 1$, let $x_n$ and $z_n$ be such that 
  $H^{n+1}(\symbs) = H^{n}(\symbs) x_n$ and
  $z_n \equiv \length{H^{n}(\symbs)} \pmod{d}$.
  Then $H^\omega(\symbs) = \symbs \Xspa x_0 x_1 x_2 \cdots$, 
  and the factor $x_n$ is 
  at morphism index~$z_n$. 
  With Lemma~\ref{lem:H:rec} we then have
  \begin{align}
    x_{n} = H_{z_{n-1}}(x_{n-1}) &&& z_{n} \equiv z_{n-1} + \length{x_{n-1}} \pmod{d}
    \label{eq:lem:prod:H:rec}
  \end{align}
  for all $n \ge 1$.
  Now we prove by induction on $n \ge 0$ that
  \begin{align*}
    x_{n}   &= \nsymba{N_i} \symbb \Xspa 
    & z_{n} &= 0 && \text{if $n = 2i \lt 2t$,} \\
    x_{n}   &= (\symbA \Xspa)^{q_{i}} \symbB \nspa{d-1-r_{i}}
    & z_{n} &= r_{i} && \text{if $n=2i+1 \lt 2t$,} \\
    x_{n}   & = \wrdemp
    & z_{n} &= 0 && \text{if $n \ge 2t$.}
  \end{align*}
  The base case is immediate. 
  Let $n \gt 0$.
  If $n = 2i \lt 2t$ for some $i \lt t$, 
  then $N_i = F(N_{i-1})$ is defined, and 
  \(
    x_n = H_{z_{n-1}}(x_{n-1}) 
        = H_{r_{i-1}}((\symbA \Xspa)^{q_{i-1}} \symbB \nspa{d-1-r_{i-1}}) 
        = \nsymba{N_i} \symbb \Xspa
  \), 
  and 
  \(
    z_n \equiv z_{n-1} + \length{x_{n-1}} 
        \equiv r_{i-1} + d(q_{i-1}+1) - r_{i-1} 
        \equiv 0 \pmod{p}
  \),
  both by \eqref{eq:lem:prod:H:rec}, the induction hypothesis 
  and Lemma~\ref{lem:prod:growth}.
  If $n = 2i+1 \lt 2t$ for some $i \lt t$, then 
  \(
    x_{n} = H_{z_{n-1}}(x_{n-1}) 
          = H_{0}(\nsymba{N_i} \symbb \Xspa)
          = (\symbA \Xspa)^{q_i} \symbB \nspa{d-1-r_i}
  \), 
  and
  \(
    z_n \equiv z_{n-1} + \length{x_{n-1}}
        \equiv 0 + N_i + d 
        \equiv r_i \pmod{p}
  \),
  again by \eqref{eq:lem:prod:H:rec}, the induction hypothesis 
  and Lemma~\ref{lem:prod:growth}.
  Finally, if $n = 2t$, then 
  \(
    x_n = H_{z_{n-1}}(x_{n-1})
        = H_{r_{t-1}}((\symbA \Xspa)^{q_{t-1}} \symbB \nspa{d-1-r_{t-1}})
        = \wrdemp
  \),
  since $\aprg$ terminates on $N_{t-1}$ 
  (and so $\applicable(N_{t-1})$ and $\applicable(r_{t-1})$ are undefined),
  and $z_n \equiv z_{n-1} + \length{x_{n-1}} \equiv r_{t-1} + d(q_{t-1} + 1) - r_{t-1} \equiv 0$.
  Clearly, then also $x_n = \wrdemp$ and $z_n = 0$ for all $n \gt 2t$.
\end{proof}

Hence, by Lemma~\ref{lem:prod} and Proposition~\ref{prop:fractran2}, deciding productivity of \pdol{} systems is undecidable.
\begin{theorem}\label{thm:prod}
  The problem of deciding on the input of a \pdol{} system $\mcl{H}$
  whether $\mcl{H}$ is productive,
  is $\cpi{0}{1}$\vphantom{\raisebox{1pt}{$\Pi^0_2$}}\nb-complete.
  \qed
\end{theorem}

\section[Turing Completeness of Non-Erasing PD0L Systems]{Turing Completeness \\ of Non-Erasing \pdol~Systems}\label{sec:comp}

In this section we extend the encoding of Fractran from the previous section
to show that every computable infinite word can be embedded in the following two ways.

\begin{definition}\label{def:embed:prefix}
  Let $\aalph$ and $\balph\supset\aalph$ be alphabets with
  letters $\symbl,\symbr \in \balph \setminus \aalph$,
  and let $\seq{w} \in \str{\aalph}$ and $\seq{u} \in \str{\balph}$ be infinite words.
  We say $\seq{w}$ is \emph{prefix embedded} in~$\seq{u}$ if the following three conditions are satisfied:
  \begin{enumerate}
    \item for every finite prefix $v$ of $\seq{w}$ there is an occurrence $\symbl\,v\,\symbr$ in $\seq{u}$,
    \item for every occurrence of a word $\symbl\,v\,\symbr$ in $\seq{u}$ with $v \in \wrd{(\balph \setminus \{\symbr\})}$
          we have that $v$ is a prefix of $\seq{w}$, and
    \item letters from $\aalph$ occur in $\seq{u}$ only in factors (subwords) 
          of the form $\symbl\,v\,\symbr$ with $v \in \wrd{(\balph \setminus \{\symbr\})}$.
  \end{enumerate} 
\end{definition}

\begin{definition}\label{def:embed:sparse}
  Let $\aalph$ and $\balph\supset\aalph$ be alphabets,
  and let $\seq{w} \in \str{\aalph}$ and $\seq{u} \in \str{\balph}$ be infinite words.
  We say $\seq{w}$ is \emph{sparsely embedded} in~$\seq{u}$
  if $\seq{w}$ is obtained from 
  $\seq{u}$ by erasing all letters in $\balph\setminus\aalph$.
\end{definition}

The crucial difference with the encoding of Section~\ref{sec:prod}
is that we now use the knowledge about the
remainder not only to select the correct fraction to multiply with,
but also to recognize when the current value is divisible
by $3$ or $5$, and correspondingly produce an output bit $\szero$ or~$\sone$,
cf.~Definition~\ref{def:fractran:output}.
The process again proceeds in two phases, for division and multiplication, 
and we employ lower- and uppercase letters accordingly. 
We introduce letters $\symbl$ (and $\symbL$) and $\symbr$ (and $\symbR$)
marking the beginning and the end of the prefix of the infinite word computed
by the Fractran program.
Furthermore, the symbol $\symbR$ produces the output bits depending on the current remainder~$r_i$.
In order to prevent that the output of $\symbR$
changes the morphism index of $\symbR$, we introduce $\symbi$ (and $\symbI$)
which compensate the production of $\symbR$ with an inverse length.
The letter $\symbe$ (and $\symbE$) marks the end of the line, 
and additionally $\symbe$ takes care of realignment after multiplication,
such that the first $\symba$ in each run stands on morphism index $0$.

\begin{definition}\label{def:dol:comp:prefix}
  Let $\aprg = \frac{\ianum{1}}{\aden},\ldots,\frac{\ianum{k}}{\aden}$ 
  be a Fractran program such that (without loss of generality) 
  the common denominator~$\aden$ is divisible by $3$ and $5$.
  Define the \pdol{} system~$\mcl{H}_{\aprg} = \tup{\balph,H,\symbs}$~with
  \begin{align*}
    \balph = \enumset{ \symbs, \spa, \smin, \symba, \symbA, \symbb, \symbB, \symbi, \symbI, \symbl, \symbL, \szero, \sone , \symbr, \symbR, \symbe, \symbE, \symbQ }
  \end{align*}
  and $H = \tuple{h_0,\ldots,h_{\aden-1}}$  
  consisting of (non-erasing) morphisms 
  $h_i : \wrd{\balph} \to \wrd{\balph}$ defined for every $i\in\nalph{\aden}$ as follows:
  \begin{align*}
    h_i(\symbs) &= \symbs \nspa{\aden-1} \symba \symba \symbb \nspa{\aden-1} \symbi \nspa{\aden-2} \symbl \symbr \nspa{\aden-1} \symbe\\
    h_i(\spa) &= \spa\\
    h_i(\symba) &= 
      \begin{cases}
      \symbA \nspa{\aden-1} &\text{if $i = \aden-1$\,,}\\
      \nspa{\aden} &\text{otherwise.}
      \end{cases} \\
    h_i(\symbb) &= \symbB\\
    h_i(\symbi) &= \symbI\\
    h_i(\symbl) &= \symbL\\
    h_i(\symbr) &= \symbR\\
    h_i(\symbe) &= \symbE \nsmin{\aden-i}\\
    h_i(\smin) &= \nspa{\aden}\\
    h_i(\symbA) &= 
    \begin{cases}
    \nsymba{\ianum{\applicable(i)}} \spa & \text{if $\applicable(i)$ is defined,}\\
    \symbQ &\text{otherwise.}
    \end{cases}\\
    h_i(\symbB) &= 
    \begin{cases}
    \nsymba{i \cdot \frac{\ianum{\applicable(i)}}{\aden}}\symbb & \text{if $\applicable(i)$ is defined,}\\
    \symbQ &\text{otherwise.}
    \end{cases}\\
    h_i(\symbI) &= 
    \begin{cases}
    \symbi \nspa{\aden-1} &\text{if $3 \divides i$ or $5 \divides i$\,,}\\
    \symbi &\text{otherwise.}
    \end{cases}\\
    h_i(\symbL) &= \symbl\\
    h_i(\szero) &= \szero\\
    h_i(\sone) &= \sone\\
    h_i(\symbR) &= 
    \begin{cases}
    \szero\symbr &\text{if $3 \divides i$\,,}\\
    \sone\symbr &\text{if $5 \divides i$\,,}\\
    \symbr &\text{otherwise.}
    \end{cases}\\
    h_i(\symbE) &= \symbe\\
    h_i(\symbQ) &= \symbQ
  \end{align*}
\end{definition}

\begin{remark}\label{rem:3:5:div:lcd}
  In Definition~\ref{def:dol:comp:prefix} we require that $3$ and $5$ divide the common denominator 
  in order for Lemma~\ref{lem:remainder} to apply.
  Informally speaking, via the remainder we can only observe factors that also divide the common denominator.
\end{remark}

\begin{remark}
  Let $\aprg$ and $H$ be as in Definition~\ref{def:dol:comp:prefix}.
  It can be shown that the symbol~$\symbQ$ occurs in the word $H^\omega(\symbs)$ 
  if and only if the Fractran program~$\aprg$ halts on input~$2$.
  This fact can be used to show that it is undecidable whether $\symbQ$ occurs in $H^\omega(\symbs)$.
  However, we prove this differently, namely by applying Theorem~\ref{thm:comp:prefix} and using the fact
  that for non-terminating Fractran programs 
  it is undecidable whether digit~$\sone$ occurs in the sequence computed by the program.
  See Theorem~\ref{thm:letter}. 
\end{remark}

Time for an example.
\begin{example}
  Consider the following Fractran program:
  \begin{align*}
    F = \frac{2}{7},\frac{3\cdot 7}{2\cdot 5},\frac{3}{2},\frac{5}{3},\frac{2}{1} 
  \end{align*}
  which has the infinite run $2$, $3$, $5$, $10$, $21$, $6$, $9$, $15$, $25$, $50$, $105$,
  $30$, $63$, $18$, $27$, $45$, $75$, $125$, $250$, $525$, $150$, $\ldots$
  and computes the word 
  \begin{align*}
    011000011000000011000000000011000000000000011\cdots,
  \end{align*}
  that is, $0^1 11 \; 0^{3+1}11\; 0^{6+1}11\; 0^{9+1}11 \; 0^{12+1}\ldots = \prod_{i=0}^\infty 0^{3\cdot i+1} 11$.
  Writing the program with the common denominator $210$ yields:
  \begin{align*}
    \aprg & = \frac{60}{210},\frac{441}{210},\frac{315}{210},\frac{350}{210},\frac{420}{210}
  \end{align*}
  {%
  \renewcommand{\sat}[2]{
    \node (n) [at=(n.south east),anchor=south west] {$#1$};
    \node (a) [at=(n.south),anchor=north,yshift=-1mm] {\rotatebox{90}{\tiny #2}};
  }%
  \renewcommand{\scalebox}[2]{
    \begin{scope}[nodes={scale=#1}]
      #2
    \end{scope}
  }%
  Let $\mcl{H}_{\aprg}$ be the \pdol{} encoding of $\aprg$, 
  as given in Definition~\ref{def:dol:comp:prefix}.
  We consider the first steps of the iteration of the morphisms;
  for easier reading, we drop blocks of consecutive symbols $\nspa{210}$ 
  (they do not change the morphism index of other letters),
  and let $\Xspa = \nspa{209}$.
  \begin{align*}
    &\xline{ \sat{\symbs}{0} }\\
    x_0 = \mbox{} &\xline{ \sat{\Xspa}{1}\sat{\nsymba{2}}{0}\sat{\symbb}{2}\sat{\Xspa}{3}\sat{\symbi}{2}\sat{\nspa{208}}{3}\sat{\symbl}{1}\sat{\symbr}{2}\sat{\Xspa}{3}\sat{\symbe}{2} }\\
    x_1 = \mbox{} &\xline{ \sat{\Xspa}{3}\sat{\symbB}{2}\sat{\Xspa}{3}\sat{\symbI}{2}\sat{\nspa{208}}{3}\sat{\symbL}{1}\sat{\symbR}{2}\sat{\Xspa}{3}\sat{\symbE}{2} \sat{\nsmin{208}}{3} }\\
    x_2 = \mbox{} &\xline{ \sat{\Xspa}{1}\sat{\nsymba{3}}{0}\sat{\symbb}{3}\sat{\Xspa}{4}\sat{\symbi}{3}\sat{\nspa{208}}{4}\sat{\symbl}{2}\sat{\symbr}{3}\sat{\Xspa}{4}\sat{\symbe}{3} }\\
    x_3 = \mbox{} &\xline{ \sat{\Xspa}{4}\sat{\symbB}{3}\sat{\Xspa}{4}\sat{\symbI}{3}\sat{\nspa{208}}{4}\sat{\symbL}{2}\sat{\symbR}{3}\sat{\Xspa}{4}\sat{\symbE}{3} \sat{\nsmin{207}}{4} }\\
    x_4 = \mbox{} &\xline{ \sat{\Xspa}{1}\sat{\nsymba{5}}{0}\sat{\symbb}{5}\sat{\Xspa}{6}\sat{\symbi}{5}\sat{\nspa{207}}{6}\sat{\symbl}{3}\sat{\szero}{4}\sat{\symbr}{5}\sat{\Xspa}{6}\sat{\symbe}{5} }\\
    x_5 = \mbox{} &\xline{ \sat{\Xspa}{6}\sat{\symbB}{5}\sat{\Xspa}{6}\sat{\symbI}{5}\sat{\nspa{207}}{6}\sat{\symbL}{3}\sat{\szero}{4}\sat{\symbR}{5}\sat{\Xspa}{6}\sat{\symbE}{5} \sat{\nsmin{205}}{6} }\\
    x_6 = \mbox{} &\xline{ \sat{\Xspa}{1}\sat{\nsymba{10}}{0}\sat{\symbb}{10}\sat{\Xspa}{11}\sat{\symbi}{10}\sat{\nspa{206}}{11}\sat{\symbl}{7}\sat{\szero}{8}\sat{\sone}{9}\sat{\symbr}{10}\sat{\Xspa}{11}\sat{\symbe}{10} }\\
    x_7 = \mbox{} &\xline{ \sat{\Xspa}{11}\sat{\symbB}{10}\sat{\Xspa}{11}\sat{\symbI}{10}\sat{\nspa{206}}{11}\sat{\symbL}{7}\sat{\szero}{8}\sat{\sone}{9}\sat{\symbR}{10}\sat{\Xspa}{11}\sat{\symbE}{10} \sat{\nsmin{200}}{11} }\\
    x_8 = \mbox{} &\xline{ \sat{\Xspa}{1}\sat{\nsymba{21}}{0}\sat{\symbb}{21}\sat{\Xspa}{22}\sat{\symbi}{21}\sat{\nspa{205}}{22}\sat{\symbl}{17}\sat{\szero}{18}\sat{\sone}{19}\sat{\sone}{20}\sat{\symbr}{21}\sat{\Xspa}{22}\sat{\symbe}{21} }\\
    x_9 = \mbox{} &\xline{ \sat{\Xspa}{22}\sat{\symbB}{21}\sat{\Xspa}{22}\sat{\symbI}{21}\sat{\nspa{205}}{22}\sat{\symbL}{17}\sat{\szero}{18}\sat{\sone}{19}\sat{\sone}{20}\sat{\symbR}{21}\sat{\Xspa}{22}\sat{\symbE}{21} \sat{\nsmin{189}}{22} }
  \end{align*}%
  Note that the number of $\symba$'s in rows $x_{2i}$ precisely
  models the run of the Fractran program $2$, $3$, $5$, $10$, $21$, \ldots.
  By construction, the morphism index of $\symbR$ equals the remainder $r_i$,
  and consequently $\symbR$ produces the prefix $011$ of the word computed by $\aprg$.
  
  Things become more interesting if the number of $\symba$'s exceeds the denominator $210$.
  We look at a few steps further down the sequence:
  \begin{align*} 
    x_{36} = \mbox{} &\xline{ \sat{\Xspa}{1}\sat{\nsymba{250}}{0}\sat{\symbb}{40}\sat{\Xspa}{41}\sat{\symbi}{40}\sat{\nspa{191}}{41}\sat{\symbl}{22} \scalebox{.41}{\sat{\szero}{23}\sat{\sone}{24}\sat{\sone}{25}\sat{\szero}{26}\sat{\szero}{27}\sat{\szero}{28}\sat{\szero}{29}\sat{\sone}{30}\sat{\sone}{31}\sat{\szero}{32}\sat{\szero}{33}\sat{\szero}{34}\sat{\szero}{35}\sat{\szero}{36}\sat{\szero}{37}\sat{\szero}{38}\sat{\sone}{39}} \sat{\symbr}{40}\sat{\Xspa} {41}\sat{\symbe}{40} }\\
    x_{37} = \mbox{} &\xline{ \sat{\Xspa}{41}\sat{\symbA}{40}\sat{\Xspa}{41}\sat{\symbB}{40}\sat{\Xspa}{41}\sat{\symbI}{40}\sat{\nspa{191}}{41}\sat{\symbL}{22}  \scalebox{.41}{\sat{\szero}{23}\sat{\sone}{24}\sat{\sone}{25}\sat{\szero}{26}\sat{\szero}{27}\sat{\szero}{28}\sat{\szero}{29}\sat{\sone}{30}\sat{\sone}{31}\sat{\szero}{32}\sat{\szero}{33}\sat{\szero}{34}\sat{\szero}{35}\sat{\szero}{36}\sat{\szero}{37}\sat{\szero}{38}\sat{\sone}{39}} \sat{\symbR}{40}\sat{\Xspa}{41}\sat{\symbE}{40} \sat{\nsmin{170}}{41} }\\
    x_{38} = \mbox{} &\xline{ \sat{\Xspa}{1}\sat{\nsymba{525}}{0}\sat{\symbb}{105}\sat{\Xspa}{106}\sat{\symbi}{105}\sat{\nspa{190}}{106}\sat{\symbl}{86} \scalebox{.41}{\sat{\szero}{87}\sat{\sone}{88}\sat{\sone}{89}\sat{\szero}{90}\sat{\szero}{91}\sat{\szero}{92}\sat{\szero}{93}\sat{\sone}{94}\sat{\sone}{95}\sat{\szero}{96}\sat{\szero}{97}\sat{\szero}{98}\sat{\szero}{99}\sat{\szero}{100}\sat{\szero}{101}\sat{\szero}{102}\sat{\sone}{103}} \sat{\sone}{104}\sat{\symbr}{105}\sat{\Xspa}{106}\sat{\symbe}{105} }\\
    x_{39} = \mbox{} &\xline{ \sat{\Xspa}{106}\sat{(}{}\sat{\symbA}{105}\sat{\Xspa}{106}\sat{)^2}{}\sat{\symbB}{105}\sat{\Xspa}{106}\sat{\symbI}{105}\sat{\nspa{190}}{106}\sat{\symbL}{86} \scalebox{.41}{\sat{\szero}{87}\sat{\sone}{88}\sat{\sone}{89}\sat{\szero}{90}\sat{\szero}{91}\sat{\szero}{92}\sat{\szero}{93}\sat{\sone}{94}\sat{\sone}{95}\sat{\szero}{96}\sat{\szero}{97}\sat{\szero}{98}\sat{\szero}{99}\sat{\szero}{100}\sat{\szero}{101}\sat{\szero}{102}\sat{\sone}{103}} \sat{\sone}{104}\sat{\symbR}{105}\sat{\Xspa}{106}\sat{\symbE}{105} \sat{\nsmin{105}}{106} }\\
    x_{40} = \mbox{} &\xline{ \sat{\Xspa}{1}\sat{\nsymba{150}}{0}\sat{\symbb}{150}\sat{\Xspa}{151}\sat{\symbi}{150}\sat{\nspa{189}}{151}\sat{\symbl}{130} \scalebox{.41}{\sat{\szero}{131}\sat{\sone}{132}\sat{\sone}{133}\sat{\szero}{134}\sat{\szero}{135}\sat{\szero}{136}\sat{\szero}{137}\sat{\sone}{138}\sat{\sone}{139}\sat{\szero}{140}\sat{\szero}{141}\sat{\szero}{142}\sat{\szero}{143}\sat{\szero}{144}\sat{\szero}{145}\sat{\szero}{146}\sat{\sone}{147}}\sat{\sone}{148}\sat{\szero}{149}\sat{\symbr}{150}\sat{\Xspa}{151}\sat{\symbe}{150} }
  \end{align*}
  In line~$x_{36}$ we have $250$ $\symba$'s giving rise to only one $\symbA$ in the subsequent line $x_{37}$.
  Again, $\symbA$, $\symbB$, $\symbI$ and $\symbR$ stand on index $40 \equiv 250 \pmod{210}$,
  and consequently both deduce that the first applicable fraction is $\frac{441}{210} = \frac{3\cdot 7}{2\cdot 5}$.
  The letter $\symbA$ represents the quotient from the division by $210$, and hence produces $441$ $\symba$'s.
  The letter $\symbB$ is responsible for the multiplication of the remainder,
  and thus produces $40 \cdot \frac{3\cdot 7}{2\cdot 5} = 84$ $\symba$'s. 
  Thus we get $441 + 84  = 525$ $\symba$'s in line $x_{38}$.
  Moreover, $\symbR$ produces a $\sone$ since $250$ is dividable by $5$ but not $3$,
  and $\symbI$ produces an $\Xspa$-block ($\Xspa = \nspa{209}$)
  to keep $\symbR$ and the remaining symbols on the correct index.

  Now the division of $525$ by $210$ has quotient $2$ and remainder $105$,
  and so we have two $\symbA$'s in line $x_{39}$, 
  and all $\symbA$'s, $\symbB$, $\symbI$ and $\symbR$ standing on index $105$.
  The first applicable fraction for $105$ is $\frac{60}{210} = \frac{2}{7}$,
  and correspondingly the two $\symbA$'s produce $60$ $\symba$'s each,
  and $\symbB$ produces $30 = 105\cdot \frac{2}{7}$ $\symba$'s, in total giving rise to $150$ $\symba$'s in line $x_{40}$.
  Now $525$ is divisible by $3$ and so $\symbR$ produces a~$\szero$.
  }
\end{example}

We now start working towards a proof of Theorem~\ref{thm:comp:prefix}.
Let $\aprg = \frac{\ianum{1}}{d},\ldots,\frac{\ianum{k}}{d}$ be a Fractran program.
We again employ the notation $H_i$ as given in Lemma~\ref{lem:H:rec}.
Furthermore, we define relations ${\simpstep},{\simp} \subseteq \aalph^* \times \aalph^*$ by
\begin{align*}
  {\simpstep} &= \{\pair{u \nspa{d} v}{uv} \where u,v \in \wrd{\aalph}\} &&&
  {\simp} &= ({\reflectbox{$\simpstep$}} \cup {\simpstep})^*
\end{align*}
Then clearly we have $H_i(u) \simp H_i(v)$\, for all $i\in\nalph{d}$, 
$u,v \in \aalph^*$ \mbox{with~$u \simp v$}.
This allows us to prove properties of $H^\omega(\symbs)$ reasoning modulo~${\simp}$.
Below, we write $\nspa{n}$ with $n < 0$ to denote the block $\nspa{m}$ with $m \in \nalph{d}$ and $n \equiv m \pmod{d}$.
For $N \ge 1$ we define 
\begin{align*}
  \obit{N} &= \szero &&\text{if $3 \divides N$}\\
  \obit{N} &= \sone &&\text{if $3 \ndivides N$ and $5 \divides N$}\\
  \obit{N} &= \wrdemp &&\text{otherwise} 
\end{align*}

\begin{lemma}\label{lem:comp:growth}
  Let $\aprg = \frac{\ianum{1}}{d},\ldots,\frac{\ianum{k}}{d}$,
  and $N \ge 1$. 
  Let $\iobit{} = \obit{N}$ and $\Xspa = \nspa{d-1}$.
  Let $q\in\nat$ and $r\in\nalph{d}$ be such that $N = qd + r$,
  and let $v \in \wrd{\enumset{\szero,\sone}}$.
  Then we have
  \begin{gather}
    \begin{aligned}
    & H_1( \Xspa\nsymba{N} \symbb \Xspa\symbi\nspa{d-2-|v|}\symbl  v \symbr \Xspa\symbe ) \\
    & 
      = \Xspa(\symbA \Xspa)^q \symbB \Xspa\symbI\nspa{d-2-\length{v}}\symbL v \symbR \Xspa\symbE \nsmin{d-r}
    \end{aligned}
    \label{eq:lem:comp:growth:1}
  \end{gather}
  of length equivalent to~$-r$ modulo~$d$.

  Moreover, if in addition $\aprg(N)$ is defined, 
  then we have
  \begin{gather}
    \begin{aligned}
    & H_{r+1}( \Xspa(\symbA \Xspa)^q \symbB \Xspa\symbI\nspa{d-2-\length{v}}\symbL v\symbR \Xspa\symbE \nsmin{d-r} )\\
    & \simp \Xspa\nsymba{\aprg(N)} \symbb \Xspa\symbi\nspa{d-2-|v\iobit{}|}\symbl  v\iobit{} \symbr \Xspa\symbe
    \end{aligned}
    \label{eq:lem:comp:growth:2}
  \end{gather}
  of length equivalent to~$\aprg(N)$ modulo~$d$.
\end{lemma}
\begin{proof}
  \eqref{eq:lem:comp:growth:1} follows immediately:
  Let $x,y\in\wrd{\balph}$ be arbitrary.
  Then
  $H_1 (\Xspa x) = \Xspa H_0(x)$, 
  and $H_0(\nsymba{N} y) = \Xspa (\symbA \Xspa)^q H_0(\nsymba{r} y) = \Xspa (\symbA \Xspa)^q H_r(y)$.
  Furthermore, 
  \(
    H_r(\symbb \Xspa\symbi\nspa{d-2-|v|}\symbl  v \symbr \Xspa\symbe) 
    = \symbB \Xspa\symbI\nspa{d-2-\length{v}}\symbL v \symbR H_{r}(\symbe)
  \)
  and $H_r(\symbe) = \nsmin{d-r}$.  
  
  To show~\eqref{eq:lem:comp:growth:2}, let $\aprg(N)$ be defined.
  Then $\applicable(N) = \applicable(r)$ is also defined (Lemma~\ref{lem:remainder}).
  Hence we get, for $w\in\wrd{\balph}$ arbitrary:
  \begin{align*}
    & H_{r+1}( \Xspa(\symbA \Xspa)^q \symbB \Xspa w) \\
    & \ = \Xspa H_r( (\symbA \Xspa)^q \symbB \Xspa w) \\
    & \ = \Xspa (\nsymba{\ianum{\applicable(r)}} \spa \Xspa)^q H_r( \symbB \Xspa w ) \\
    & \ = \Xspa (\nsymba{\ianum{\applicable(r)}} \spa \Xspa)^q \nsymba{r \cdot \frac{\ianum{\applicable(r)}}{\aden}}\symbb \Xspa H_r(w) \\ 
    & \ \simp \Xspa \nsymba{\aprg(N)} \symbb \Xspa H_r(w) 
  \end{align*}
  Finally, to compute $H_r(w)$ for $w = \symbI\nspa{d-2-\length{v}}\symbL v\symbR \Xspa\symbE \nsmin{d-r}$,
  distinguish the following cases: 
%
  If $3$ or $5$ divides $r$, 
  then $3$ or $5$ divides also $N$, respectively, by Lemma~\ref{lem:remainder}.
  Hence we have $\iobit{} = \obit{N} = \obit{r}$
  and $\length{\iobit{}} = 1$,
  and 
  \begin{align*}
    H_{r}(w) 
    & = \symbi \Xspa \nspa{d-2-\length{v}} \symbl v H_{r}( \symbR \Xspa\symbE \nsmin{d-r}) \\
    & = \symbi \Xspa \nspa{d-2-\length{v}} \symbl v \iobit{} \symbr \Xspa \symbe \nspa{d(d-r)} \\
    & \simp \symbi \nspa{d-3-\length{v}} \symbl v \iobit{} \symbr \Xspa \symbe
  \end{align*}
  as required.
  And , if $3 \ndivides r$ and $5 \ndivides r$, then $\iobit{} = \obit{N} = \obit{r} = \wrdemp$, 
  and $H_r(w) \simp \symbi \nspa{d-2-\length{v}} \symbl v \symbr \Xspa \symbe$,
  as required.
\end{proof}

\begin{lemma}\label{lem:comp:prefix}
  Let $\aprg$ be a Fractran program computing an infinite word~$\seq{w}$.
  Then the \pdol{} system~$\mcl{H}_\aprg$ from Definition~\ref{def:dol:comp:prefix} 
  generates a word that prefix embeds $\seq{w}$ (see Definition~\ref{def:embed:prefix}).
\end{lemma}
\begin{proof}
  Let $\aprg = \frac{\ianum{1}}{\aden},\ldots,\frac{\ianum{k}}{\aden}$ be a Fractran program,
  computing $\seq{w} \in \str{\{\szero,\sone\}}$,
  i.e., $\seq{w} = W_{\aprg}$ with $W_{\aprg}$ as defined in Definition~\ref{def:fractran:output}.
  Let $N_0,N_1,N_2,\ldots$ be the infinite run of $\aprg$ starting on $N_0 = 2$
  (so with infinitely many $N_i$ divisible by~$3$~or~$5$).
  Let $q_i\in\nat$ and $r_i\in\nalph{d}$ be such that $N_i = q_i d + r_i$.
  Let $\mcl{H}_{\aprg} = \tup{\balph,H,\symbs}$ be the \pdol{} system defined in Definition~\ref{def:dol:comp:prefix}.

  We show that $\seq{u} = H^\omega(\symbs)$ satisfies the conditions~(i), (ii), and~(iii) of Definition~\ref{def:embed:prefix}, 
  by characterizing the contribution of every iteration of $H$.
  For every $i\in\nat$ we let $v_i \in \wrd{\aalph}$
  be defined by $v_0 = \wrdemp$ 
  and $v_{i+1} = v_i \obit{N_i}$ so that $\seq{w} = \lim_{i\to\infty} v_i$.
  For all $n \ge 0$
  let $x_n \in \wrd{\balph}$ and $z_n \in \nalph{d}$ be such that 
  $H^{n+1}(\symbs) = H^n(\symbs) x_n$ and $z_n \equiv \length{H^n(\symbs)} \pmod{d}$.
  Then $H^\omega(\symbs) = \symbs x_0 x_1 \cdots$,
  and 
  \begin{align}
    x_n & = H_{z_{n-1}}(x_{n-1}) &&& z_n & \equiv z_{n-1} + \length{x_{n-1}} \pmod{d} 
    \label{eq:lem:comp:H:rec}
  \end{align}
  Let us abbreviate $Y_i = \nspa{d-2-|v_{i}|}$.
  We prove that $x_n$ and $z_n$ satisfy
  \begin{align*}
    x_{2i} & \simp \Xspa \nsymba{N_i} \symbb \Xspa \symbi Y_i \symbl  v_i \symbr \Xspa \symbe \\ 
    x_{2i+1} & \simp \Xspa (\symbA \Xspa)^{q_i} \symbB \Xspa \symbI Y_i \symbL v_i \symbR \Xspa \symbE \nsmin{d-r_i} \\
    z_{2i} & \equiv 1 \pmod{d}\\
    z_{2i+1} & \equiv r_i + 1 \pmod{d} 
  \end{align*}
  by induction on $n$.

  For the base case, we see $z_0 = \length{\symbs} = 1$
  and $x_0 = \Xspa \symba \symba \symbb \Xspa \symbi \nspa{\aden-2} \symbl \symbr \Xspa \symbe$,
  as required.
  So let $n \ge 1$. 
  If $n=2i$ for some $i \ge 1$, 
  it follows from \eqref{eq:lem:comp:H:rec}, Lemma~\ref{lem:comp:growth}
  and the induction hypothesis that
  \begin{align*}
    x_n 
    & = H_{z_{n-1}}(x_{n-1}) \\
    & \simp H_{r_{i-1}+1}(\Xspa (\symbA \Xspa)^{q_{i-1}} \symbB \Xspa \symbI Y_{i-1} \symbL v_{i-1} \symbR \Xspa \symbE \nsmin{d-r_{i-1}})\\
    & \simp \Xspa\nsymba{N_i} \symbb \Xspa \symbi Y_i \symbl  v_i \symbr \Xspa \symbe
  \end{align*}
  and $z_n \equiv z_{n-1} + \length{x_{n-1}} \equiv r_{i-1} + 1 - r_{i-1} \equiv 1 \pmod{d}$.

  Similarly, if $n = 2i+1$ for some $i \ge 0$,
  we obtain
  \begin{align*}
    x_n 
    & = H_{z_{n-1}}(x_{n-1}) \\
    & \simp H_{1}(\Xspa \nsymba{N_i} \symbb \Xspa \symbi Y_i \symbl  v_i \symbr \Xspa \symbe) \\
    & \simp \Xspa \nsymba{N_{i+1}} \symbb \Xspa \symbi Y_{i+1} \symbl  v_{i+1} \symbr \Xspa \symbe
  \end{align*}
  and $z_n \equiv z_{n-1} + \length{x_{n-1}} \equiv 1 + N_{i+1} \equiv r_i + 1 \pmod{d}$.
  
  Knowing the exact shape (modulo~$\simp$) of $\seq{u} = H^\omega(\symbs)$,
  it is now easy to verify that 
  $\seq{u}$
  it satisfies (i), (ii), and (iii),
  taking into account that~$\spa$ 
  does not occur in any factor $\symbl\,v\,\symbr$ of $\seq{u}$ with $v \in \wrd{(\balph \setminus \{\symbr\})}$, 
  by the definition of the morphisms.
\end{proof}

We are ready to collect our main results.
\begin{theorem}\label{thm:comp:prefix}
  Every computable infinite word can be prefix embedded in a \pdol{} word (see Definition~\ref{def:embed:prefix}).
\end{theorem}
\begin{proof}
  Let $\seq{w} \in \str{\{\szero,\sone\}}$ be
  an infinite \mbox{computable} word. 
  Then, by Proposition~\ref{prop:frac:comp}, 
  $\seq{w}$ is computed by some Fractran program.
  By Lemma~\ref{lem:comp:prefix} the claim follows.
\end{proof}

\begin{definition}\label{def:dol:comp:sparse}
  Let $\aprg$ be a Fractran program,
  and $\mcl{H}_{\aprg}$ the \pdol{} system 
  given in Definition~\ref{def:dol:comp:prefix}
  We define the \pdol{} system $\mcl{H}'_{\aprg}$ 
  as the result of replacing in $\mcl{H}_{\aprg}$ 
  the rules $h_i(\szero) = \szero$ and $h_i(\sone) = \sone$ 
  by
  $h_i(\szero) = \spa$ and $h_i(\sone) = \spa$,
  for all $i \in \nalph{\aden}$.
\end{definition}

\begin{lemma}\label{lem:comp:sparse}
  Let $\aprg$ be a Fractran program computing an infinite word~$\seq{w}$,
  and let $\seq{u}\in\str{\balph}$ be the \pdol{} word generated 
  by the $\mcl{H}'_\aprg$ defined in Definition~\ref{def:dol:comp:sparse}.
  Then $\seq{w}$ is sparsely embedded in $\seq{u}$.
\end{lemma}
\begin{proof}
  By an easy adaptation of the proof of Lemma~\ref{lem:comp:prefix},
  noting that every output $\szero$ and $\sone$ is produced precisely
  once and in the next iteration replaced by $\spa$.
\end{proof}

\begin{theorem}\label{thm:comp:sparse}\label{thm:erasing:letters}
  Every computable infinite word can be sparsely embedded in a \pdol{} word (see Definition~\ref{def:embed:sparse}).
\end{theorem}
\begin{proof}
  Analogous to the proof of Theorem~\ref{thm:comp:prefix},
  replacing Lemma~\ref{lem:comp:prefix} by Lemma \ref{lem:comp:sparse}.
\end{proof}

It is known that the set of morphic words is closed under finite state transductions~\cite[Theorem~7.9.1]{allo:shal:2003}.
In particular, if we erase all occurrences of a certain letter from a morphic word, 
the result is a morphic (or finite) word.
From Theorem~\ref{thm:erasing:letters} it follows that this is not the case for \pdol{} words,
establishing a negative answer to Problem~29~(1) and~(2) of \cite{much:prit:seme:2009}.

\begin{corollary}\label{cor:fst}
  The set of \pdol{} words is not closed under finite state transductions.
\end{corollary}
  
\begin{proof}
  There are computable streams that are not \pdol{} words~\cite{culi:karh:1992};
  hence the class of \pdol{} words is not closed under finite state transductions,
  by Theorem~\ref{thm:erasing:letters} 
  (erasing letters is a finite state transduction).
\end{proof}

Finite state transducers play a central role in computer science.
The transducibility relation via finite state transducers (FST) gives
rise to a hierarchy of degrees of infinite words~\cite{endr:hend:klop:2011},
analogous to the recursion theoretic hierarchy.
But, in contrast to the latter, the FST-hierarchy does not identify all computable streams.
An open problem in this area is the lack of methods
for discriminating infinite words $\seq{u}, \seq{v}$,
that is, to show that there exists no finite state transducer
that transduces $\seq{u}$ to $\seq{v}$.
Discriminating morphic words seems to require heavier machinery
than arguments based on the pumping lemma.

We will now collect several immediate consequences of Theorem~\ref{thm:comp:prefix}.
First of all, we have solved the open problem~\cite{lepi:1993} 
on the existence of \pdol{} words that have exponential subword complexity. 
\begin{theorem}\label{thm:exp}
  There is a \pdol{} word~$\seq{u}$ 
  such that $\subwcompl{\seq{u}}{n} \ge 2^n$. 
\end{theorem}
\begin{proof}
  Let $\aprg = \aprg_{\BIN}$ be the Fractran program defined in Section~\ref{sec:fractran},
  computing the word $W_{\!\aprg} = \BIN$ (Proposition~\ref{prop:BIN}).
  Furthermore, let $\mcl{H}_{\aprg} = \tuple{\balph,H,\symbs}$ be 
  the \pdol{} system of Definition~\ref{def:dol:comp:prefix}.
  Then, by Lemma~\ref{lem:comp:prefix}, 
  $\seq{u} = H^{\omega}(\symbs)$ is the word we are looking for.
\end{proof}

Lemma~\ref{lem:comp:prefix} also allows us to 
give a negative answer to \cite[Problem~29 (3)]{much:prit:seme:2009}.
%
\begin{theorem}\label{thm:letter}
  The following problems are undecidable: 
  \problem{%
    \pdol{} system $\mcl{H} = \tuple{\balph,H,\symbs}$, letter $b \in \balph$
  }{%
      \begin{enumerate}
        \item 
          Does $b$ occur in $H^\omega(\symbs)$?
        \item 
          Does $b$ occur infinitely many times in $H^\omega(\symbs)$?
      \end{enumerate}%
  }
  \noindent
\end{theorem}
\begin{proof}
  We show that the following problem is undecidable: 
  given a Fractran program $\aprg$ computing an infinite word $\seq{w}$ over the alphabet~$\{0,1\}$, 
  does the letter~$1$ occur in $\seq{w}$?
  This suffices since by Lemma~\ref{lem:comp:prefix},
  if $\seq{u}$ is the infinite word generated by  $\mcl{H}_\aprg$,
  then the letter $1$ occurs in~$\seq{u}$
  if and only if $1$ 
  occurs infinitely often in $\seq{u}$
  if and only if $1$ 
  occurs in $\seq{w}$.

  We use the input~$2$ halting problem for Fractran programs which
  is $\csig{0}{1}$-complete by Proposition~\ref{prop:fractran2}.
  Let $\aprg$ be an arbitrary Fractran program.
  By Remark~\ref{rem:primes:interchangeable}
  we can replace the primes in $\aprg$ to obtain a program $\aprg'$ that
  does not contain the primes $\{2,3,5\}$ such that 
  $\aprg'$ halts on $7$ if and only if $\aprg$ halts on $2$.
  We now extend $\aprg'$ to $\aprg''$ by adding in front
  the fraction $\frac{3\cdot 7}{2}$ and at the end the fractions $\frac{5}{3}$~and~$\frac{1}{1}$.
  Then the first fraction of $\aprg''$ starts $\aprg'$ on input $7$ and 
  ensures that the output is $0$ for every step that $\aprg'$ is running, and only
  when $\aprg'$ terminates, the last two fractions of $\aprg''$ 
  switch the output to $1$ and keep running forever.
\end{proof}

From Theorem~\ref{thm:letter} it follows immediately that 
the first-order (and monadic second-order) theory of \pdol{} words is undecidable,
answering~\cite[Problem~28]{much:prit:seme:2009};
see~\cite{much:prit:seme:2009} also for the definition of the first-order and monadic theory of a sequence.
This again stands in contrast to the case for morphic sequences,
which are known to have a decidable monadic second-order theory~\cite{cart:thom:2002}.
\begin{corollary}\label{cor:logic:undecidable}
  The first-order theory of \pdol{} words 
  is undecidable.
\end{corollary}

Also immediate from Theorem~\ref{thm:letter} is the undecidability of equivalence of \pdol{} systems
(equality of the limit words they generate).
We note that equivalence of \dol{} systems is decidable~\cite{culi:harj:1984},
whereas that of \cdol{} words is an open problem.
\begin{corollary}\label{cor:equality:undecidable}
  Equality of \pdol{} words (given by their \pdol{} systems) is undecidable.
\end{corollary}
\begin{proof}
  We reduce problem~(i) stated in Theorem~\ref{thm:letter} 
  to equivalence of \pdol{} systems, as follows.
  Let $\mcl{H} = \tup{\aalph,H,s}$ be a \pdol{} system and $b \in \aalph$, 
  and let $\mcl{H}' = \tup{\aalph\cup\{b'\},H',s'}$ where $b' \not\in \aalph$ and
  $H'$ and $s'$ are obtained from $H$ and $s$
  by replacing all occurrences of $b$ by $b'$, and letting $H'(b)=b$.
  Then $b$ does not occur in the word generated by $\mcl{H}$ if and only if $\mcl{H}$ and $\mcl{H'}$ 
  generate the same word. By Theorem~\ref{thm:letter} this is undecidable.
\end{proof}

\section[A Concrete PD0L Word with Exponential Subword Complexity]{A Concrete \pdol{} Word \\ with Exponential Subword Complexity}\label{sec:exp}
\newcommand{\act}{a}
\newcommand{\nrestart}{r}
\newcommand{\carry}{c}
\newcommand{\one}{o}
\renewcommand{\symbb}{\cbox{corange}{\symbfont{b}}}
\renewcommand{\symbB}{\cbox{corange}{\symbfont{B}}}
\renewcommand{\symbR}[1]{\cbox{cgreen}{\symbfont{R}_{#1}}}
\newcommand{\symbRns}{\cbox{cgreen}{\symbfont{R}}}
\renewcommand{\symbI}[1]{\cbox{cdarkgreen}{\symbfont{Z}_{#1}}}
\newcommand{\symbIns}{\cbox{cdarkgreen}{\symbfont{Z}}}

In this section we give a concrete example of a \pdol{} system 
which generates an infinite word with exponential subword complexity.
The word embeds all prefixes of the word $\BIN = \zrep{0} \zrep{1} \zrep{2} \cdots$ 
given in Definition~\ref{def:BIN}.
We refrain from proving that it indeed does have this property;
the existence of such a \pdol{} word is already proved in the previous section, 
see Theorem~\ref{thm:comp:prefix}.

We define a \pdol{} system $H = \tuple{h_0,h_1,\ldots,h_{15}}$
consisting of $16$ morphisms.
We express morphism indices $i \in \nalph{16}$ by linear combinations
\begin{align*}
  i = \act(i) \cdot 2^3 + \nrestart(i) \cdot 2^2 + \carry(i) \cdot 2^1 + \one(i) \cdot 2^0.
\end{align*}
with $\act(i),\nrestart(i),\carry(i),\one(i) \in \{0,1\}$ which we call \emph{flags}.
We use these flags to transmit information between symbols:
\begin{itemize}
  \item $\act(i) = 1$ stands for \emph{active},
  \item $\nrestart(i) = 1$ stands for \emph{running},
  \item $\carry(i) = 1$ stands for \emph{carry flag}, and
  \item $\one(i) = 1$ stands for \emph{output one}.
\end{itemize}
The idea is to simulate a binary counter, 
using the representation of Definition~\ref{def:BIN}.
The counter repeatedly increments ($+1$) the current value,
and thereby brings $\zrep{n}$ to $\zrep{n+1}$.
During an increment process we need to shift the activity from bit to bit.
To this end, the activity flag $\act(i)$ indicates 
whether a symbol at morphism index $i$ is active.

We explain the increment process using the following example word.
Here $\cbox{chighlight}{\ldots}$ is the already produced prefix of $\BIN$,
and we assume for the moment that the symbols $\symba$,  $\symbb$ and $\symbd$ 
each stand for a word of length $16$, and $\symbc$ for a word of length $8$.
\begin{align}
  \xline{
    \sat{\symba}{\act} \sat{\symbc}{\act} \sat{\symbb}{} \sat{\symbc}{} \sat{\symbc}{\act} \sat{\symbb}{} \sat{\symbc}{} \sat{\symbc}{\act} \sat{\symbc}{} \sat{\symbc}{\act} \sat{\symba}{} \sat{\symbL}{} \sat{\cbox{chighlight}{\ldots}}{} \sat{\symbRns}{}
  }
  \label{eq:0110:0}
\end{align}
Here $\symba$ and $\symbb$ represent the bits $\szero$ and $\sone$, respectively,
and we shall continue to call them bits.
Ignoring the $\symbc$'s in between, \eqref{eq:0110:0} represents the word $\cbox{chighlight}{0110}$ 
(in turn representing the integer $21$).
Apart from incrementing this initial word~$\symba\symbb\symbb\symba$,
it is at the same time `copied' bit by bit to the word $\cbox{chighlight}{0110}$ 
between symbols $\symbL$ and~$\symbRns$.
The least significant bit is left,
and consequently the process of incrementing will proceed from left to right.
The symbol $\symbc$ (being of length $8$) swaps the value of $\act(i)$ 
for the morphism index~$i$ of all subsequent letters.
Note that between the $n$\nb-th and $(n+1)$\nb-th occurrence of bits ($\symba$ or $\symbb$), 
there are $2^{n-1}$ $\symbc$'s.
Hence, if the first bit is active, 
then this is the only active bit.

We now describe the transition from~\eqref{eq:0110:0} to its \pdol{} image~\eqref{eq:0110:1}.
Note that starting from the first occurrence of $\symbc$, 
every second occurrence in \eqref{eq:0110:0} has the activity flag set.
When the symbol $\symbc$ is active, it will be eliminated, that is
replaced by the symbol~$\symbd$ (of assumed length~$16$),
thus activating the next bit for the next iteration~\eqref{eq:0110:1}. 
\begin{align}
  \xline{
    \sat{\symbB}{\act} \sat{\symbd}{\act} \sat{\symbb}{\act} \sat{\symbc}{\act}\sat{\symbd}{} \sat{\symbb}{} \sat{\symbc}{} \sat{\symbd}{\act}\sat{\symbc}{\act}\sat{\symbd}{} \sat{\symba}{} \sat{\symbL}{} \sat{\cbox{chighlight}{\ldots}}{} \sat{\szero}{} \sat{\symbRns}{}
  }
  \label{eq:0110:1}
\end{align}
Note that $\symba$ is replaced by $\symbB$; 
uppercase letters are used for indicating 
the already processed bits during the increment loop.
When the increment loop is finished, uppercase will be turned to lowercase,
and the process restarts.
The switch from $\symba$ to $\symbB$ corresponds to incrementing.
This is controlled by the carry flag indicating whether a bit has to be flipped.
The carry flag is always set at the start of an increment loop.
To keep this example simple we do not display this flag.
At the end of this section we give the first iterations of the \pdol{} system
displaying all flags.

Notice that in \eqref{eq:0110:1} the second bit $\symbb$ is the only active bit (ignoring $\symbB$ which we have already dealt with).
Again, eliminating the active $\symbc$'s will shift the activity to the following bit:
\begin{align}
  \xline{
    \sat{\symbB}{\act} \sat{\symbd}{\act} \sat{\symbB}{\act} \sat{\symbd}{\act}\sat{\symbd}{\act} \sat{\symbb}{\act} \sat{\symbc}{\act} \sat{\symbd}{}\sat{\symbd}{}\sat{\symbd}{} \sat{\symba}{} \sat{\symbL}{} \sat{\cbox{chighlight}{\ldots}}{} \sat{\szero}{} \sat{\sone}{} \sat{\symbRns}{}
  }
\end{align}
After one more step we obtain:
\begin{align}
  \xline{
    \sat{\symbB}{\act} \sat{\symbd}{\act} \sat{\symbB}{\act} \sat{\symbd}{\act}\sat{\symbd}{\act} \sat{\symbB}{\act} \sat{\symbd}{\act}\sat{\symbd}{\act}\sat{\symbd}{\act}\sat{\symbd}{\act} \sat{\symba}{\act} \sat{\symbL}{\act} \sat{\cbox{chighlight}{\ldots}}{} \sat{\szero}{} \sat{\sone}{} \sat{\sone}{} \sat{\symbRns}{\act}
  }
\end{align}
and finally:
\begin{align}
  \xline{
    \sat{\symbB}{\act} \sat{\symbd}{\act} \sat{\symbB}{\act} \sat{\symbd}{\act}\sat{\symbd}{\act} \sat{\symbB}{\act} \sat{\symbd}{\act}\sat{\symbd}{\act}\sat{\symbd}{\act}\sat{\symbd}{\act} \sat{\symbA}{\act} \sat{\symbL}{\act} \sat{\cbox{chighlight}{\ldots}}{} \sat{\szero}{} \sat{\sone}{} \sat{\sone}{} \sat{\szero}{} \sat{\symbRns}{\act}
  }
\end{align}
As soon as the most significant bit $\symba$ is active, $\symbRns$ becomes active as well.
This can be used to recognize when the addition is finished,
and then $\symbRns$ unsets the bit $\nrestart(i)$ to restart the addition procedure.

The active bit makes use of the flag $\one(i)$ to `communicate' with the symbol~$\symbRns$
whether to output a $\szero$ or $\sone$.
This actually means that $\symbRns$ can produce the $\szero$ or $\sone$ only
two iterations later; for simplicity we have in this intuitive explanation abstracted
from this technicality and produce the $\szero$'s and $\sone$'s in the immediately
following iteration (after a bit has become active).

There are more symbols and technical subtleties to be explained, 
but we leave this to the imagination of the reader. Enjoy!

The morphisms $h_i$ are defined for all $i \in \nalph{16}$ as follows:
\begin{align*}
  h_i(\symbs) &= 
  \symbs\nspa{13}\splu\symba\nspa{15}\symbc\nspa{7}\symbX\nspa{15}\symbO\nspa{15}\symbI{1}\nspa{14}\symbL\symbR{1}\nspa{8}\\
  h_i(\nspa{}) &= 
  \nspa{}\\
  h_i(\splu) &= 
  \nspa{2}\\
  h_i(\nsmin{}) &= 
  \nspa{16}\\
  h_i(\symbp^{}) &= 
  \nspa{2}\nsmin{}\\
  h_i(\szero) &= 
  \szero\\
  h_i(\sone) &= 
  \sone\\
  h_i(\symba) &= 
  \begin{cases}
  \symba&\text{if }\neg \act(i)\\
  \nspa{14}\symbp^{}\symbA&\text{if }\act(i)  \wedge \neg \carry(i)\\
  \nspa{14}\symbp^{}\nspa{12}\nsplu{2}\symbB&\text{if }\act(i) \wedge \carry(i)
  \end{cases}\\
  h_i(\symbA) &= 
  \begin{cases}
  \symba&\text{if }\neg \nrestart(i)\\
  \symbA&\text{if }\nrestart(i)
  \end{cases}\\
  h_i(\symbb) &= 
  \begin{cases}
  \symbb&\text{if }\neg \act(i)\\
  \symbB&\text{if }\act(i)  \wedge \neg \carry(i)\\
  \symbA&\text{if }\act(i) \wedge \carry(i)
  \end{cases}\\
  h_i(\symbB) &= 
  \begin{cases}
  \symbb&\text{if }\neg \nrestart(i)\\
  \symbB&\text{if }\nrestart(i)
  \end{cases}\\
  h_i(\symbc) &= 
  \begin{cases}
  \symbc&\text{if }\neg \act(i)\\
  \symbd\nspa{8}&\text{if }\act(i)
  \end{cases}\\
  h_i(\symbd) &= 
  \begin{cases}
  \symbc\nspa{8}&\text{if }\neg \nrestart(i)\\
  \symbd&\text{if }\nrestart(i)
  \end{cases}\\
  h_i(\symbX) &= 
  \begin{cases}
  \symbX&\text{if }\neg \carry(i)\\
  \symbX&\text{if }\carry(i)  \wedge \neg \act(i)\\
  \symba&\text{if }\carry(i) \wedge \act(i)
  \end{cases}\\
  h_i(\symbo) &= 
  \begin{cases}
  \symbo&\text{if }\neg \carry(i)\\
  \symbo&\text{if }\carry(i)  \wedge \neg \act(i)\\
  \symbd&\text{if }\carry(i) \wedge \act(i)
  \end{cases}\\
  h_i(\symbO) &= 
  \begin{cases}
  \symbO&\text{if }\neg \act(i)\\
  \symbo\nspa{15}\symbO&\text{if }\act(i)  \wedge \neg \carry(i)\\
  \symbd\nspa{15}\symbX\nspa{15}\symbO&\text{if }\act(i) \wedge \carry(i)
  \end{cases}\\
  h_i(\symbIns) &= 
  \begin{cases}
  \symbIns&\text{if }\neg \nrestart(i)\\
  \symbIns\nspa{15}&\text{if }\nrestart(i)  \wedge \neg \act(i)\\
  \symbI{3}\nspa{15}&\text{if }\nrestart(i) \wedge \act(i)
  \end{cases}\\
  h_i(\symbI{1}) &= 
  \symbIns\\
  h_i(\symbI{2}) &= 
  \symbI{1}\\
  h_i(\symbI{3}) &= 
  \symbI{2}\\
  h_i(\symbL) &= 
  \symbL\\
  h_i(\symbRns) &= 
  \begin{cases}
  \symbRns&\text{if }\neg \nrestart(i)\\
  \szero\symbRns&\text{if }\nrestart(i)  \wedge \neg \one(i)  \wedge \neg \act(i)\\
  \szero\symbR{3}\nspa{8}\symbp^{4}\nsmin{10}&\text{if }\nrestart(i)  \wedge \neg \one(i) \wedge \act(i)  \wedge \neg \carry(i)\\
  \szero\symbR{3}\nspa{8}\symbp^{4}\nsmin{8}&\text{if }\nrestart(i)  \wedge \neg \one(i) \wedge \act(i) \wedge \carry(i)\\
  \sone\symbRns&\text{if }\nrestart(i) \wedge \one(i)  \wedge \neg \act(i)\\
  \sone\symbR{3}\nspa{8}\symbp^{4}\nsmin{10}&\text{if }\nrestart(i) \wedge \one(i) \wedge \act(i)  \wedge \neg \carry(i)\\
  \sone\symbR{3}\nspa{8}\symbp^{4}\nsmin{8}&\text{if }\nrestart(i) \wedge \one(i) \wedge \act(i) \wedge \carry(i)
  \end{cases}\\
  h_i(\symbR{1}) &= 
  \symbRns\\
  h_i(\symbR{2}) &= 
  \symbR{1}\\
  h_i(\symbR{3}) &= 
  \symbR{2}
\end{align*}
The \pdol{} word $H^\omega(\symbs)$ starts as follows:%
\renewcommand{\sat}[6]{
  \node (n) [at=(n.south east),anchor=south west] {$#1$};
  \node (a) [at=(n.south),anchor=north,yshift=-1mm] {{\tiny #2}};
  \node (a) [at=(n.south),anchor=north,yshift=-3.5mm] {{\tiny #3}};
  \node (a) [at=(n.south),anchor=north,yshift=-5mm] {{\tiny #4}};
  \node (a) [at=(n.south),anchor=north,yshift=-6.5mm] {{\tiny #5}};
  \node (a) [at=(n.south),anchor=north,yshift=-8mm] {{\tiny #6}};
}
\begin{align*}
  &\xline{ \sat{\symbs}{0}{{\color{white}\act}}{{\color{white}\nrestart}}{{\color{white}\carry}}{{\color{white}\one}}\sat{\nspa{13}}{1}{{\color{white}\act}}{{\color{white}\nrestart}}{{\color{white}\carry}}{\one}\sat{\nsplu{}}{14}{\act}{\nrestart}{\carry}{{\color{white}\one}}\sat{\symba}{15}{\act}{\nrestart}{\carry}{\one}\sat{\nspa{15}}{0}{{\color{white}\act}}{{\color{white}\nrestart}}{{\color{white}\carry}}{{\color{white}\one}}\sat{\symbc}{15}{\act}{\nrestart}{\carry}{\one}\sat{\nspa{7}}{0}{{\color{white}\act}}{{\color{white}\nrestart}}{{\color{white}\carry}}{{\color{white}\one}}\sat{\symbX}{7}{{\color{white}\act}}{\nrestart}{\carry}{\one}\sat{\nspa{15}}{8}{\act}{{\color{white}\nrestart}}{{\color{white}\carry}}{{\color{white}\one}}\sat{\symbO}{7}{{\color{white}\act}}{\nrestart}{\carry}{\one}\sat{\nspa{15}}{8}{\act}{{\color{white}\nrestart}}{{\color{white}\carry}}{{\color{white}\one}}\sat{\symbI{1}}{7}{{\color{white}\act}}{\nrestart}{\carry}{\one}\sat{\nspa{14}}{8}{\act}{{\color{white}\nrestart}}{{\color{white}\carry}}{{\color{white}\one}}\sat{\symbL}{6}{{\color{white}\act}}{\nrestart}{\carry}{{\color{white}\one}}\sat{\symbR{1}{}}{7}{{\color{white}\act}}{\nrestart}{\carry}{\one} \sat{\nspa{8}}{8}{\act}{{\color{white}\nrestart}}{{\color{white}\carry}}{{\color{white}\one}} }\\
  &\xline{ \sat{\nspa{29}}{0}{{\color{white}\act}}{{\color{white}\nrestart}}{{\color{white}\carry}}{{\color{white}\one}}\sat{\nsymbpp{}}{13}{\act}{\nrestart}{{\color{white}\carry}}{\one}\sat{\nspa{12}}{14}{\act}{\nrestart}{\carry}{{\color{white}\one}}\sat{\nsplu{2}}{10}{\act}{{\color{white}\nrestart}}{\carry}{{\color{white}\one}}\sat{\symbB}{12}{\act}{\nrestart}{{\color{white}\carry}}{{\color{white}\one}}\sat{\nspa{15}}{13}{\act}{\nrestart}{{\color{white}\carry}}{\one}\sat{\symbd}{12}{\act}{\nrestart}{{\color{white}\carry}}{{\color{white}\one}}\sat{\nspa{15}}{13}{\act}{\nrestart}{{\color{white}\carry}}{\one}\sat{\symbX}{12}{\act}{\nrestart}{{\color{white}\carry}}{{\color{white}\one}}\sat{\nspa{15}}{13}{\act}{\nrestart}{{\color{white}\carry}}{\one}\sat{\symbO}{12}{\act}{\nrestart}{{\color{white}\carry}}{{\color{white}\one}}\sat{\nspa{15}}{13}{\act}{\nrestart}{{\color{white}\carry}}{\one}\sat{\symbIns}{12}{\act}{\nrestart}{{\color{white}\carry}}{{\color{white}\one}}\sat{\nspa{14}}{13}{\act}{\nrestart}{{\color{white}\carry}}{\one}\sat{\symbL}{11}{\act}{{\color{white}\nrestart}}{\carry}{\one}\sat{\symbRns{}}{12}{\act}{\nrestart}{{\color{white}\carry}}{{\color{white}\one}} \sat{\nspa{8}}{13}{\act}{\nrestart}{{\color{white}\carry}}{\one} }\\
  &\xline{ \sat{\nspa{29}}{5}{{\color{white}\act}}{\nrestart}{{\color{white}\carry}}{\one}\sat{(\nspa{2}\smin)^{}}{2}{{\color{white}\act}}{{\color{white}\nrestart}}{\carry}{{\color{white}\one}}\sat{\nspa{16}}{5}{{\color{white}\act}}{\nrestart}{{\color{white}\carry}}{\one}\sat{\symbB}{5}{{\color{white}\act}}{\nrestart}{{\color{white}\carry}}{\one}\sat{\nspa{15}}{6}{{\color{white}\act}}{\nrestart}{\carry}{{\color{white}\one}}\sat{\symbd}{5}{{\color{white}\act}}{\nrestart}{{\color{white}\carry}}{\one}\sat{\nspa{15}}{6}{{\color{white}\act}}{\nrestart}{\carry}{{\color{white}\one}}\sat{\symbX}{5}{{\color{white}\act}}{\nrestart}{{\color{white}\carry}}{\one}\sat{\nspa{15}}{6}{{\color{white}\act}}{\nrestart}{\carry}{{\color{white}\one}}\sat{\symbo}{5}{{\color{white}\act}}{\nrestart}{{\color{white}\carry}}{\one}\sat{\nspa{15}}{6}{{\color{white}\act}}{\nrestart}{\carry}{{\color{white}\one}}\sat{\symbO}{5}{{\color{white}\act}}{\nrestart}{{\color{white}\carry}}{\one}\sat{\nspa{15}}{6}{{\color{white}\act}}{\nrestart}{\carry}{{\color{white}\one}}\sat{\symbI{3}}{5}{{\color{white}\act}}{\nrestart}{{\color{white}\carry}}{\one}\sat{\nspa{29}}{6}{{\color{white}\act}}{\nrestart}{\carry}{{\color{white}\one}}\sat{\symbL}{3}{{\color{white}\act}}{{\color{white}\nrestart}}{\carry}{\one}\sat{\szero}{4}{{\color{white}\act}}{\nrestart}{{\color{white}\carry}}{{\color{white}\one}}\sat{\symbR{2}{}}{5}{{\color{white}\act}}{\nrestart}{{\color{white}\carry}}{\one} \sat{\nspa{8}}{6}{{\color{white}\act}}{\nrestart}{\carry}{{\color{white}\one}} }\\
  &\xline{ \sat{\nsymbpp{4}}{14}{\act}{\nrestart}{\carry}{{\color{white}\one}}\sat{\nsmin{10}}{2}{{\color{white}\act}}{{\color{white}\nrestart}}{\carry}{{\color{white}\one}}\sat{\nspa{71}}{12}{\act}{\nrestart}{{\color{white}\carry}}{{\color{white}\one}}\sat{\symbB}{3}{{\color{white}\act}}{{\color{white}\nrestart}}{\carry}{\one}\sat{\nspa{15}}{4}{{\color{white}\act}}{\nrestart}{{\color{white}\carry}}{{\color{white}\one}}\sat{\symbd}{3}{{\color{white}\act}}{{\color{white}\nrestart}}{\carry}{\one}\sat{\nspa{15}}{4}{{\color{white}\act}}{\nrestart}{{\color{white}\carry}}{{\color{white}\one}}\sat{\symbX}{3}{{\color{white}\act}}{{\color{white}\nrestart}}{\carry}{\one}\sat{\nspa{15}}{4}{{\color{white}\act}}{\nrestart}{{\color{white}\carry}}{{\color{white}\one}}\sat{\symbo}{3}{{\color{white}\act}}{{\color{white}\nrestart}}{\carry}{\one}\sat{\nspa{15}}{4}{{\color{white}\act}}{\nrestart}{{\color{white}\carry}}{{\color{white}\one}}\sat{\symbO}{3}{{\color{white}\act}}{{\color{white}\nrestart}}{\carry}{\one}\sat{\nspa{15}}{4}{{\color{white}\act}}{\nrestart}{{\color{white}\carry}}{{\color{white}\one}}\sat{\symbI{2}}{3}{{\color{white}\act}}{{\color{white}\nrestart}}{\carry}{\one}\sat{\nspa{29}}{4}{{\color{white}\act}}{\nrestart}{{\color{white}\carry}}{{\color{white}\one}}\sat{\symbL}{1}{{\color{white}\act}}{{\color{white}\nrestart}}{{\color{white}\carry}}{\one}\sat{\szero}{2}{{\color{white}\act}}{{\color{white}\nrestart}}{\carry}{{\color{white}\one}}\sat{\symbR{2}{}}{3}{{\color{white}\act}}{{\color{white}\nrestart}}{\carry}{\one} \sat{\nspa{8}}{4}{{\color{white}\act}}{\nrestart}{{\color{white}\carry}}{{\color{white}\one}} }\\
  &\xline{ \sat{(\nspa{2}\smin)^{4}}{12}{\act}{\nrestart}{{\color{white}\carry}}{{\color{white}\one}}\sat{\nspa{231}}{8}{\act}{{\color{white}\nrestart}}{{\color{white}\carry}}{{\color{white}\one}}\sat{\symbb}{15}{\act}{\nrestart}{\carry}{\one}\sat{\nspa{15}}{0}{{\color{white}\act}}{{\color{white}\nrestart}}{{\color{white}\carry}}{{\color{white}\one}}\sat{\symbc}{15}{\act}{\nrestart}{\carry}{\one}\sat{\nspa{23}}{0}{{\color{white}\act}}{{\color{white}\nrestart}}{{\color{white}\carry}}{{\color{white}\one}}\sat{\symbX}{7}{{\color{white}\act}}{\nrestart}{\carry}{\one}\sat{\nspa{15}}{8}{\act}{{\color{white}\nrestart}}{{\color{white}\carry}}{{\color{white}\one}}\sat{\symbo}{7}{{\color{white}\act}}{\nrestart}{\carry}{\one}\sat{\nspa{15}}{8}{\act}{{\color{white}\nrestart}}{{\color{white}\carry}}{{\color{white}\one}}\sat{\symbO}{7}{{\color{white}\act}}{\nrestart}{\carry}{\one}\sat{\nspa{15}}{8}{\act}{{\color{white}\nrestart}}{{\color{white}\carry}}{{\color{white}\one}}\sat{\symbI{1}}{7}{{\color{white}\act}}{\nrestart}{\carry}{\one}\sat{\nspa{29}}{8}{\act}{{\color{white}\nrestart}}{{\color{white}\carry}}{{\color{white}\one}}\sat{\symbL}{5}{{\color{white}\act}}{\nrestart}{{\color{white}\carry}}{\one}\sat{\szero}{6}{{\color{white}\act}}{\nrestart}{\carry}{{\color{white}\one}}\sat{\symbR{1}{}}{7}{{\color{white}\act}}{\nrestart}{\carry}{\one} \sat{\nspa{8}}{8}{\act}{{\color{white}\nrestart}}{{\color{white}\carry}}{{\color{white}\one}} }
\end{align*}
For compactness, we continue without displaying the symbols~$\spa$.
The length $n$ of blocks $\nspa{n}$ matters only modulo $16$, 
and can be deduced from the morphism indexes of the surrounding letters.
\begin{align*}
  &\xline{ \sat{\symbA}{15}{\act}{\nrestart}{\carry}{\one}\sat{\symbd}{15}{\act}{\nrestart}{\carry}{\one}\sat{\symbX}{15}{\act}{\nrestart}{\carry}{\one}\sat{\symbo}{15}{\act}{\nrestart}{\carry}{\one}\sat{\symbO}{15}{\act}{\nrestart}{\carry}{\one}\sat{\symbIns}{15}{\act}{\nrestart}{\carry}{\one}\sat{\symbL}{13}{\act}{\nrestart}{{\color{white}\carry}}{\one}\sat{\szero}{14}{\act}{\nrestart}{\carry}{{\color{white}\one}}\sat{\symbRns}{15}{\act}{\nrestart}{\carry}{\one} }\\
  &\xline{ \sat{\symbA}{7}{{\color{white}\act}}{\nrestart}{\carry}{\one}\sat{\symbd}{7}{{\color{white}\act}}{\nrestart}{\carry}{\one}\sat{\symba}{7}{{\color{white}\act}}{\nrestart}{\carry}{\one}\sat{\symbd}{7}{{\color{white}\act}}{\nrestart}{\carry}{\one}\sat{\symbd}{7}{{\color{white}\act}}{\nrestart}{\carry}{\one}\sat{\symbX}{7}{{\color{white}\act}}{\nrestart}{\carry}{\one}\sat{\symbO}{7}{{\color{white}\act}}{\nrestart}{\carry}{\one}\sat{\symbI{3}}{7}{{\color{white}\act}}{\nrestart}{\carry}{\one}\sat{\symbL}{4}{{\color{white}\act}}{\nrestart}{{\color{white}\carry}}{{\color{white}\one}}\sat{\szero}{5}{{\color{white}\act}}{\nrestart}{{\color{white}\carry}}{\one}\sat{\sone}{6}{{\color{white}\act}}{\nrestart}{\carry}{{\color{white}\one}}\sat{\symbR{3}}{7}{{\color{white}\act}}{\nrestart}{\carry}{\one} }\\
  &\xline{ \sat{\nsymbpp{4}}{0}{{\color{white}\act}}{{\color{white}\nrestart}}{{\color{white}\carry}}{{\color{white}\one}}\sat{\nsmin{8}}{4}{{\color{white}\act}}{\nrestart}{{\color{white}\carry}}{{\color{white}\one}}\sat{\symbA}{3}{{\color{white}\act}}{{\color{white}\nrestart}}{\carry}{\one}\sat{\symbd}{3}{{\color{white}\act}}{{\color{white}\nrestart}}{\carry}{\one}\sat{\symba}{3}{{\color{white}\act}}{{\color{white}\nrestart}}{\carry}{\one}\sat{\symbd}{3}{{\color{white}\act}}{{\color{white}\nrestart}}{\carry}{\one}\sat{\symbd}{3}{{\color{white}\act}}{{\color{white}\nrestart}}{\carry}{\one}\sat{\symbX}{3}{{\color{white}\act}}{{\color{white}\nrestart}}{\carry}{\one}\sat{\symbO}{3}{{\color{white}\act}}{{\color{white}\nrestart}}{\carry}{\one}\sat{\symbI{2}}{3}{{\color{white}\act}}{{\color{white}\nrestart}}{\carry}{\one}\sat{\symbL}{0}{{\color{white}\act}}{{\color{white}\nrestart}}{{\color{white}\carry}}{{\color{white}\one}}\sat{\szero}{1}{{\color{white}\act}}{{\color{white}\nrestart}}{{\color{white}\carry}}{\one}\sat{\sone}{2}{{\color{white}\act}}{{\color{white}\nrestart}}{\carry}{{\color{white}\one}}\sat{\symbR{2}}{3}{{\color{white}\act}}{{\color{white}\nrestart}}{\carry}{\one} }\\
  &\xline{ \sat{\nsmin{4}}{14}{\act}{\nrestart}{\carry}{{\color{white}\one}}\sat{\symba}{15}{\act}{\nrestart}{\carry}{\one}\sat{\symbc}{15}{\act}{\nrestart}{\carry}{\one}\sat{\symba}{7}{{\color{white}\act}}{\nrestart}{\carry}{\one}\sat{\symbc}{7}{{\color{white}\act}}{\nrestart}{\carry}{\one}\sat{\symbc}{15}{\act}{\nrestart}{\carry}{\one}\sat{\symbX}{7}{{\color{white}\act}}{\nrestart}{\carry}{\one}\sat{\symbO}{7}{{\color{white}\act}}{\nrestart}{\carry}{\one}\sat{\symbI{1}}{7}{{\color{white}\act}}{\nrestart}{\carry}{\one}\sat{\symbL}{4}{{\color{white}\act}}{\nrestart}{{\color{white}\carry}}{{\color{white}\one}}\sat{\szero}{5}{{\color{white}\act}}{\nrestart}{{\color{white}\carry}}{\one}\sat{\sone}{6}{{\color{white}\act}}{\nrestart}{\carry}{{\color{white}\one}}\sat{\symbR{1}}{7}{{\color{white}\act}}{\nrestart}{\carry}{\one} }\\
  &\xline{ \sat{\nsymbpp{}}{13}{\act}{\nrestart}{{\color{white}\carry}}{\one}\sat{\nsplu{2}}{10}{\act}{{\color{white}\nrestart}}{\carry}{{\color{white}\one}}\sat{\symbB}{12}{\act}{\nrestart}{{\color{white}\carry}}{{\color{white}\one}}\sat{\symbd}{12}{\act}{\nrestart}{{\color{white}\carry}}{{\color{white}\one}}\sat{\symba}{12}{\act}{\nrestart}{{\color{white}\carry}}{{\color{white}\one}}\sat{\symbc}{12}{\act}{\nrestart}{{\color{white}\carry}}{{\color{white}\one}}\sat{\symbd}{4}{{\color{white}\act}}{\nrestart}{{\color{white}\carry}}{{\color{white}\one}}\sat{\symbX}{4}{{\color{white}\act}}{\nrestart}{{\color{white}\carry}}{{\color{white}\one}}\sat{\symbO}{4}{{\color{white}\act}}{\nrestart}{{\color{white}\carry}}{{\color{white}\one}}\sat{\symbIns}{4}{{\color{white}\act}}{\nrestart}{{\color{white}\carry}}{{\color{white}\one}}\sat{\symbL}{1}{{\color{white}\act}}{{\color{white}\nrestart}}{{\color{white}\carry}}{\one}\sat{\szero}{2}{{\color{white}\act}}{{\color{white}\nrestart}}{\carry}{{\color{white}\one}}\sat{\sone}{3}{{\color{white}\act}}{{\color{white}\nrestart}}{\carry}{\one}\sat{\symbRns}{4}{{\color{white}\act}}{\nrestart}{{\color{white}\carry}}{{\color{white}\one}} }\\
  &\xline{ \sat{\nsmin{}}{12}{\act}{\nrestart}{{\color{white}\carry}}{{\color{white}\one}}\sat{\symbB}{13}{\act}{\nrestart}{{\color{white}\carry}}{\one}\sat{\symbd}{13}{\act}{\nrestart}{{\color{white}\carry}}{\one}\sat{\nsymbpp{}}{11}{\act}{{\color{white}\nrestart}}{\carry}{\one}\sat{\symbA}{12}{\act}{\nrestart}{{\color{white}\carry}}{{\color{white}\one}}\sat{\symbd}{12}{\act}{\nrestart}{{\color{white}\carry}}{{\color{white}\one}}\sat{\symbd}{12}{\act}{\nrestart}{{\color{white}\carry}}{{\color{white}\one}}\sat{\symbX}{12}{\act}{\nrestart}{{\color{white}\carry}}{{\color{white}\one}}\sat{\symbO}{12}{\act}{\nrestart}{{\color{white}\carry}}{{\color{white}\one}}\sat{\symbIns}{12}{\act}{\nrestart}{{\color{white}\carry}}{{\color{white}\one}}\sat{\symbL}{8}{\act}{{\color{white}\nrestart}}{{\color{white}\carry}}{{\color{white}\one}}\sat{\szero}{9}{\act}{{\color{white}\nrestart}}{{\color{white}\carry}}{\one}\sat{\sone}{10}{\act}{{\color{white}\nrestart}}{\carry}{{\color{white}\one}}\sat{\szero}{11}{\act}{{\color{white}\nrestart}}{\carry}{\one}\sat{\symbRns}{12}{\act}{\nrestart}{{\color{white}\carry}}{{\color{white}\one}} }\\
  &\xline{ \sat{\symbB}{4}{{\color{white}\act}}{\nrestart}{{\color{white}\carry}}{{\color{white}\one}}\sat{\symbd}{4}{{\color{white}\act}}{\nrestart}{{\color{white}\carry}}{{\color{white}\one}}\sat{\nsmin{}}{4}{{\color{white}\act}}{\nrestart}{{\color{white}\carry}}{{\color{white}\one}}\sat{\symbA}{5}{{\color{white}\act}}{\nrestart}{{\color{white}\carry}}{\one}\sat{\symbd}{5}{{\color{white}\act}}{\nrestart}{{\color{white}\carry}}{\one}\sat{\symbd}{5}{{\color{white}\act}}{\nrestart}{{\color{white}\carry}}{\one}\sat{\symbX}{5}{{\color{white}\act}}{\nrestart}{{\color{white}\carry}}{\one}\sat{\symbo}{5}{{\color{white}\act}}{\nrestart}{{\color{white}\carry}}{\one}\sat{\symbO}{5}{{\color{white}\act}}{\nrestart}{{\color{white}\carry}}{\one}\sat{\symbI{3}}{5}{{\color{white}\act}}{\nrestart}{{\color{white}\carry}}{\one}\sat{\symbL}{0}{{\color{white}\act}}{{\color{white}\nrestart}}{{\color{white}\carry}}{{\color{white}\one}}\sat{\szero}{1}{{\color{white}\act}}{{\color{white}\nrestart}}{{\color{white}\carry}}{\one}\sat{\sone}{2}{{\color{white}\act}}{{\color{white}\nrestart}}{\carry}{{\color{white}\one}}\sat{\szero}{3}{{\color{white}\act}}{{\color{white}\nrestart}}{\carry}{\one}\sat{\szero}{4}{{\color{white}\act}}{\nrestart}{{\color{white}\carry}}{{\color{white}\one}}\sat{\symbR{3}}{5}{{\color{white}\act}}{\nrestart}{{\color{white}\carry}}{\one} }\\
  &\xline{ \sat{\nsymbpp{4}}{14}{\act}{\nrestart}{\carry}{{\color{white}\one}}\sat{\nsmin{10}}{2}{{\color{white}\act}}{{\color{white}\nrestart}}{\carry}{{\color{white}\one}}\sat{\symbB}{3}{{\color{white}\act}}{{\color{white}\nrestart}}{\carry}{\one}\sat{\symbd}{3}{{\color{white}\act}}{{\color{white}\nrestart}}{\carry}{\one}\sat{\symbA}{3}{{\color{white}\act}}{{\color{white}\nrestart}}{\carry}{\one}\sat{\symbd}{3}{{\color{white}\act}}{{\color{white}\nrestart}}{\carry}{\one}\sat{\symbd}{3}{{\color{white}\act}}{{\color{white}\nrestart}}{\carry}{\one}\sat{\symbX}{3}{{\color{white}\act}}{{\color{white}\nrestart}}{\carry}{\one}\sat{\symbo}{3}{{\color{white}\act}}{{\color{white}\nrestart}}{\carry}{\one}\sat{\symbO}{3}{{\color{white}\act}}{{\color{white}\nrestart}}{\carry}{\one}\sat{\symbI{2}}{3}{{\color{white}\act}}{{\color{white}\nrestart}}{\carry}{\one}\sat{\symbL}{14}{\act}{\nrestart}{\carry}{{\color{white}\one}}\sat{\szero}{15}{\act}{\nrestart}{\carry}{\one}\sat{\sone}{0}{{\color{white}\act}}{{\color{white}\nrestart}}{{\color{white}\carry}}{{\color{white}\one}}\sat{\szero}{1}{{\color{white}\act}}{{\color{white}\nrestart}}{{\color{white}\carry}}{\one}\sat{\szero}{2}{{\color{white}\act}}{{\color{white}\nrestart}}{\carry}{{\color{white}\one}}\sat{\symbR{2}}{3}{{\color{white}\act}}{{\color{white}\nrestart}}{\carry}{\one} }\\
  &\xline{ \sat{\nsmin{4}}{14}{\act}{\nrestart}{\carry}{{\color{white}\one}}\sat{\symbb}{15}{\act}{\nrestart}{\carry}{\one}\sat{\symbc}{15}{\act}{\nrestart}{\carry}{\one}\sat{\symba}{7}{{\color{white}\act}}{\nrestart}{\carry}{\one}\sat{\symbc}{7}{{\color{white}\act}}{\nrestart}{\carry}{\one}\sat{\symbc}{15}{\act}{\nrestart}{\carry}{\one}\sat{\symbX}{7}{{\color{white}\act}}{\nrestart}{\carry}{\one}\sat{\symbo}{7}{{\color{white}\act}}{\nrestart}{\carry}{\one}\sat{\symbO}{7}{{\color{white}\act}}{\nrestart}{\carry}{\one}\sat{\symbI{1}}{7}{{\color{white}\act}}{\nrestart}{\carry}{\one}\sat{\symbL}{2}{{\color{white}\act}}{{\color{white}\nrestart}}{\carry}{{\color{white}\one}}\sat{\szero}{3}{{\color{white}\act}}{{\color{white}\nrestart}}{\carry}{\one}\sat{\sone}{4}{{\color{white}\act}}{\nrestart}{{\color{white}\carry}}{{\color{white}\one}}\sat{\szero}{5}{{\color{white}\act}}{\nrestart}{{\color{white}\carry}}{\one}\sat{\szero}{6}{{\color{white}\act}}{\nrestart}{\carry}{{\color{white}\one}}\sat{\symbR{1}}{7}{{\color{white}\act}}{\nrestart}{\carry}{\one} }\\
  &\xline{ \sat{\symbA}{15}{\act}{\nrestart}{\carry}{\one}\sat{\symbd}{15}{\act}{\nrestart}{\carry}{\one}\sat{\symba}{15}{\act}{\nrestart}{\carry}{\one}\sat{\symbc}{15}{\act}{\nrestart}{\carry}{\one}\sat{\symbd}{7}{{\color{white}\act}}{\nrestart}{\carry}{\one}\sat{\symbX}{7}{{\color{white}\act}}{\nrestart}{\carry}{\one}\sat{\symbo}{7}{{\color{white}\act}}{\nrestart}{\carry}{\one}\sat{\symbO}{7}{{\color{white}\act}}{\nrestart}{\carry}{\one}\sat{\symbIns}{7}{{\color{white}\act}}{\nrestart}{\carry}{\one}\sat{\symbL}{2}{{\color{white}\act}}{{\color{white}\nrestart}}{\carry}{{\color{white}\one}}\sat{\szero}{3}{{\color{white}\act}}{{\color{white}\nrestart}}{\carry}{\one}\sat{\sone}{4}{{\color{white}\act}}{\nrestart}{{\color{white}\carry}}{{\color{white}\one}}\sat{\szero}{5}{{\color{white}\act}}{\nrestart}{{\color{white}\carry}}{\one}\sat{\szero}{6}{{\color{white}\act}}{\nrestart}{\carry}{{\color{white}\one}}\sat{\symbRns}{7}{{\color{white}\act}}{\nrestart}{\carry}{\one} }\\
  &\xline{ \sat{\symbA}{15}{\act}{\nrestart}{\carry}{\one}\sat{\symbd}{15}{\act}{\nrestart}{\carry}{\one}\sat{\nsymbpp{}}{13}{\act}{\nrestart}{{\color{white}\carry}}{\one}\sat{\nsplu{2}}{10}{\act}{{\color{white}\nrestart}}{\carry}{{\color{white}\one}}\sat{\symbB}{12}{\act}{\nrestart}{{\color{white}\carry}}{{\color{white}\one}}\sat{\symbd}{12}{\act}{\nrestart}{{\color{white}\carry}}{{\color{white}\one}}\sat{\symbd}{12}{\act}{\nrestart}{{\color{white}\carry}}{{\color{white}\one}}\sat{\symbX}{12}{\act}{\nrestart}{{\color{white}\carry}}{{\color{white}\one}}\sat{\symbo}{12}{\act}{\nrestart}{{\color{white}\carry}}{{\color{white}\one}}\sat{\symbO}{12}{\act}{\nrestart}{{\color{white}\carry}}{{\color{white}\one}}\sat{\symbIns}{12}{\act}{\nrestart}{{\color{white}\carry}}{{\color{white}\one}}\sat{\symbL}{6}{{\color{white}\act}}{\nrestart}{\carry}{{\color{white}\one}}\sat{\szero}{7}{{\color{white}\act}}{\nrestart}{\carry}{\one}\sat{\sone}{8}{\act}{{\color{white}\nrestart}}{{\color{white}\carry}}{{\color{white}\one}}\sat{\szero}{9}{\act}{{\color{white}\nrestart}}{{\color{white}\carry}}{\one}\sat{\szero}{10}{\act}{{\color{white}\nrestart}}{\carry}{{\color{white}\one}}\sat{\sone}{11}{\act}{{\color{white}\nrestart}}{\carry}{\one}\sat{\symbRns}{12}{\act}{\nrestart}{{\color{white}\carry}}{{\color{white}\one}} }\\
  &\xline{ \sat{\symbA}{4}{{\color{white}\act}}{\nrestart}{{\color{white}\carry}}{{\color{white}\one}}\sat{\symbd}{4}{{\color{white}\act}}{\nrestart}{{\color{white}\carry}}{{\color{white}\one}}\sat{\nsmin{}}{4}{{\color{white}\act}}{\nrestart}{{\color{white}\carry}}{{\color{white}\one}}\sat{\symbB}{5}{{\color{white}\act}}{\nrestart}{{\color{white}\carry}}{\one}\sat{\symbd}{5}{{\color{white}\act}}{\nrestart}{{\color{white}\carry}}{\one}\sat{\symbd}{5}{{\color{white}\act}}{\nrestart}{{\color{white}\carry}}{\one}\sat{\symbX}{5}{{\color{white}\act}}{\nrestart}{{\color{white}\carry}}{\one}\sat{\symbo}{5}{{\color{white}\act}}{\nrestart}{{\color{white}\carry}}{\one}\sat{\symbo}{5}{{\color{white}\act}}{\nrestart}{{\color{white}\carry}}{\one}\sat{\symbO}{5}{{\color{white}\act}}{\nrestart}{{\color{white}\carry}}{\one}\sat{\symbI{3}}{5}{{\color{white}\act}}{\nrestart}{{\color{white}\carry}}{\one}\sat{\symbL}{14}{\act}{\nrestart}{\carry}{{\color{white}\one}}\sat{\szero}{15}{\act}{\nrestart}{\carry}{\one}\sat{\sone}{0}{{\color{white}\act}}{{\color{white}\nrestart}}{{\color{white}\carry}}{{\color{white}\one}}\sat{\szero}{1}{{\color{white}\act}}{{\color{white}\nrestart}}{{\color{white}\carry}}{\one}\sat{\szero}{2}{{\color{white}\act}}{{\color{white}\nrestart}}{\carry}{{\color{white}\one}}\sat{\sone}{3}{{\color{white}\act}}{{\color{white}\nrestart}}{\carry}{\one}\sat{\szero}{4}{{\color{white}\act}}{\nrestart}{{\color{white}\carry}}{{\color{white}\one}}\sat{\symbR{3}}{5}{{\color{white}\act}}{\nrestart}{{\color{white}\carry}}{\one} }\\
  &\xline{ \sat{\nsymbpp{4}}{14}{\act}{\nrestart}{\carry}{{\color{white}\one}}\sat{\nsmin{10}}{2}{{\color{white}\act}}{{\color{white}\nrestart}}{\carry}{{\color{white}\one}}\sat{\symbA}{3}{{\color{white}\act}}{{\color{white}\nrestart}}{\carry}{\one}\sat{\symbd}{3}{{\color{white}\act}}{{\color{white}\nrestart}}{\carry}{\one}\sat{\symbB}{3}{{\color{white}\act}}{{\color{white}\nrestart}}{\carry}{\one}\sat{\symbd}{3}{{\color{white}\act}}{{\color{white}\nrestart}}{\carry}{\one}\sat{\symbd}{3}{{\color{white}\act}}{{\color{white}\nrestart}}{\carry}{\one}\sat{\symbX}{3}{{\color{white}\act}}{{\color{white}\nrestart}}{\carry}{\one}\sat{\symbo}{3}{{\color{white}\act}}{{\color{white}\nrestart}}{\carry}{\one}\sat{\symbo}{3}{{\color{white}\act}}{{\color{white}\nrestart}}{\carry}{\one}\sat{\symbO}{3}{{\color{white}\act}}{{\color{white}\nrestart}}{\carry}{\one}\sat{\symbI{2}}{3}{{\color{white}\act}}{{\color{white}\nrestart}}{\carry}{\one}\sat{\symbL}{12}{\act}{\nrestart}{{\color{white}\carry}}{{\color{white}\one}}\sat{\szero}{13}{\act}{\nrestart}{{\color{white}\carry}}{\one}\sat{\sone}{14}{\act}{\nrestart}{\carry}{{\color{white}\one}}\sat{\szero}{15}{\act}{\nrestart}{\carry}{\one}\sat{\szero}{0}{{\color{white}\act}}{{\color{white}\nrestart}}{{\color{white}\carry}}{{\color{white}\one}}\sat{\sone}{1}{{\color{white}\act}}{{\color{white}\nrestart}}{{\color{white}\carry}}{\one}\sat{\szero}{2}{{\color{white}\act}}{{\color{white}\nrestart}}{\carry}{{\color{white}\one}}\sat{\symbR{2}}{3}{{\color{white}\act}}{{\color{white}\nrestart}}{\carry}{\one} }\\
  &\xline{ \sat{\nsmin{4}}{14}{\act}{\nrestart}{\carry}{{\color{white}\one}}\sat{\symba}{15}{\act}{\nrestart}{\carry}{\one}\sat{\symbc}{15}{\act}{\nrestart}{\carry}{\one}\sat{\symbb}{7}{{\color{white}\act}}{\nrestart}{\carry}{\one}\sat{\symbc}{7}{{\color{white}\act}}{\nrestart}{\carry}{\one}\sat{\symbc}{15}{\act}{\nrestart}{\carry}{\one}\sat{\symbX}{7}{{\color{white}\act}}{\nrestart}{\carry}{\one}\sat{\symbo}{7}{{\color{white}\act}}{\nrestart}{\carry}{\one}\sat{\symbo}{7}{{\color{white}\act}}{\nrestart}{\carry}{\one}\sat{\symbO}{7}{{\color{white}\act}}{\nrestart}{\carry}{\one}\sat{\symbI{1}}{7}{{\color{white}\act}}{\nrestart}{\carry}{\one}\sat{\symbL}{0}{{\color{white}\act}}{{\color{white}\nrestart}}{{\color{white}\carry}}{{\color{white}\one}}\sat{\szero}{1}{{\color{white}\act}}{{\color{white}\nrestart}}{{\color{white}\carry}}{\one}\sat{\sone}{2}{{\color{white}\act}}{{\color{white}\nrestart}}{\carry}{{\color{white}\one}}\sat{\szero}{3}{{\color{white}\act}}{{\color{white}\nrestart}}{\carry}{\one}\sat{\szero}{4}{{\color{white}\act}}{\nrestart}{{\color{white}\carry}}{{\color{white}\one}}\sat{\sone}{5}{{\color{white}\act}}{\nrestart}{{\color{white}\carry}}{\one}\sat{\szero}{6}{{\color{white}\act}}{\nrestart}{\carry}{{\color{white}\one}}\sat{\symbR{1}}{7}{{\color{white}\act}}{\nrestart}{\carry}{\one} }\\
  &\xline{ \sat{\nsymbpp{}}{13}{\act}{\nrestart}{{\color{white}\carry}}{\one}\sat{\nsplu{2}}{10}{\act}{{\color{white}\nrestart}}{\carry}{{\color{white}\one}}\sat{\symbB}{12}{\act}{\nrestart}{{\color{white}\carry}}{{\color{white}\one}}\sat{\symbd}{12}{\act}{\nrestart}{{\color{white}\carry}}{{\color{white}\one}}\sat{\symbb}{12}{\act}{\nrestart}{{\color{white}\carry}}{{\color{white}\one}}\sat{\symbc}{12}{\act}{\nrestart}{{\color{white}\carry}}{{\color{white}\one}}\sat{\symbd}{4}{{\color{white}\act}}{\nrestart}{{\color{white}\carry}}{{\color{white}\one}}\sat{\symbX}{4}{{\color{white}\act}}{\nrestart}{{\color{white}\carry}}{{\color{white}\one}}\sat{\symbo}{4}{{\color{white}\act}}{\nrestart}{{\color{white}\carry}}{{\color{white}\one}}\sat{\symbo}{4}{{\color{white}\act}}{\nrestart}{{\color{white}\carry}}{{\color{white}\one}}\sat{\symbO}{4}{{\color{white}\act}}{\nrestart}{{\color{white}\carry}}{{\color{white}\one}}\sat{\symbIns}{4}{{\color{white}\act}}{\nrestart}{{\color{white}\carry}}{{\color{white}\one}}\sat{\symbL}{13}{\act}{\nrestart}{{\color{white}\carry}}{\one}\sat{\szero}{14}{\act}{\nrestart}{\carry}{{\color{white}\one}}\sat{\sone}{15}{\act}{\nrestart}{\carry}{\one}\sat{\szero}{0}{{\color{white}\act}}{{\color{white}\nrestart}}{{\color{white}\carry}}{{\color{white}\one}}\sat{\szero}{1}{{\color{white}\act}}{{\color{white}\nrestart}}{{\color{white}\carry}}{\one}\sat{\sone}{2}{{\color{white}\act}}{{\color{white}\nrestart}}{\carry}{{\color{white}\one}}\sat{\szero}{3}{{\color{white}\act}}{{\color{white}\nrestart}}{\carry}{\one}\sat{\symbRns}{4}{{\color{white}\act}}{\nrestart}{{\color{white}\carry}}{{\color{white}\one}} }\\
  &\xline{ \sat{\nsmin{}}{12}{\act}{\nrestart}{{\color{white}\carry}}{{\color{white}\one}}\sat{\symbB}{13}{\act}{\nrestart}{{\color{white}\carry}}{\one}\sat{\symbd}{13}{\act}{\nrestart}{{\color{white}\carry}}{\one}\sat{\symbB}{13}{\act}{\nrestart}{{\color{white}\carry}}{\one}\sat{\symbd}{13}{\act}{\nrestart}{{\color{white}\carry}}{\one}\sat{\symbd}{13}{\act}{\nrestart}{{\color{white}\carry}}{\one}\sat{\symbX}{13}{\act}{\nrestart}{{\color{white}\carry}}{\one}\sat{\symbo}{13}{\act}{\nrestart}{{\color{white}\carry}}{\one}\sat{\symbo}{13}{\act}{\nrestart}{{\color{white}\carry}}{\one}\sat{\symbO}{13}{\act}{\nrestart}{{\color{white}\carry}}{\one}\sat{\symbIns}{13}{\act}{\nrestart}{{\color{white}\carry}}{\one}\sat{\symbL}{5}{{\color{white}\act}}{\nrestart}{{\color{white}\carry}}{\one}\sat{\szero}{6}{{\color{white}\act}}{\nrestart}{\carry}{{\color{white}\one}}\sat{\sone}{7}{{\color{white}\act}}{\nrestart}{\carry}{\one}\sat{\szero}{8}{\act}{{\color{white}\nrestart}}{{\color{white}\carry}}{{\color{white}\one}}\sat{\szero}{9}{\act}{{\color{white}\nrestart}}{{\color{white}\carry}}{\one}\sat{\sone}{10}{\act}{{\color{white}\nrestart}}{\carry}{{\color{white}\one}}\sat{\szero}{11}{\act}{{\color{white}\nrestart}}{\carry}{\one}\sat{\szero}{12}{\act}{\nrestart}{{\color{white}\carry}}{{\color{white}\one}}\sat{\symbRns}{13}{\act}{\nrestart}{{\color{white}\carry}}{\one} }\\
  &\xline{ \sat{\symbB}{5}{{\color{white}\act}}{\nrestart}{{\color{white}\carry}}{\one}\sat{\symbd}{5}{{\color{white}\act}}{\nrestart}{{\color{white}\carry}}{\one}\sat{\symbB}{5}{{\color{white}\act}}{\nrestart}{{\color{white}\carry}}{\one}\sat{\symbd}{5}{{\color{white}\act}}{\nrestart}{{\color{white}\carry}}{\one}\sat{\symbd}{5}{{\color{white}\act}}{\nrestart}{{\color{white}\carry}}{\one}\sat{\symbX}{5}{{\color{white}\act}}{\nrestart}{{\color{white}\carry}}{\one}\sat{\symbo}{5}{{\color{white}\act}}{\nrestart}{{\color{white}\carry}}{\one}\sat{\symbo}{5}{{\color{white}\act}}{\nrestart}{{\color{white}\carry}}{\one}\sat{\symbo}{5}{{\color{white}\act}}{\nrestart}{{\color{white}\carry}}{\one}\sat{\symbO}{5}{{\color{white}\act}}{\nrestart}{{\color{white}\carry}}{\one}\sat{\symbI{3}}{5}{{\color{white}\act}}{\nrestart}{{\color{white}\carry}}{\one}\sat{\symbL}{12}{\act}{\nrestart}{{\color{white}\carry}}{{\color{white}\one}}\sat{\szero}{13}{\act}{\nrestart}{{\color{white}\carry}}{\one}\sat{\sone}{14}{\act}{\nrestart}{\carry}{{\color{white}\one}}\sat{\szero}{15}{\act}{\nrestart}{\carry}{\one}\sat{\szero}{0}{{\color{white}\act}}{{\color{white}\nrestart}}{{\color{white}\carry}}{{\color{white}\one}}\sat{\sone}{1}{{\color{white}\act}}{{\color{white}\nrestart}}{{\color{white}\carry}}{\one}\sat{\szero}{2}{{\color{white}\act}}{{\color{white}\nrestart}}{\carry}{{\color{white}\one}}\sat{\szero}{3}{{\color{white}\act}}{{\color{white}\nrestart}}{\carry}{\one}\sat{\sone}{4}{{\color{white}\act}}{\nrestart}{{\color{white}\carry}}{{\color{white}\one}}\sat{\symbR{3}}{5}{{\color{white}\act}}{\nrestart}{{\color{white}\carry}}{\one} }\\
  &\xline{ \sat{\nsymbpp{4}}{14}{\act}{\nrestart}{\carry}{{\color{white}\one}}\sat{\nsmin{10}}{2}{{\color{white}\act}}{{\color{white}\nrestart}}{\carry}{{\color{white}\one}}\sat{\symbB}{3}{{\color{white}\act}}{{\color{white}\nrestart}}{\carry}{\one}\sat{\symbd}{3}{{\color{white}\act}}{{\color{white}\nrestart}}{\carry}{\one}\sat{\symbB}{3}{{\color{white}\act}}{{\color{white}\nrestart}}{\carry}{\one}\sat{\symbd}{3}{{\color{white}\act}}{{\color{white}\nrestart}}{\carry}{\one}\sat{\symbd}{3}{{\color{white}\act}}{{\color{white}\nrestart}}{\carry}{\one}\sat{\symbX}{3}{{\color{white}\act}}{{\color{white}\nrestart}}{\carry}{\one}\sat{\symbo}{3}{{\color{white}\act}}{{\color{white}\nrestart}}{\carry}{\one}\sat{\symbo}{3}{{\color{white}\act}}{{\color{white}\nrestart}}{\carry}{\one}\sat{\symbo}{3}{{\color{white}\act}}{{\color{white}\nrestart}}{\carry}{\one}\sat{\symbO}{3}{{\color{white}\act}}{{\color{white}\nrestart}}{\carry}{\one}\sat{\symbI{2}}{3}{{\color{white}\act}}{{\color{white}\nrestart}}{\carry}{\one}\sat{\symbL}{10}{\act}{{\color{white}\nrestart}}{\carry}{{\color{white}\one}}\sat{\szero}{11}{\act}{{\color{white}\nrestart}}{\carry}{\one}\sat{\sone}{12}{\act}{\nrestart}{{\color{white}\carry}}{{\color{white}\one}}\sat{\szero}{13}{\act}{\nrestart}{{\color{white}\carry}}{\one}\sat{\szero}{14}{\act}{\nrestart}{\carry}{{\color{white}\one}}\sat{\sone}{15}{\act}{\nrestart}{\carry}{\one}\sat{\szero}{0}{{\color{white}\act}}{{\color{white}\nrestart}}{{\color{white}\carry}}{{\color{white}\one}}\sat{\szero}{1}{{\color{white}\act}}{{\color{white}\nrestart}}{{\color{white}\carry}}{\one}\sat{\sone}{2}{{\color{white}\act}}{{\color{white}\nrestart}}{\carry}{{\color{white}\one}}\sat{\symbR{2}}{3}{{\color{white}\act}}{{\color{white}\nrestart}}{\carry}{\one} }\\
  &\xline{ \sat{\nsmin{4}}{14}{\act}{\nrestart}{\carry}{{\color{white}\one}}\sat{\symbb}{15}{\act}{\nrestart}{\carry}{\one}\sat{\symbc}{15}{\act}{\nrestart}{\carry}{\one}\sat{\symbb}{7}{{\color{white}\act}}{\nrestart}{\carry}{\one}\sat{\symbc}{7}{{\color{white}\act}}{\nrestart}{\carry}{\one}\sat{\symbc}{15}{\act}{\nrestart}{\carry}{\one}\sat{\symbX}{7}{{\color{white}\act}}{\nrestart}{\carry}{\one}\sat{\symbo}{7}{{\color{white}\act}}{\nrestart}{\carry}{\one}\sat{\symbo}{7}{{\color{white}\act}}{\nrestart}{\carry}{\one}\sat{\symbo}{7}{{\color{white}\act}}{\nrestart}{\carry}{\one}\sat{\symbO}{7}{{\color{white}\act}}{\nrestart}{\carry}{\one}\sat{\symbI{1}}{7}{{\color{white}\act}}{\nrestart}{\carry}{\one}\sat{\symbL}{14}{\act}{\nrestart}{\carry}{{\color{white}\one}}\sat{\szero}{15}{\act}{\nrestart}{\carry}{\one}\sat{\sone}{0}{{\color{white}\act}}{{\color{white}\nrestart}}{{\color{white}\carry}}{{\color{white}\one}}\sat{\szero}{1}{{\color{white}\act}}{{\color{white}\nrestart}}{{\color{white}\carry}}{\one}\sat{\szero}{2}{{\color{white}\act}}{{\color{white}\nrestart}}{\carry}{{\color{white}\one}}\sat{\sone}{3}{{\color{white}\act}}{{\color{white}\nrestart}}{\carry}{\one}\sat{\szero}{4}{{\color{white}\act}}{\nrestart}{{\color{white}\carry}}{{\color{white}\one}}\sat{\szero}{5}{{\color{white}\act}}{\nrestart}{{\color{white}\carry}}{\one}\sat{\sone}{6}{{\color{white}\act}}{\nrestart}{\carry}{{\color{white}\one}}\sat{\symbR{1}}{7}{{\color{white}\act}}{\nrestart}{\carry}{\one} }\\
  &\xline{ \sat{\symbA}{15}{\act}{\nrestart}{\carry}{\one}\sat{\symbd}{15}{\act}{\nrestart}{\carry}{\one}\sat{\symbb}{15}{\act}{\nrestart}{\carry}{\one}\sat{\symbc}{15}{\act}{\nrestart}{\carry}{\one}\sat{\symbd}{7}{{\color{white}\act}}{\nrestart}{\carry}{\one}\sat{\symbX}{7}{{\color{white}\act}}{\nrestart}{\carry}{\one}\sat{\symbo}{7}{{\color{white}\act}}{\nrestart}{\carry}{\one}\sat{\symbo}{7}{{\color{white}\act}}{\nrestart}{\carry}{\one}\sat{\symbo}{7}{{\color{white}\act}}{\nrestart}{\carry}{\one}\sat{\symbO}{7}{{\color{white}\act}}{\nrestart}{\carry}{\one}\sat{\symbIns}{7}{{\color{white}\act}}{\nrestart}{\carry}{\one}\sat{\symbL}{14}{\act}{\nrestart}{\carry}{{\color{white}\one}}\sat{\szero}{15}{\act}{\nrestart}{\carry}{\one}\sat{\sone}{0}{{\color{white}\act}}{{\color{white}\nrestart}}{{\color{white}\carry}}{{\color{white}\one}}\sat{\szero}{1}{{\color{white}\act}}{{\color{white}\nrestart}}{{\color{white}\carry}}{\one}\sat{\szero}{2}{{\color{white}\act}}{{\color{white}\nrestart}}{\carry}{{\color{white}\one}}\sat{\sone}{3}{{\color{white}\act}}{{\color{white}\nrestart}}{\carry}{\one}\sat{\szero}{4}{{\color{white}\act}}{\nrestart}{{\color{white}\carry}}{{\color{white}\one}}\sat{\szero}{5}{{\color{white}\act}}{\nrestart}{{\color{white}\carry}}{\one}\sat{\sone}{6}{{\color{white}\act}}{\nrestart}{\carry}{{\color{white}\one}}\sat{\symbRns}{7}{{\color{white}\act}}{\nrestart}{\carry}{\one} }\\
  &\xline{ \sat{\symbA}{15}{\act}{\nrestart}{\carry}{\one}\sat{\symbd}{15}{\act}{\nrestart}{\carry}{\one}\sat{\symbA}{15}{\act}{\nrestart}{\carry}{\one}\sat{\symbd}{15}{\act}{\nrestart}{\carry}{\one}\sat{\symbd}{15}{\act}{\nrestart}{\carry}{\one}\sat{\symbX}{15}{\act}{\nrestart}{\carry}{\one}\sat{\symbo}{15}{\act}{\nrestart}{\carry}{\one}\sat{\symbo}{15}{\act}{\nrestart}{\carry}{\one}\sat{\symbo}{15}{\act}{\nrestart}{\carry}{\one}\sat{\symbO}{15}{\act}{\nrestart}{\carry}{\one}\sat{\symbIns}{15}{\act}{\nrestart}{\carry}{\one}\sat{\symbL}{5}{{\color{white}\act}}{\nrestart}{{\color{white}\carry}}{\one}\sat{\szero}{6}{{\color{white}\act}}{\nrestart}{\carry}{{\color{white}\one}}\sat{\sone}{7}{{\color{white}\act}}{\nrestart}{\carry}{\one}\sat{\szero}{8}{\act}{{\color{white}\nrestart}}{{\color{white}\carry}}{{\color{white}\one}}\sat{\szero}{9}{\act}{{\color{white}\nrestart}}{{\color{white}\carry}}{\one}\sat{\sone}{10}{\act}{{\color{white}\nrestart}}{\carry}{{\color{white}\one}}\sat{\szero}{11}{\act}{{\color{white}\nrestart}}{\carry}{\one}\sat{\szero}{12}{\act}{\nrestart}{{\color{white}\carry}}{{\color{white}\one}}\sat{\sone}{13}{\act}{\nrestart}{{\color{white}\carry}}{\one}\sat{\sone}{14}{\act}{\nrestart}{\carry}{{\color{white}\one}}\sat{\symbRns}{15}{\act}{\nrestart}{\carry}{\one} }\\
  &\xline{ \sat{\symbA}{7}{{\color{white}\act}}{\nrestart}{\carry}{\one}\sat{\symbd}{7}{{\color{white}\act}}{\nrestart}{\carry}{\one}\sat{\symbA}{7}{{\color{white}\act}}{\nrestart}{\carry}{\one}\sat{\symbd}{7}{{\color{white}\act}}{\nrestart}{\carry}{\one}\sat{\symbd}{7}{{\color{white}\act}}{\nrestart}{\carry}{\one}\sat{\symba}{7}{{\color{white}\act}}{\nrestart}{\carry}{\one}\sat{\symbd}{7}{{\color{white}\act}}{\nrestart}{\carry}{\one}\sat{\symbd}{7}{{\color{white}\act}}{\nrestart}{\carry}{\one}\sat{\symbd}{7}{{\color{white}\act}}{\nrestart}{\carry}{\one}\sat{\symbd}{7}{{\color{white}\act}}{\nrestart}{\carry}{\one}\sat{\symbX}{7}{{\color{white}\act}}{\nrestart}{\carry}{\one}\sat{\symbO}{7}{{\color{white}\act}}{\nrestart}{\carry}{\one}\sat{\symbI{3}}{7}{{\color{white}\act}}{\nrestart}{\carry}{\one}\sat{\symbL}{12}{\act}{\nrestart}{{\color{white}\carry}}{{\color{white}\one}}\sat{\szero}{13}{\act}{\nrestart}{{\color{white}\carry}}{\one}\sat{\sone}{14}{\act}{\nrestart}{\carry}{{\color{white}\one}}\sat{\szero}{15}{\act}{\nrestart}{\carry}{\one}\sat{\szero}{0}{{\color{white}\act}}{{\color{white}\nrestart}}{{\color{white}\carry}}{{\color{white}\one}}\sat{\sone}{1}{{\color{white}\act}}{{\color{white}\nrestart}}{{\color{white}\carry}}{\one}\sat{\szero}{2}{{\color{white}\act}}{{\color{white}\nrestart}}{\carry}{{\color{white}\one}}\sat{\szero}{3}{{\color{white}\act}}{{\color{white}\nrestart}}{\carry}{\one}\sat{\sone}{4}{{\color{white}\act}}{\nrestart}{{\color{white}\carry}}{{\color{white}\one}}\sat{\sone}{5}{{\color{white}\act}}{\nrestart}{{\color{white}\carry}}{\one}\sat{\sone}{6}{{\color{white}\act}}{\nrestart}{\carry}{{\color{white}\one}}\sat{\symbR{3}}{7}{{\color{white}\act}}{\nrestart}{\carry}{\one} }\\
  &\xline{ \sat{\nsymbpp{4}}{0}{{\color{white}\act}}{{\color{white}\nrestart}}{{\color{white}\carry}}{{\color{white}\one}}\sat{\nsmin{8}}{4}{{\color{white}\act}}{\nrestart}{{\color{white}\carry}}{{\color{white}\one}}\sat{\symbA}{3}{{\color{white}\act}}{{\color{white}\nrestart}}{\carry}{\one}\sat{\symbd}{3}{{\color{white}\act}}{{\color{white}\nrestart}}{\carry}{\one}\sat{\symbA}{3}{{\color{white}\act}}{{\color{white}\nrestart}}{\carry}{\one}\sat{\symbd}{3}{{\color{white}\act}}{{\color{white}\nrestart}}{\carry}{\one}\sat{\symbd}{3}{{\color{white}\act}}{{\color{white}\nrestart}}{\carry}{\one}\sat{\symba}{3}{{\color{white}\act}}{{\color{white}\nrestart}}{\carry}{\one}\sat{\symbd}{3}{{\color{white}\act}}{{\color{white}\nrestart}}{\carry}{\one}\sat{\symbd}{3}{{\color{white}\act}}{{\color{white}\nrestart}}{\carry}{\one}\sat{\symbd}{3}{{\color{white}\act}}{{\color{white}\nrestart}}{\carry}{\one}\sat{\symbd}{3}{{\color{white}\act}}{{\color{white}\nrestart}}{\carry}{\one}\sat{\symbX}{3}{{\color{white}\act}}{{\color{white}\nrestart}}{\carry}{\one}\sat{\symbO}{3}{{\color{white}\act}}{{\color{white}\nrestart}}{\carry}{\one}\sat{\symbI{2}}{3}{{\color{white}\act}}{{\color{white}\nrestart}}{\carry}{\one}\sat{\symbL}{8}{\act}{{\color{white}\nrestart}}{{\color{white}\carry}}{{\color{white}\one}}\sat{\szero}{9}{\act}{{\color{white}\nrestart}}{{\color{white}\carry}}{\one}\sat{\sone}{10}{\act}{{\color{white}\nrestart}}{\carry}{{\color{white}\one}}\sat{\szero}{11}{\act}{{\color{white}\nrestart}}{\carry}{\one}\sat{\szero}{12}{\act}{\nrestart}{{\color{white}\carry}}{{\color{white}\one}}\sat{\sone}{13}{\act}{\nrestart}{{\color{white}\carry}}{\one}\sat{\szero}{14}{\act}{\nrestart}{\carry}{{\color{white}\one}}\sat{\szero}{15}{\act}{\nrestart}{\carry}{\one}\sat{\sone}{0}{{\color{white}\act}}{{\color{white}\nrestart}}{{\color{white}\carry}}{{\color{white}\one}}\sat{\sone}{1}{{\color{white}\act}}{{\color{white}\nrestart}}{{\color{white}\carry}}{\one}\sat{\sone}{2}{{\color{white}\act}}{{\color{white}\nrestart}}{\carry}{{\color{white}\one}}\sat{\symbR{2}}{3}{{\color{white}\act}}{{\color{white}\nrestart}}{\carry}{\one} }\\
  &\xline{ \sat{\nsmin{4}}{14}{\act}{\nrestart}{\carry}{{\color{white}\one}}\sat{\symba}{15}{\act}{\nrestart}{\carry}{\one}\sat{\symbc}{15}{\act}{\nrestart}{\carry}{\one}\sat{\symba}{7}{{\color{white}\act}}{\nrestart}{\carry}{\one}\sat{\symbc}{7}{{\color{white}\act}}{\nrestart}{\carry}{\one}\sat{\symbc}{15}{\act}{\nrestart}{\carry}{\one}\sat{\symba}{7}{{\color{white}\act}}{\nrestart}{\carry}{\one}\sat{\symbc}{7}{{\color{white}\act}}{\nrestart}{\carry}{\one}\sat{\symbc}{15}{\act}{\nrestart}{\carry}{\one}\sat{\symbc}{7}{{\color{white}\act}}{\nrestart}{\carry}{\one}\sat{\symbc}{15}{\act}{\nrestart}{\carry}{\one}\sat{\symbX}{7}{{\color{white}\act}}{\nrestart}{\carry}{\one}\sat{\symbO}{7}{{\color{white}\act}}{\nrestart}{\carry}{\one}\sat{\symbI{1}}{7}{{\color{white}\act}}{\nrestart}{\carry}{\one}\sat{\symbL}{12}{\act}{\nrestart}{{\color{white}\carry}}{{\color{white}\one}}\sat{\szero}{13}{\act}{\nrestart}{{\color{white}\carry}}{\one}\sat{\sone}{14}{\act}{\nrestart}{\carry}{{\color{white}\one}}\sat{\szero}{15}{\act}{\nrestart}{\carry}{\one}\sat{\szero}{0}{{\color{white}\act}}{{\color{white}\nrestart}}{{\color{white}\carry}}{{\color{white}\one}}\sat{\sone}{1}{{\color{white}\act}}{{\color{white}\nrestart}}{{\color{white}\carry}}{\one}\sat{\szero}{2}{{\color{white}\act}}{{\color{white}\nrestart}}{\carry}{{\color{white}\one}}\sat{\szero}{3}{{\color{white}\act}}{{\color{white}\nrestart}}{\carry}{\one}\sat{\sone}{4}{{\color{white}\act}}{\nrestart}{{\color{white}\carry}}{{\color{white}\one}}\sat{\sone}{5}{{\color{white}\act}}{\nrestart}{{\color{white}\carry}}{\one}\sat{\sone}{6}{{\color{white}\act}}{\nrestart}{\carry}{{\color{white}\one}}\sat{\symbR{1}}{7}{{\color{white}\act}}{\nrestart}{\carry}{\one} }\\
  &\xline{ \sat{\nsymbpp{}}{13}{\act}{\nrestart}{{\color{white}\carry}}{\one}\sat{\nsplu{2}}{10}{\act}{{\color{white}\nrestart}}{\carry}{{\color{white}\one}}\sat{\symbB}{12}{\act}{\nrestart}{{\color{white}\carry}}{{\color{white}\one}}\sat{\symbd}{12}{\act}{\nrestart}{{\color{white}\carry}}{{\color{white}\one}}\sat{\symba}{12}{\act}{\nrestart}{{\color{white}\carry}}{{\color{white}\one}}\sat{\symbc}{12}{\act}{\nrestart}{{\color{white}\carry}}{{\color{white}\one}}\sat{\symbd}{4}{{\color{white}\act}}{\nrestart}{{\color{white}\carry}}{{\color{white}\one}}\sat{\symba}{4}{{\color{white}\act}}{\nrestart}{{\color{white}\carry}}{{\color{white}\one}}\sat{\symbc}{4}{{\color{white}\act}}{\nrestart}{{\color{white}\carry}}{{\color{white}\one}}\sat{\symbd}{12}{\act}{\nrestart}{{\color{white}\carry}}{{\color{white}\one}}\sat{\symbc}{12}{\act}{\nrestart}{{\color{white}\carry}}{{\color{white}\one}}\sat{\symbd}{4}{{\color{white}\act}}{\nrestart}{{\color{white}\carry}}{{\color{white}\one}}\sat{\symbX}{4}{{\color{white}\act}}{\nrestart}{{\color{white}\carry}}{{\color{white}\one}}\sat{\symbO}{4}{{\color{white}\act}}{\nrestart}{{\color{white}\carry}}{{\color{white}\one}}\sat{\symbIns}{4}{{\color{white}\act}}{\nrestart}{{\color{white}\carry}}{{\color{white}\one}}\sat{\symbL}{9}{\act}{{\color{white}\nrestart}}{{\color{white}\carry}}{\one}\sat{\szero}{10}{\act}{{\color{white}\nrestart}}{\carry}{{\color{white}\one}}\sat{\sone}{11}{\act}{{\color{white}\nrestart}}{\carry}{\one}\sat{\szero}{12}{\act}{\nrestart}{{\color{white}\carry}}{{\color{white}\one}}\sat{\szero}{13}{\act}{\nrestart}{{\color{white}\carry}}{\one}\sat{\sone}{14}{\act}{\nrestart}{\carry}{{\color{white}\one}}\sat{\szero}{15}{\act}{\nrestart}{\carry}{\one}\sat{\szero}{0}{{\color{white}\act}}{{\color{white}\nrestart}}{{\color{white}\carry}}{{\color{white}\one}}\sat{\sone}{1}{{\color{white}\act}}{{\color{white}\nrestart}}{{\color{white}\carry}}{\one}\sat{\sone}{2}{{\color{white}\act}}{{\color{white}\nrestart}}{\carry}{{\color{white}\one}}\sat{\sone}{3}{{\color{white}\act}}{{\color{white}\nrestart}}{\carry}{\one}\sat{\symbRns}{4}{{\color{white}\act}}{\nrestart}{{\color{white}\carry}}{{\color{white}\one}} }\\
  &\xline{ \sat{\nsmin{}}{12}{\act}{\nrestart}{{\color{white}\carry}}{{\color{white}\one}}\sat{\symbB}{13}{\act}{\nrestart}{{\color{white}\carry}}{\one}\sat{\symbd}{13}{\act}{\nrestart}{{\color{white}\carry}}{\one}\sat{\nsymbpp{}}{11}{\act}{{\color{white}\nrestart}}{\carry}{\one}\sat{\symbA}{12}{\act}{\nrestart}{{\color{white}\carry}}{{\color{white}\one}}\sat{\symbd}{12}{\act}{\nrestart}{{\color{white}\carry}}{{\color{white}\one}}\sat{\symbd}{12}{\act}{\nrestart}{{\color{white}\carry}}{{\color{white}\one}}\sat{\symba}{12}{\act}{\nrestart}{{\color{white}\carry}}{{\color{white}\one}}\sat{\symbc}{12}{\act}{\nrestart}{{\color{white}\carry}}{{\color{white}\one}}\sat{\symbd}{4}{{\color{white}\act}}{\nrestart}{{\color{white}\carry}}{{\color{white}\one}}\sat{\symbd}{4}{{\color{white}\act}}{\nrestart}{{\color{white}\carry}}{{\color{white}\one}}\sat{\symbd}{4}{{\color{white}\act}}{\nrestart}{{\color{white}\carry}}{{\color{white}\one}}\sat{\symbX}{4}{{\color{white}\act}}{\nrestart}{{\color{white}\carry}}{{\color{white}\one}}\sat{\symbO}{4}{{\color{white}\act}}{\nrestart}{{\color{white}\carry}}{{\color{white}\one}}\sat{\symbIns}{4}{{\color{white}\act}}{\nrestart}{{\color{white}\carry}}{{\color{white}\one}}\sat{\symbL}{8}{\act}{{\color{white}\nrestart}}{{\color{white}\carry}}{{\color{white}\one}}\sat{\szero}{9}{\act}{{\color{white}\nrestart}}{{\color{white}\carry}}{\one}\sat{\sone}{10}{\act}{{\color{white}\nrestart}}{\carry}{{\color{white}\one}}\sat{\szero}{11}{\act}{{\color{white}\nrestart}}{\carry}{\one}\sat{\szero}{12}{\act}{\nrestart}{{\color{white}\carry}}{{\color{white}\one}}\sat{\sone}{13}{\act}{\nrestart}{{\color{white}\carry}}{\one}\sat{\szero}{14}{\act}{\nrestart}{\carry}{{\color{white}\one}}\sat{\szero}{15}{\act}{\nrestart}{\carry}{\one}\sat{\sone}{0}{{\color{white}\act}}{{\color{white}\nrestart}}{{\color{white}\carry}}{{\color{white}\one}}\sat{\sone}{1}{{\color{white}\act}}{{\color{white}\nrestart}}{{\color{white}\carry}}{\one}\sat{\sone}{2}{{\color{white}\act}}{{\color{white}\nrestart}}{\carry}{{\color{white}\one}}\sat{\szero}{3}{{\color{white}\act}}{{\color{white}\nrestart}}{\carry}{\one}\sat{\symbRns}{4}{{\color{white}\act}}{\nrestart}{{\color{white}\carry}}{{\color{white}\one}} }\\
  &\xline{ \sat{\symbB}{12}{\act}{\nrestart}{{\color{white}\carry}}{{\color{white}\one}}\sat{\symbd}{12}{\act}{\nrestart}{{\color{white}\carry}}{{\color{white}\one}}\sat{\nsmin{}}{12}{\act}{\nrestart}{{\color{white}\carry}}{{\color{white}\one}}\sat{\symbA}{13}{\act}{\nrestart}{{\color{white}\carry}}{\one}\sat{\symbd}{13}{\act}{\nrestart}{{\color{white}\carry}}{\one}\sat{\symbd}{13}{\act}{\nrestart}{{\color{white}\carry}}{\one}\sat{\nsymbpp{}}{11}{\act}{{\color{white}\nrestart}}{\carry}{\one}\sat{\symbA}{12}{\act}{\nrestart}{{\color{white}\carry}}{{\color{white}\one}}\sat{\symbd}{12}{\act}{\nrestart}{{\color{white}\carry}}{{\color{white}\one}}\sat{\symbd}{12}{\act}{\nrestart}{{\color{white}\carry}}{{\color{white}\one}}\sat{\symbd}{12}{\act}{\nrestart}{{\color{white}\carry}}{{\color{white}\one}}\sat{\symbd}{12}{\act}{\nrestart}{{\color{white}\carry}}{{\color{white}\one}}\sat{\symbX}{12}{\act}{\nrestart}{{\color{white}\carry}}{{\color{white}\one}}\sat{\symbO}{12}{\act}{\nrestart}{{\color{white}\carry}}{{\color{white}\one}}\sat{\symbIns}{12}{\act}{\nrestart}{{\color{white}\carry}}{{\color{white}\one}}\sat{\symbL}{15}{\act}{\nrestart}{\carry}{\one}\sat{\szero}{0}{{\color{white}\act}}{{\color{white}\nrestart}}{{\color{white}\carry}}{{\color{white}\one}}\sat{\sone}{1}{{\color{white}\act}}{{\color{white}\nrestart}}{{\color{white}\carry}}{\one}\sat{\szero}{2}{{\color{white}\act}}{{\color{white}\nrestart}}{\carry}{{\color{white}\one}}\sat{\szero}{3}{{\color{white}\act}}{{\color{white}\nrestart}}{\carry}{\one}\sat{\sone}{4}{{\color{white}\act}}{\nrestart}{{\color{white}\carry}}{{\color{white}\one}}\sat{\szero}{5}{{\color{white}\act}}{\nrestart}{{\color{white}\carry}}{\one}\sat{\szero}{6}{{\color{white}\act}}{\nrestart}{\carry}{{\color{white}\one}}\sat{\sone}{7}{{\color{white}\act}}{\nrestart}{\carry}{\one}\sat{\sone}{8}{\act}{{\color{white}\nrestart}}{{\color{white}\carry}}{{\color{white}\one}}\sat{\sone}{9}{\act}{{\color{white}\nrestart}}{{\color{white}\carry}}{\one}\sat{\szero}{10}{\act}{{\color{white}\nrestart}}{\carry}{{\color{white}\one}}\sat{\szero}{11}{\act}{{\color{white}\nrestart}}{\carry}{\one}\sat{\symbRns}{12}{\act}{\nrestart}{{\color{white}\carry}}{{\color{white}\one}} }
\end{align*}

\section{Discussion}
In Section~\ref{sec:exp} we have encoded the state of a binary counter using a binary encoding.
In comparison with the binary counter obtained from the Fractran encoding, this yields an
enormous simplification concerning the number of required morphisms.
Moreover, we have illustrated a construction which allows for shifting the activity
from one letter to the next in each iteration of the morphism,
and how the letters can `communicate' computation results to the following letter.
It would be interesting to investigate whether Turing machines can be encoded in a similar way. 
The crucial difference would be that for Turing machines we need to shift the
activity left or right depending on the outcome of the current step;
the binary counter always shifts the activity to the right.
It is unclear to us whether the encoding from Section~\ref{sec:exp}
can be extended in this direction.
Compared to our Fractran encoding of Section~\ref{sec:comp},
such an encoding of Turing machines 
could lead to significantly less morphisms (but with a slightly larger alphabet).

\bibliography{main}

\newcommand{\SortNoop}[1]{}
\begin{thebibliography}{10}

\bibitem{allo:1994}
J.-P. Allouche.
\newblock {Sur la Complexit\'{e} des Suites Infinies}.
\newblock {\em Journ\'{e}es Montoises}, 1(2):133--143, 1994.

\bibitem{allo:shal:2003}
J.-P. Allouche and J.~Shallit.
\newblock {\em Automatic Sequences: Theory, Applications, Generalizations}.
\newblock Cambridge University Press, New York, 2003.

\bibitem{arsh:1937}
S.~Arshon.
\newblock {D\'{e}monstration de l'Existence de Suites Asym\'{e}triques
  Infinies}.
\newblock {\em Matematicheskii Sbornik}, 44:769--777, 1937.
\newblock In Russian. French summary: 777--779.

\bibitem{bers:1980}
J.~Berstel.
\newblock {Mots Sans Carr\'{e} et Morphismes It\'{e}r\'{e}s}.
\newblock {\em Discrete Mathematics}, 29:235--244, 1980.

\bibitem{cart:thom:2002}
O.~Carton and W.~Thomas.
\newblock {The Monadic Theory of Morphic Infinite Words and Generalizations}.
\newblock {\em Information and Computation}, 176(1):51--65, 2002.

\bibitem{cass:karh:1997}
J.~Cassaigne and J.~Karhum{\"a}ki.
\newblock {Toeplitz Words, Generalized Periodicity and Periodically Iterated
  Morphisms}.
\newblock {\em European Journal of Combinatorics}, 18(5):497--510, 1997.

\bibitem{conw:1972}
J.~H. Conway.
\newblock {Unpredictable Iterations}.
\newblock In {\em Proc. of the 1972~Number Theory Conference}, pages 49--52.
  University of Colorado, 1972.

\bibitem{conw:1987}
J.~H. Conway.
\newblock {Fractran: A Simple Universal Programming Language for Arithmetic}.
\newblock In {\em {Open Problems in Communication and Computation}}, pages
  4--26. Springer, 1987.

\bibitem{culi:harj:1984}
K.~{\SortNoop{Culik}}\u{C}ulik~II and T.~Harju.
\newblock {The $\omega$-Sequence Problem for \dol{} Systems Is Decidable}.
\newblock {\em Journal of the Association for Computing Machinery},
  31(2):282--298, 1984.

\bibitem{culi:karh:1992}
K.~{\SortNoop{Culik}}\u{C}ulik~II and J.~Karhum{\"a}ki.
\newblock {Iterative Devices Generating Infinite Words}.
\newblock In {\em Proc. 9th Ann. Symp. on Theoretical Aspects of Computer
  Science (STACS~1992)}, volume 577 of {\em Lecture Notes in Computer Science},
  pages 531--543. Springer, 1992.

\bibitem{culi:karh:lepi:1992}
K.~{\SortNoop{Culik}}\u{C}ulik~II, J.~Karhum{\"a}ki, and A.~Lepist{\"o}.
\newblock {Alternating Iteration of Morphisms and Kolakovski [\textit{sic}]
  Sequence}.
\newblock In G.~Rozenberg and A.~Salomaa, editors, {\em Lindermayer Systems,
  Impacts on Theoretical Computer Science, Computer Graphics and Developmental
  Biology}, pages 93--106. Springer, 1992.

\bibitem{ehre:lee:roze:1975}
A.~Ehrenfeucht, K.~P. Lee, and G.~Rozenberg.
\newblock {Subword Complexity of Various Classes of Deterministic Languages
  without Interaction}.
\newblock {\em Theoretical Computer Science}, 1:59--75, 1975.

\bibitem{endr:grab:hend:2008}
J.~Endrullis, C.~Grabmayer, and D.~Hendriks.
\newblock {Data-Oblivious Stream Productivity}.
\newblock In {\em Proc. 15th Int. Conf. on Logic for Programming, Artifical
  Intelligence and Reasoning (LPAR~2008)}, volume 5330 of {\em Lecture Notes in
  Computer Science}, pages 79--96. Springer, 2008.

\bibitem{endr:grab:hend:2009}
J.~Endrullis, C.~Grabmayer, and D.~Hendriks.
\newblock {Complexity of Fractran and Productivity}.
\newblock In {\em Proc. 22nd Int. Conf. on Automated Deduction (CADE~2009)},
  volume 5663 of {\em Lecture Notes in Computer Science}, pages 371--387.
  Springer, 2009.

\bibitem{endr:grab:hend:isih:klop:2010}
J.~Endrullis, C.~Grabmayer, D.~Hendriks, A.~Isihara, and J.W. Klop.
\newblock {Productivity of Stream Definitions}.
\newblock {\em Theoretical Computer Science}, 411:765--782, 2010.

\bibitem{endr:hend:2011}
J.~Endrullis and D.~Hendriks.
\newblock {Lazy Productivity via Termination}.
\newblock {\em Theoretical Computer Science}, 412(28):3203--3225, 2011.

\bibitem{endr:hend:klop:2011}
J.~Endrullis, D.~Hendriks, and J.W. Klop.
\newblock {Degrees of Streams}.
\newblock {\em Journal of Integers}, 11B(A6):1--40, 2011.
\newblock Proceedings of the Leiden Numeration Conference 2010.

\bibitem{fere:1999}
S.~Ferenczi.
\newblock {Complexity of Sequences and Dynamical Systems}.
\newblock {\em Discrete Mathematics}, 206(1-3):145--154, 1999.

\bibitem{grab:endr:hend:klop:moss:2012}
C.~Grabmayer, J.~Endrullis, D.~Hendriks, J.W. Klop, and L.S. Moss.
\newblock {Automatic Sequences and Zip-Specifications}.
\newblock In {\em Proc.\ Symp.\ on Logic in Computer Science (LICS~2012)}. IEEE
  Computer Society, 2012.

\bibitem{jaco:kean:1969}
K.~Jacobs and M.~Keane.
\newblock {$0$-$1$-Sequences of Toeplitz Type}.
\newblock {\em Zeitschrift f\"{u}r Wahrscheinlichkeitstheorie und Verwandte
  Gebiete}, 13(2):123--131, 1969.

\bibitem{kola:1965}
W.~Kolakovski.
\newblock {Self Generating Runs}.
\newblock {\em The American Mathematical Monthly}, 72, 1965.
\newblock Problem 5304.

\bibitem{lepi:1993}
A.~Lepist{\"o}.
\newblock {On the Power of Periodic Iteration of Morphisms}.
\newblock In {\em Proc. 20th Int. Coll. on Automata, Languages and Programming
  (ICALP~1993)}, volume 700 of {\em Lecture Notes in Computer Science}, pages
  496--506. Springer, 1993.

\bibitem{much:prit:seme:2009}
A.~A. Muchnik, Y.~L. Pritykin, and A.~L. Semenov.
\newblock {Sequences Close to Periodic}.
\newblock {\em Russian Mathematical Surveys}, 64(5):805--871, 2009.

\bibitem{seeb:2003}
P.~S{\'e}{\'e}bold.
\newblock {On Some Generalizations of the Thue-Morse Morphism}.
\newblock {\em Theoretical Computer Science}, 292(1):283--298, 2003.

\bibitem{sijt:1989}
B.~A. Sijtsma.
\newblock {On the Productivity of Recursive List Definitions}.
\newblock {\em ACM Transactions on Programming Languages and Systems},
  11(4):633--649, 1989.

\end{thebibliography}

\end{document}